\documentclass[11pt,palatino]{article}
\usepackage{amsmath,amsfonts,amssymb,amsthm}
\usepackage{fullpage}
\usepackage[ruled]{algorithm2e}

\usepackage{subfigure}
\usepackage{epstopdf}
\usepackage{graphicx}
\usepackage{cite}
\usepackage{color}
\usepackage[sans]{dsfont}
\usepackage{mathtools}
\usepackage{hyperref}
\usepackage{cleveref}
\usepackage{tcolorbox}
\usepackage{multirow}
\usepackage{xspace}
\tcbuselibrary{skins,breakable}
\tcbset{enhanced jigsaw}

\newcommand{\myparskip}{3pt}
\parskip \myparskip
\renewcommand{\tau}{{\mathcal T}}
\newtheorem{theorem}{Theorem}
\newtheorem{heuristic algorithm}{Heuristic Algorithm}
\newtheorem{lemma}{Lemma}[section]

\newtheorem{claim}{Claim}[section]

\newtheorem{observation}{Observation}[section]

\makeatletter\@addtoreset{section}{part}\makeatother%

\begin{document}
\bibliographystyle{alpha}

\newcommand{\eps}{\varepsilon}
\newcommand{\dist}{\operatorname{dist}}
\newcommand{\gset}{{\mathcal{G}}}
\newcommand{\hset}{{\mathcal{H}}}
\newcommand{\sset}{{\mathcal{S}}}
\newcommand{\mset}{{\mathcal{M}}}
\newcommand{\rset}{{\mathcal{R}}}
\newcommand{\pset}{{\mathcal{P}}}
\newcommand{\cutpath}{\textsc{Split}}
\newcommand{\gluepath}{\textsc{Glue}}
\newcommand{\basesizebound}{{\frac{c^*\log^2 k}{\eps^{20}}}}
\newcommand{\algmerge}{\textsc{Combine}}

\newcommand{\colnote}[3]{\textcolor{#1}{$\ll$\textsf{#2}$\gg$\marginpar{\tiny\bf \textcolor{#1}{#3}}}}
\newcommand{\rnote}[1]{\colnote{blue}{#1--Robi}{RK}}

\newcommand{\algline}{
	\rule{0.5\linewidth}{.1pt}\hspace{\fill}%
	\par\nointerlineskip \vspace{.1pt}
}
\newenvironment{tbox}{\begin{tcolorbox}[
		enlarge top by=5pt,
		enlarge bottom by=5pt,
		breakable,
		boxsep=0pt,
		left=4pt,
		right=4pt,
		top=10pt,
		boxrule=1pt,toprule=1pt,
		colback=white,
		arc=-1pt,
		]
	}
	{\end{tcolorbox}}


\newenvironment{proofof}[1]{\noindent{\bf Proof of #1.}}
{\hspace*{\fill}\stopproof}

\newenvironment{properties}[2][0]
{\renewcommand{\theenumi}{#2\arabic{enumi}}
	\begin{enumerate} \setcounter{enumi}{#1}}{\end{enumerate}\renewcommand{\theenumi}{\arabic{enumi}}}

\newif\ifnocomments


\ifnocomments

\newcommand{\znote}[1]{}

\else
\newcommand{\znote}[1]{\textcolor{red}{\sc{[ZT: #1]}}}

\fi

\ifnocomments

\newcommand{\snote}[1]{}

\else
\newcommand{\snote}[1]{\textcolor{red}{\sc{[SK: #1]}}}

\fi


\newcommand{\tG}{\textbf{G}}
\newcommand{\tH}{\textbf{H}}
\newcommand{\tE}{\textbf{E}'}
\newcommand{\tC}{\textbf{C}}
\newcommand{\tphi}{\bm{\phi}}
\newcommand{\tpsi}{\bm{\psi}}
\newcommand{\tSigma}{\bm{\Sigma}}
\newcommand{\tB}{\tilde B}
\newcommand{\dout}{D_{\mbox{\tiny{out}}}}
\newcommand{\notF}{\overline{F}}
\newcommand{\St}{Steiner Tree\xspace}
\newcommand{\ST}{Steiner Tree\xspace}

\renewcommand{\P}{\mbox{\sf P}}
\newcommand{\NP}{\mbox{\sf NP}}
\newcommand{\PCP}{\mbox{\sf PCP}}
\newcommand{\ZPP}{\mbox{\sf ZPP}}
\newcommand{\DTIME}{\mbox{\sf DTIME}}
\newcommand{\opt}{\mathsf{OPT}}
\newcommand{\optcro}{\mathsf{OPT}_{\mathsf{cr}}}
\newcommand{\optcrors}{\mathsf{OPT}_{\mathsf{cnwrs}}}
\newcommand{\set}[1]{\left\{ #1 \right\}}
\newcommand{\sse}{\subseteq}
\newcommand{\B}{{\mathcal{B}}}
\newcommand{\tset}{{\mathcal T}}
\newcommand{\uset}{{\mathcal U}}
\newcommand{\iset}{{\mathcal{I}}}
\newcommand{\nset}{{\mathcal{N}}}
\newcommand{\dset}{{\mathcal{D}}}
\newcommand{\tpset}{\tilde{\mathcal{P}}}
\newcommand{\qset}{{\mathcal{Q}}}
\newcommand{\tqset}{\tilde{\mathcal{Q}}}
\newcommand{\lset}{{\mathcal{L}}}
\newcommand{\bset}{{\mathcal{B}}}
\newcommand{\tbset}{\tilde{\mathcal{B}}}
\newcommand{\aset}{{\mathcal{A}}}
\newcommand{\cset}{{\mathcal{C}}}
\newcommand{\fset}{{\mathcal{F}}}
\newcommand{\jset}{{\mathcal{J}}}
\newcommand{\xset}{{\mathcal{X}}}
\newcommand{\wset}{{\mathcal{W}}}
\newcommand{\oset}{{\mathcal{O}}}
\newcommand{\yset}{{\mathcal{Y}}}
\newcommand{\I}{{\mathcal I}}
\newcommand{\zset}{{\mathcal{Z}}}
\newcommand{\notu}{\overline U}
\newcommand{\vol}{\operatorname{vol}}
\newcommand{\nots}{\overline S}
\newcommand{\eint}{E^{\tiny\mbox{int}}}
\newcommand{\event}{{\cal{E}}}
\newcommand{\floor}[1]{\ensuremath{\left\lfloor#1\right\rfloor}}
\newcommand{\ceil}[1]{\ensuremath{\left\lceil#1\right\rceil}}

\newcommand{\marcon}{{\mathsf{MC}}}
\newcommand{\cov}{{\mathsf{cov}}}
\newcommand{\mst}{{\mathsf{MST}}}
\newcommand{\card}[1]{|#1|}
\newcommand{\coi}{{\mathsf{COI}}}
\newcommand{\setcover}{{\textnormal{\sf SC}}}
\newcommand{\algsetcover}{{\textnormal{\sf AlgSetCover}}}
\newcommand{\stcost}{{\mathsf{ST}}}

\newcommand{\cover}{\textsf{cover}}
\newcommand{\bfs}{\textnormal{\textsf{BFS}}}
\newcommand{\pbfs}{\textnormal{\textsf{BFS}}}
\newcommand{\lv}{\textsf{lv}}
\newcommand{\tsp}{\mathsf{TSP}}
\newcommand{\gtsp}{\textsf{GTSP}}
\newcommand{\ebt}{\tset}
\newcommand{\eb}{\textsf{EB}}
\newcommand{\optmst}{\textsf{MST}}
\newcommand{\defi}{\textsf{def}}
\newcommand{\ord}{\textsf{ord}}
\newcommand{\rc}{\textnormal{\textsf{rc}}}
\newcommand{\cost}{\textnormal{\textsf{cost}}}
\newcommand{\bw}{\textsf{bw}}
\newcommand{\local}{\textsf{Local}}
\newcommand{\pseudo}{\textsf{Pseudo-IP}}
\newcommand{\vin}{v^{\textnormal{\textsf{in}}}}
\newcommand{\vout}{v^{\textnormal{\textsf{out}}}}
\newcommand{\diam}{\textsf{diam}}
\newcommand{\expect}{\mathbb{E}}
\newcommand{\proover}{\pi_{\textsf{Overwrite}}}
\newcommand{\promst}{\pi_{\textsf{MST}}}
\newcommand{\protsp}{\pi_{\textsf{TSP}}}
\newcommand{\mstest}{\textsf{MST}_{\textsf{apx}}}
\newcommand{\tspest}{\textsf{TSP}_{\textsf{apx}}}
\newcommand{\proind}{\pi_{\textsf{Index}}}
\newcommand{\ind}{\textsf{Index}}
\newcommand{\distIND}{\mathcal{D}_{\textsf{Index}}}
\newcommand{\distMST}{\mathcal{D}_{\textsf{MST}}}
\newcommand{\ic}{\textnormal{\textsf{IC}}}
\newcommand{\cc}{\textnormal{\textsf{CC}}}
\newcommand{\tvd}[2]{\ensuremath{\Delta_{\textnormal{\texttt{TV}}}(#1,#2)}}
\newcommand{\dkl}[2]{\ensuremath{D_{\textnormal{\textsf{KL}}}(#1 \| #2)}}
\newcommand{\pr}{\textsf{Pr}}

\newcommand{\hel}{h}
\newcommand{\II}{I}
\newcommand{\HH}{H}

\newcommand{\RV}[1]{\mathbf{#1}}
\newcommand{\prot}{\ensuremath{\Pi}}
\newcommand{\Prot}{\ensuremath{\Pi}}
\newcommand{\findmiss}{\sf{FindBit}}
\newcommand{\overwrite}{\sf{Overwrite}}
\newcommand{\distfind}{\mathcal{D}_{\textsf{FindBit}}}
\newcommand{\distover}{\mathcal{D}_{\textsf{Overwrite}}}
\newcommand{\temp}{\textsf{temp}}
\newcommand{\IA}{\textsf{IA}}
\newcommand{\IB}{\textsf{IB}}

\newcommand{\row}{\textsf{Row}}
\newcommand{\col}{\textsf{Col}}
\newcommand{\alg}{\ensuremath{\mathsf{Alg}}\xspace}

\newcommand{\opttsp}{\textnormal{\textsf{TSP}}}

\newcommand{\sep}{\sf{sep}}
\newcommand{\core}{\sf{core}}
\newcommand{\scut}{\sf{Shortcut}}
\newcommand{\adv}{\mathsf{adv}}
\newcommand{\lig}{\sf{light}}
\newcommand{\maxmat}{\mathsf{MM}}
\newcommand{\midd}{\mathsf{mid}}
\newcommand{\bottom}{\mathsf{bot}}
\newcommand{\topp}{\mathsf{top}}
\newcommand{\snfl}{tree\xspace}
\newcommand{\snfls}{trees\xspace}
\newcommand{\inn}{\sf in}
\newcommand{\wD}{w_{\downarrow}}
\newcommand{\wU}{w_{\uparrow}}
\newcommand{\walkcost}{\mathsf{MWC}}

\begin{titlepage}
	
	\title{Query Complexity of the Metric Steiner Tree Problem}

	\author{Yu Chen\thanks{EPFL, Lausanne, Switzerland. Email: {\tt yu.chen@epfl.ch}. Supported by ERC Starting Grant 759471. Work done while the author was a graduate student at University of Pennsylvania.}   \and Sanjeev Khanna\thanks{University of Pennsylvania, Philadelphia, PA, USA. Email: {\tt  sanjeev@cis.upenn.edu}. Supported in part by NSF awards CCF-1763514, CCF-1934876, and CCF-2008305.} \and Zihan Tan\thanks{Rutgers University, NJ, USA. Email: {\tt zihantan1993@gmail.com}. Supported by a grant to DIMACS from the Simons Foundation (820931). Work done while the author was a graduate student at University of Chicago.}} 

	\maketitle

	\thispagestyle{empty}
	\begin{abstract}
		
In the metric Steiner Tree problem, we are given an $n \times n$ metric $w$ on a set $V$ of vertices along with a set $T \subseteq V$ of $k$ terminals, and the goal is to find a tree of minimum cost that contains all terminals in $T$. 
This is a well-known NP-hard problem and much of the previous work has focused on understanding its polynomial-time approximability. In this work, we initiate a study of the query complexity of the metric Steiner Tree problem. Specifically, if we desire an $\alpha$-approximate estimate of the metric Steiner Tree cost, how many entries need to be queried in the metric $w$? For the related minimum spanning tree (MST) problem, this question is well-understood: for any fixed $\varepsilon > 0$, one can estimate the MST cost to within a $(1+\varepsilon)$-factor using only $\tilde{O}(n)$ queries, and this is known to be essentially tight. Can one obtain a similar result for Steiner Tree cost? Note that a $(2 + \varepsilon)$-approximate estimate of Steiner Tree cost can be obtained with $\tilde{O}(k)$ queries by simply applying the MST cost estimation algorithm on the metric induced by the terminals. 

Our first result shows that the Steiner Tree problem behaves in a fundamentally different manner from MST: any (randomized) algorithm that estimates the Steiner Tree cost to within a $(5/3 - \varepsilon)$-factor requires $\Omega(n^2)$ queries, even if $k$ is a constant. 
This lower bound is in sharp contrast to an upper bound of $O(nk)$ queries for computing a $(5/3)$-approximate Steiner Tree, which follows from previous work by Du and Zelikovsky.

Our second main result, and the main technical contribution of this work, is a {\em sublinear} query algorithm for estimating the Steiner Tree cost to within a {\em strictly better-than-$2$} factor. We give an algorithm that achieves this goal, with a query complexity of $\tilde{O}(n^{12/7} + n^{6/7}\cdot k)$; since $k\le n$, the algorithm performs at most $\tilde{O}(n^{13/7})=o(n^2)$ queries in the worst-case. Our estimation algorithm reduces this task to that of designing a sublinear query algorithm for a suitable set cover problem. We complement this result by showing an $\tilde{\Omega}(n + k^{6/5})$ query lower bound for any algorithm that estimates Steiner Tree cost to a strictly better than $2$ factor. Thus $\tilde{\Omega}(n^{6/5})$ queries are needed to just beat $2$-approximation when $k = \Omega(n)$; a sharp contrast to MST cost estimation where a $(1+o(1))$-approximate estimate of cost is achievable with only $\tilde{O}(n)$ queries.
	
	\end{abstract}
\end{titlepage}

\section{Introduction}

In the Steiner Tree problem, we are given a weighted (undirected) graph $G$ and a subset $T$ of vertices in $G$ called \emph{terminals}, and the goal is to compute a minimum weight connected subgraph of $G$ (a Steiner Tree)  that spans all terminals in $T$.
This is one of the most fundamental NP-hard problems~\cite{garey1979computers}, and has been studied extensively over the past several decades from the perspective of approximation algorithms \cite{gilbert1968steiner,zelikovsky199311,zelikovsky1996better,karpinski1997new,robins2005tighter,byrka2010improved,goemans2012matroids} (see also \cite{hauptmann2013compendium} for a compendium of its variants).
The current best known approximation ratio is
$\ln 4+\eps<1.39$ achieved by \cite{byrka2010improved} (see also \cite{goemans2012matroids,traub2022local}), and it has been shown \cite{chlebik2008steiner} that approximating to within a factor better than $96/95$ is NP-hard.

In this paper, we study the query complexity of the Steiner Tree problem. In particular, we consider an equivalent variant called the \emph{metric Steiner Tree} problem, where the input consists of a \emph{metric} $w$ on a set $V$ of $n$ points (equivalently, a weighted complete graph on $V$) and a subset $T\subseteq V$ of $k$ points called \emph{terminals}. We are allowed to perform weight queries between vertices in $V$\footnote{This is also known as the Dense Graph Model \cite{goldreich1998property}. In the other model, the Bounded-Degree Graph Model \cite{goldreich1997property} (where the max-degree of the input graph is bounded by $d$), it is easy to show that estimating the minimum Steiner Tree cost, or even the minimum spanning tree cost, in an $n$-vertex graph within any non-trivial factor requires $\Omega(nd)$ queries.}, and the goal is to design an algorithm for either computing a Steiner Tree with minimum cost or estimating the cost of an optimal Steiner Tree, using as few queries as possible. 

It is well-known that a minimum weight spanning tree of the metric induced by the terminals gives a $2$-approximate Steiner Tree \cite{gilbert1968steiner}. Moreover, the metric Minimum Spanning Tree (MST) cost can be estimated to within factor $(1+\eps)$ be performing $\tilde O(n/\eps^{O(1)})$ queries \cite{chazelle2005approximating,czumaj2009estimating}. Therefore, the minimum metric Steiner Tree cost can be estimated within factor $(2+\eps)$ with only $\tilde O(k/\eps^{O(1)})$ queries.
However, the query complexity of obtaining a better-than-$2$ estimation of the minimum metric Steiner Tree cost remains wide open, and so is the  query complexity of computing such a Steiner Tree. We note that this is the interesting regime of the metric Steiner Tree problem: approximating or estimating the cost to a factor better than $2$ crucially requires some knowledge of the metric incident on Steiner nodes.

\subsection{Our Results}

In this paper, we provide a comprehensive understanding on the trade-off between approximation ratio and query complexity of the metric Steiner Tree problem.

Our first result establishes a separation between the behavior the Steiner Tree problem and the MST problem. Specifically, we show that for any $\eps > 0$, any randomized algorithm to estimate Steiner Tree cost to within a $(5/3 - \eps)$-factor requires $\Omega(n^2)$ queries even if $k$ is a constant. 
Together with an upper bound of $O(nk)$ queries for computing a $(5/3)$-approximate Steiner Tree (which follows from \cite{du1995component} and \cite{zelikovsky199311}\footnote{See \Cref{apd: 5/3 upper} for an explanation.}), this result shows a phase transition in the query complexity at $(5/3)$-approximation. 
This is in contrast to the MST cost estimation problem where for any $\eps > 0$, $\tilde O(n)$ queries suffice to estimate MST cost to within a factor of $(1+\eps)$.

\begin{theorem}
\label{thm: lower-main}
For any constant $0<\eps< 2/3$, any randomized algorithm that with high probability estimates the metric Steiner Tree cost to within a factor of $(5/3-\eps)$ performs $\Omega(n^2/4^{(1/\eps)})$ queries in the worst case, even when $k$ is a constant.
\end{theorem}

Our proof of this result is based on constructing a pair of distributions on Steiner Tree instances whose costs differ by a $(5/3-\eps)$ factor, and yet whose metrics differ in $O_{\eps}(1)$ entries, leading to an $\Omega(n^2)$ lower bound for any fixed $\eps > 0$.

We complement the above result by showing that even if we weaken the goal to simply computing a slightly better-than-$2$ approximate Steiner Tree, the query complexity remains $\Omega(nk)$.

\begin{theorem}
\label{thm: beat-2-lower-computing}
For any constant $0<\eps< 1/3$, any randomized algorithm that outputs a $(2-\eps)$-approximate Steiner Tree performs at least $\Omega(nk)$ queries in the worst case.
\end{theorem}

Our second set of results is concerned with understanding the query complexity of obtaining a strictly better-than-$2$ estimate of the Steiner Tree cost. The main technical contribution of this paper is a sublinear-query algorithm that obtains a strictly better-than-$2$ estimate of the cost, by performing $\tilde O(n^{12/7}+n^{6/7}\cdot k)$ queries (as $k\le n$, the query complexity is $\tilde O(n^{13/7})=o(n^2)$).

\begin{theorem}
\label{thm: beat-2-main}
There is an efficient randomized algorithm that with high probability estimates the metric Steiner Tree cost to within a factor of $(2-\eps_0)$ for some universal constant $\eps_0>0$, by performing $\tilde O(n^{12/7}+n^{6/7}\cdot k)$ queries.
\end{theorem}

At a high-level, the proof of the above theorem starts with a minimum spanning tree $\tau$ of the graph induced by terminals. Even on simple metrics such as a metric where all weights are $1$ or $2$, the cost of such a tree can be up to a factor $2$ away from the optimal Steiner Tree cost. But in this case, the optimal tree necessarily improves upon $\tau$ by using Steiner nodes to efficiently connect together many terminals. The first challenge then becomes if such opportunities can be identified only by local exploration of the metric. Our main insight is that this task can be cast as a suitable {\em set cover} problem where the objective is to estimate the universe size minus the optimal set cover size. We then design a sublinear query algorithm for estimating the value of this set cover objective, and use it to determine whether or not the optimal Steiner tree cost is close to the cost of $\tau$, or bounded away from it. The second challenge in obtaining this result is that explicit computation of the MST $\tau$ in the graph induced by the terminals requires $O(k^2)$ queries which rules out a sublinear query complexity when $k = \Omega(n)$. To get around this, we design an outer algorithm that efficiently simulates access to $\tau$ without ever explicitly computing it. The composed algorithm, achieves a strictly better-than-$2$ estimate of the Steiner Tree cost in $\tilde O(n^{12/7}+n^{6/7}\cdot k)$ queries.

The result above raises a natural question: can the task of obtaining a strictly better-than-$2$ estimate of the Steiner Tree cost be achieved with only $\tilde O(n)$ queries? Our next result rules out this possibility at least when $k$ is sufficiently large. We show that any algorithm that estimates the Steiner Tree cost to a factor strictly better than $2$, necessarily requires $\tilde \Omega(n+k^{6/5})$ queries.

\begin{theorem}
\label{thm: beat-2-lower-main}
For any constant $0<\eps< 1/3$, any randomized algorithm that with high probability estimates the metric Steiner Tree cost to within a factor of $(2-\eps)$ performs at least $\tilde \Omega(n+k^{6/5})$ queries in the worst case.
\end{theorem}

Our third and final set of results is concerned with understanding the query complexity of computing an $\alpha$-approximate Steiner Tree for any $\alpha\ge 2$. 
We show that $\tilde\Theta(k^{2}/\alpha)$ queries are both sufficient and necessary for this task.

\begin{theorem}
\label{thm: >2-main}
Let $\alpha\ge 2$ be any constant. Then 
\begin{itemize}
\item there exists an efficient randomized algorithm that with high probability computes an $\alpha$-approximate Steiner Tree, by performing $\tilde O(k^{2}/\alpha)$ queries; and
\item any randomized algorithm that outputs an $\alpha$-approximate Steiner Tree performs at least $\Omega(k^{2}/\alpha)$ queries in the worst case.
\end{itemize} 
\end{theorem}

Our results on the tradeoff between query complexity and approximation quality are summarized in \Cref{fig: results}. Together, they illustrate several interesting phase transitions in the query complexity of approximating metric Steiner Tree. The query complexity remains $\Theta(n^2)$ up to an approximation factor of $5/3$ even if $k$ is a constant and the goal is to only estimate the Steiner Tree cost. Then at $(5/3)$-approximation, it drops to $\Theta(nk)$, and it is possible to also find a $(5/3)$-approximate Steiner Tree with $\Theta(nk)$ queries. Next if the goal is to find an $\alpha$-approximate Steiner Tree, the query complexity remains $\Theta(nk)$ even as $\alpha$ approaches $2$. At this point, another phase transition occurs: for any $\alpha \ge 2$, the query complexity of finding an $\alpha$-approximate Steiner Tree becomes $\tilde{\Theta}(k^2/\alpha)$. For $\alpha$-approximate estimation of the cost for any $\alpha < 2$, we show that $\tilde{\Omega}(n^{6/5})$ queries are necessary when $k = \Omega(n)$. On the other hand, we give a $(2 - \eps_0)$-estimation algorithm that uses only  $\tilde O (n^{12/7}+n^{6/7}\cdot k)$ queries, for some universal $\eps_0>0$.

\begin{figure}[h]
	\centering
	\includegraphics[scale=0.12]{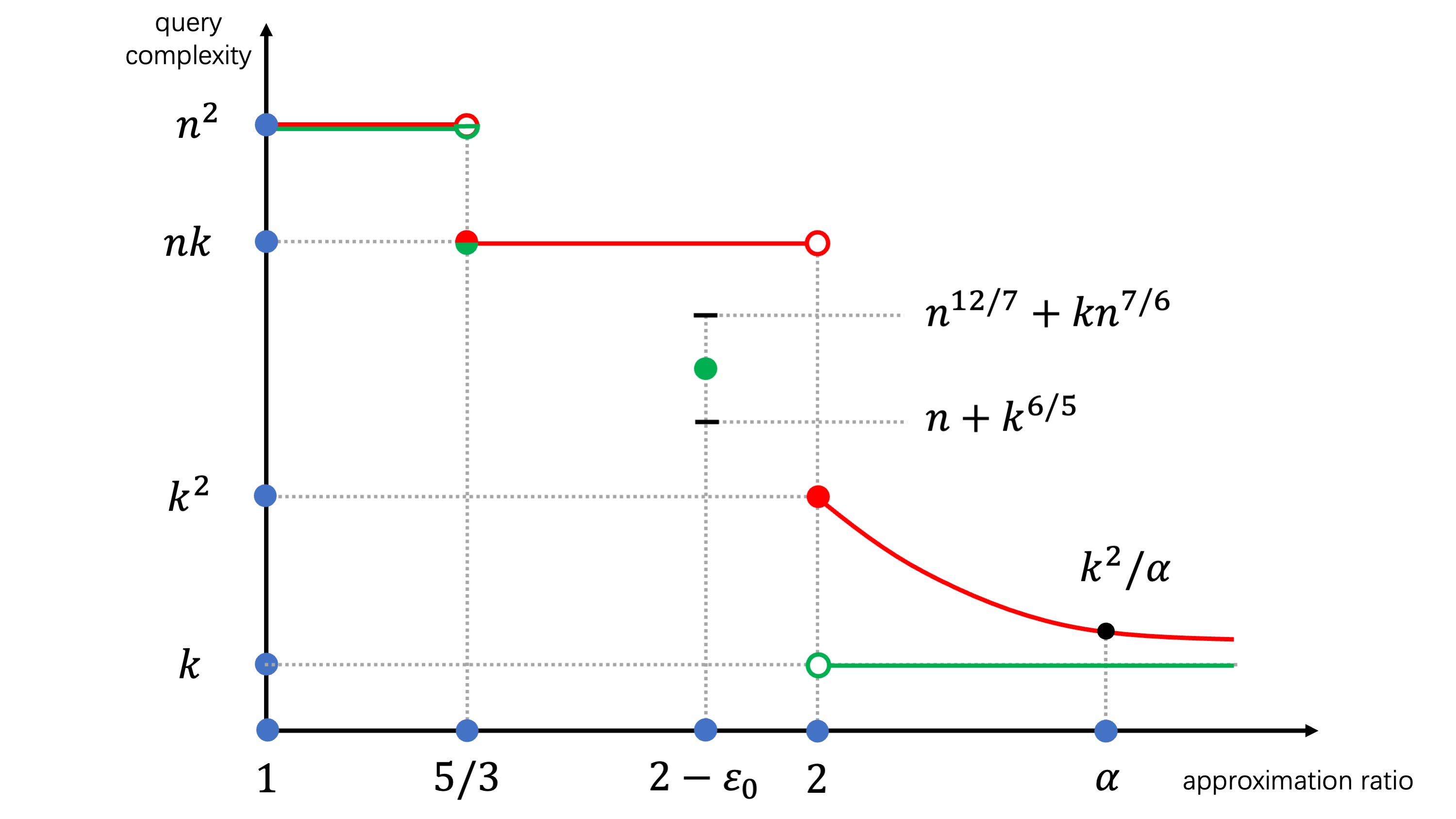}
	\caption{An illustration of the trade-off between query complexity and approximation ratio for the metric Steiner Tree problem. The red curve shows the complexity of computing a Steiner Tree, while the green curve shows the complexity of estimating the minimum metric Steiner Tree cost. 
	The upper bound at $(5/3)$-approximation follows from \cite{du1995component} and \cite{zelikovsky199311}; the bottom green curve (showing $\Theta(k)$ query complexity for $\alpha$-estimating the cost where $\alpha>2$) is due to \cite{czumaj2009estimating}, and all other curves are results of this paper. All terms are inside a $\tilde O(\cdot)$ symbol.} \label{fig: results}
\end{figure}

\paragraph{Organization.} We start with some preliminaries in \Cref{sec: prelim}. 
We then provide the proof of $(5/3-\eps)$-approximation lower bound (\Cref{thm: lower-main}) in \Cref{sec: 5/3-lower}.
We present the $(2-\eps_0)$-approximation algorithm for \Cref{thm: beat-2-main} in \Cref{sec: beat-2-upper}.
The proofs of the lower bounds in \Cref{thm: beat-2-lower-computing}, \Cref{thm: beat-2-lower-main} and \Cref{thm: >2-main} are provided in \Cref{sec: proof of beat-2-lower-computing}, \Cref{sec: beat-2-lower} and \Cref{sec: >2-main}, respectively.

\section{Preliminaries}
\label{sec: prelim}

Let $G$ be a graph and let $S$ be a subset of its vertices. We denote by $G[S]$ the subgraph of $G$ induced by $S$.
For two (not necessarily disjoint) subsets $A,B$ of vertices of $G$, we denote by $E_G(A,B)$ the set of all edges with one endpoint in $A$ and the other endpoint in $B$, and we denote by $E_G(A)$ the set of all edges with both endpoints in $A$. 
For a vertex $v\in V(G)$, we denote by $\deg_G(v)$ the degree of $v$ in $G$.
We sometimes omit the subscript $G$ in the above notations if it is clear from the context.

Let $G=(V,E,w)$ be a weighted graph where $w: E(G)\to \mathbb{R}^+$ is a weight function on edges in $G$. For a subgraph $H\subseteq G$, we define $w(H)=\sum_{e\in E(H)}w(e)$. For a pair $u,u'$ of vertices in $G$, we denote by $\dist_G(u,u')$ the shortest-path distance in $G$ between $u$ and $u'$.
In this paper, we will often consider the case where $G$ is a complete graph and $w$ satisfies the triangle inequality. That is, for all $u,u',u''\in V$, $w(u,u')\le w(u,u'')+w(u',u'')$. In this case, $w$ can be also viewed as a metric on the set $V$ of points.
For a vertex $v$ and a set $U\subseteq V$, we denote $w(v,U)=\min\set{w(v,u)\mid u\in U}$. For a pair $U,U'\subseteq V$ of sets, we denote $w(U,U')=\min\set{w(u,u')\mid u\in U,u'\in U'}$.

Let $\tau$ be a tree rooted at a vertex $r\in V(\tau)$. Let $u$ be a vertex in $\tau$. The \emph{height} of $u$ in $\tau$ is defined to be the minimum hop-distance between $u$ and any leaf in the subtree of $\tau$ rooted at $u$. For example, the height of a leaf is $0$, and the height of a parent of a leaf is $1$, etc.
 
For a weighted graph $G=(V,E,w)$ and a subset $T$ of vertices in $G$, we denote by $(G,T)$ the instance of the Steiner Tree problem where $G$ is the graph and $T$ is the set of vertices to be connected. We denote the optimal cost of a solution to this instance by $\stcost(G,T)$. When $G$ is a complete graph and $w$ is a metric on $V$, an instance of the Steiner Tree problem is also denoted by $(V,T,w)$, and the optimal cost of a solution is also denoted by $\stcost(V,T,w)$. Vertices in $T$ are called \emph{terminals}, and vertices in $V\setminus T$ are called \emph{Steiner vertices}.

Throughout the paper, we will refer to the algorithms that compute an $\alpha$-approximate Steiner Tree of cost bounded by $\alpha$-\emph{approximations}, and will refer to the algorithms that estimate the metric Steiner Tree cost to within factor $\alpha$ by $\alpha$-\emph{estimations}.

We use the following standard version of Chernoff Bound (see. e.g., \cite{dubhashi2009concentration}).

\begin{lemma}[Chernoff Bound]
	\label{lem: Chernoff}
	Let $X_1,\ldots,X_n$ be independent randon variables taking values in $\set{0,1}$. Let $X=\sum_{1\le i\le n}X_i$, and let $\mu=\mathbb{E}[X]$. Then for any $t>2e\mu$,
	\[\Pr\Big[X>t\Big]\le 2^{-t}.\]
	Additionally, for any $0\le \delta \le 1$,
	\[\Pr\Big[X<(1-\delta)\cdot\mu\Big]\le e^{-\frac{\delta^2\cdot\mu}{2}}.\]
\end{lemma}

\section{An $\Omega(n^2)$ Lower Bound for $(5/3-\eps)$-Estimation}
\label{sec: 5/3-lower}

\newcommand{\wy}{w_{\sf Y}}
\newcommand{\wn}{w_{\sf N}}
\newcommand{\iy}{I_{\sf Y}}
\newcommand{\ino}{I_{\sf N}}

In this section we provide the proof of \Cref{thm: lower-main}. Specifically, for any constant $0<\eps <2/3$, we will construct a pair $I_{\sf Y}=(V,T,w_{\sf Y}), I_{\sf N}=(V,T,w_{\sf N})$ of instances of the metric Steiner Tree problem where $|T|=O(2^{(1/\eps)})$, such that the metric \ST costs of instance $I_{\sf Y}$ and instance $I_{\sf N}$ differ by factor $(5/3-\eps)$, while the metrics $\wy$ and $\wn$ differ at only $O_{\eps}(1)$ places. 
We will then construct two distributions $\dset_{\sf Y},\dset_{\sf N}$ of instances by randomly naming the vertices in the instances $I_{\sf Y},I_{\sf N}$ respectively; and show that any algorithm that with probability at least $0.51$ distinguishes between instances $\dset_{\sf Y}$ and $\dset_{\sf N}$ (and in particular between the distributions on metrics $\wy$ and $\wn$) has to perform at least $\Omega(n^2/4^{(1/\eps)})$ queries.
We remark that it is easy to construct (by adding dummy terminals to instances $I_{\sf Y}$ and $I_{\sf N}$), for every $k\ge \Omega(2^{1/\eps})$, such a pair of instances $I^{(k)}_{\sf Y}$ and $I^{(k)}_{\sf N}$ with $k$ terminals each. Our constructions are similar to the examples used in \cite{borchers1997thek} to determine the worst-case $k$-Steiner ratios.

We start by giving a high-level overview of the construction of instances $(V,T,\wy)$ and $(V,T,\wn)$.
We would like to ensure that (i) metrics $\wy$ and $\wn$ differ only in $O_{\eps}(1)$ places; and (ii) $\stcost(V,T,\wy)$ and $\stcost(V,T,\wn)$ differ by a factor of $(5/3)$ roughly. In order to achieve property (i), for every pair $(u,u')$ of vertices in $V$ such that at least one of $u,u'$ is a terminal, $\wy(u,u')=\wn(u,u')$ must hold, since otherwise simply querying all terminal-involved distances will distinguish between metrics $\wy$ and $\wn$ by performing $O(nk)$ queries. 
However, once we are given that $\wy$ and $\wn$ are identical on all terminal-involved pairs, the previous results in \cite{zelikovsky199311} and \cite{du1995component} imply that the values $\stcost(V,T,\wy)$ and $\stcost(V,T,\wn)$ differ by a factor of at most $(5/3)$. Therefore, we have to design metrics $\wy$ and $\wn$ such that the analysis from \cite{zelikovsky199311} and \cite{du1995component} is nearly tight. It turns out that the tree $\tau_Y$ has to be quite balanced and symmetric.

We now describe the construction of instances $(V,T,\wy)$ and $(V,T,\wn)$ in detail. 
Let $d=\ceil{1/\eps}$. We first define an auxiliary weighted tree $\rho$ as follows. The tree $\rho$ is a complete binary tree of depth $d$, so $|V(\rho)|=2^{d+1}-1$. Let $r$ be the root of $\rho$. For each node $v\in V(\rho)$, we say that $v$ is \emph{at level $i$} (or the level of $v$ is $i$), iff the unique path in $\rho$ that connects $v$ to $r$ contains $i$ edges. Clearly, all leaves of $\rho$ are at level $d$. For an edge $(u,u')\in E(\rho)$ where $u'$ is the parent of $u$, we say that $(u,u')$ is a \emph{level-$i$ edge} iff $u$ is at level $i$. We now define the weights on edges in $E(\rho)$: all level-$d$ edges have weight $1$, and for each $1\le i\le d-1$, all level-$i$ edges have weight $2^{(d-1)-i}$. We denote by $L(\rho)$ the set of all leaves of $\rho$.

In order to avoid ambiguity, we refer to vertices of $\rho$ as nodes and points in $V$ as vertices.
The vertex set $V$ is partitioned into $(2^{d+1}-1)$ subsets: $V=\bigcup_{x\in V(\rho)} V_x$, where each subset is indexed by a node in $\rho$. For each leaf node $x$ in $L(\rho)$, the set $V_x$ contains a single vertex, that we denote by $u_x$. For each non-leaf node $x$ in $\rho$, the set $V_x$ contains either $\floor{\frac{n-2^{d}}{2^d-1}}$ or $\ceil{\frac{n-2^{d}}{2^d-1}}$ vertices (so that the total number of vertices in $V$ is $n$), with one of them designated as the \emph{special vertex} in $V_x$, denoted by $u_x$, and all other vertices are called \emph{regular vertices}. 
For consistency, for each leaf node $x\in L(\rho)$, we also call the only vertex $u_x$ in $V_x$ a special vertex.
The terminal set is defined to be $T=\set{u_x\mid x\in L(\rho)}$, so $|T|=2^d$.
We denote by $S$ the set of all special vertices, so $T\subseteq S$ and $|S|=2^{d+1}-1$.

We now define metrics $w_{\sf Y}$ and $w_{\sf N}$ as follows.
We first define $w_{\sf N}$. For every pair $v,v'$ of vertices in $V$, assume $v\in V_x$ and $v'\in V_{x'}$; denote by $\hat \ell$ the level of the lowest common ancestor of nodes $x$ and $x'$ in $\rho$; and assume without loss of generality that the level of $x$ is at least the level of $x'$. Now let $\tilde x$ be any leaf of $\rho$ that lies in the subtree of $\rho$ rooted at $x$; then $\wn(v,v')=\dist_{\rho}(\tilde x,x)+ \dist_{\rho}(\tilde x,x')$ (note that this is well-defined since any such leaf $\tilde x$ will give the same value of $\dist_{\rho}(\tilde x,x)+ \dist_{\rho}(\tilde x,x')$).
We now define $\wy$.
For every pair $v,v'\in V$ such that at least one of $v,v'$ does not lie in $S$, $\wy(v,v')$ is defined identically as $\wn(v,v')$.
For every pair $v,v'$ of vertices in $S$ with $v\in V_x$ and $v'\in V_{x'}$, the value $\wy(v,v')$ is defined slightly different as $\wy(v,v')=\dist_{\rho}(x,x')$.

We prove the following claim which says on $\wy$ and $\wn$ defined above are indeed metrics (that is, they satisfy the triangle inequality). The proof is based on a straightforward case analysis, and is deferred to \Cref{apd: Proof YN metrics}.

\begin{claim}
\label{clm: YN metrics}
$\wy,\wn$ are metrics on $V$.
\end{claim}

We next prove the following two claims showing that the minimum Steiner Tree cost of instances $(V,T,\wy)$ and $(V,T,\wn)$ differ by a factor of roughly $(5/3)$.

\begin{claim}
	\label{clm: yes_cost}
	The minimum \ST cost of instance $(V,T,\wy)$ is at most $(d+1)\cdot 2^{d-1}$.
\end{claim}
\begin{proof}
	Consider the following \St $\tau$ whose vertex set is $S$. Recall that $S=\set{u_x\mid x\in V(\rho)}$. The edge set of $\tau$ contains, for each edge $(x,x')\in E(\rho)$, an edge $(u_x,u_{x'})$. From the definition of $\wy$, it is easy to verify that the tree $\tau$ is identical to the tree $\rho$ (together with its edge weights). Therefore, 
	$$\wy(\tau)=w(\rho)=2^d+\sum_{1\le i\le d-1}2^{(d-1)-i}\cdot2^{i}=(d+1)\cdot 2^{d-1}.$$
\end{proof}

\begin{claim}
\label{clm: no_cost}
The minimum \ST cost of instance $(V,T,\wn)$ is at least $(5/3)\cdot d\cdot 2^{d-1}$.
\end{claim}
\begin{proof}
We start by proving the following observation on an optimal \St of instance $(V,T,\wn)$.
\begin{observation}
\label{obs: one from each group}
There exists an optimal \St $\tau$ of instance $(V,T,\wn)$, such that for each node $x\in V(\rho)$, $|V(\tau)\cap V_x|$ is either $0$ or $1$.
\end{observation}
\begin{proof}
Let $\tau^*$ be an optimal \St of instance $(V,T,\wn)$, and assume that there exists a node $x\in V(\rho)$, such that $\tau^*$ contains at least two distinct vertices of $V_x$, that we denote by $v,v'$. Clearly, $x$ is a non-leaf node in $\rho$ and so both $v$ and $v'$ are Steiner vertices in $\tau^*$. From the definition of $\wn$, for any vertex $u\in V, u\ne v,v'$, $\wn(u,v)=\wn(u,v')$. Let $\tau$ be the tree obtained from $\tau^*$ by replacing every $v'$-incident edge $(u,v')\in E(\tau^*)$ with edge $(u,v)$ and then removing any parallel edges. It is easy to verify that $\tau$ is also a \St of instance $(V,T,\wn)$, and that $w(\tau)\le w(\tau^*)$. Since $\tau^*$ is an optimal \St, $\tau$ has to be an optimal \St as well. We can keep modifying $\tau$ in the same way until every group $V_x$ contains at most one vertex in $\tau$, while ensuring that the resulting tree $\tau$ stays an optimal \St of instance $(V,T,\wn)$. This completes the proof of the observation.
\end{proof}

From \Cref{obs: one from each group}, and since all vertices in the same group behave identically with respect to $\wn$, we conclude that there is an optimal \St $\tau$ with $V(\tau)\subseteq S$.
Therefore, we from now on focus on showing that the optimal \St of instance $(S,T,\wn)$ is at least $(5/3)\cdot d\cdot 2^{d-1}$.
Recall that $S=\set{u_x\mid x\in V(\rho)}$.
We use the following observation, whose proof is deferred to \Cref{apd: Proof of edges love terminal}.

\begin{observation}
\label{obs: edges love terminal}
Let $\tau$ be an optimal \St of the instance $(S,T,\wn)$, then every edge of $\tau$ is incident to a vertex in $T$.
\end{observation}

For each non-leaf node $x\in V(\rho)$, we denote by $R(x)$ the subtree of $\rho$ rooted at $x$ and denote by $R_1(x)$, $R_2(x)$ the subtrees of $\rho$ rooted at two children of $x$, respectively. 
Recall that $S=\set{u_{x'}\mid x'\in V(\rho)}$. We define sets $S_1(x)=\set{u_{x'}\mid x'\in R_1(x)}$, $S_2(x)=\set{u_{x'}\mid x'\in R_2(x)}$ and $S_0(x)=\set{u_{x'}\mid x'\notin R(x)}$. 
We then define
$T_i(x)=T\cap S_i(x)$ for $i\in \set{0,1,2}$, and $T(x)=T_1(x)\cup T_2(x)$.
For convenience, for every node $x\in V(\rho)$ at level $i$ of $\rho$, we also say that $u_x$ is a \emph{level-$i$ vertex in $S$}. Similarly, if node $x$ is the parent of node $x'$ in $\rho$, we also say that vertex $u_x$ is the parent of vertex $u_{x'}$ in $S$.

We use the following observation, whose proof is deferred to \Cref{apd: Proof of OPT_props}.

\begin{observation}
\label{obs: OPT_props}
There is an optimal \St $\tau$ of the instance $(S,T,\wn)$, such that for every Steiner vertex $u_x$ in $\tau$: (i) $u_x$ has either one or two neighbors in $T_0(x)$, exactly one neighbor in $T_1(x)$, and exactly one neighbor in $T_2(x)$;
(ii) if we denote by $W_1$ and $W_2$ the connected components in the graph $\tau'\setminus \set{u_x}$ that contains the $T_1(x)$-neighbor of $u_x$ and the $T_2(x)$-neighbor of $u_x$, respectively, then $T_1(x)\subseteq V(W_1)\subseteq S_1(x)$, and
$T_2(x)\subseteq V(W_2)\subseteq S_2(x)$; and (iii)
for every vertex $u_x$ at level at most $d-2$ in $S$, exactly one vertex from the set containing $u_x$ and its two children belongs to $V(\tau)$.
\end{observation}

We are now ready to complete the proof of \Cref{clm: no_cost}. Let $r$ be the root of tree $\rho$. Note that metric $\wn$ and tree $\rho$ are both determined by a single nonnegative integer $d$. In order to avoid ambiguity, we denote by $(S,T,\wn)_d$ the instance determined by $d$.
For each $d\ge 0$, we define 
\begin{itemize}
\item $A(d)$ as the minimum cost of a Steiner tree of instance $(S,T,\wn)_d$ that does not contain $u_r$; and
\item $B(d)$ as the minimum cost of a Steiner tree of instance $(S,T,\wn)_d$ that contains $u_r$.
\end{itemize}
It is easy to verify that $A(1)=2$ and $B(1)=2$. We now show that, for each $d\ge 1$, 
\begin{equation}
\label{eq: recurrence}
A(d+1)=A(d)+B(d)+2^{d-1}+2^d; \quad\text{and}\quad 
B(d+1)=2\cdot A(d)+2^{d+1}.
\end{equation}

On the one hand, let $\tau$ be an optimal Steiner tree of instance $(S,T,\wn)_{d+1}$ that does not contain $u_r$.
Let $u_1,u_2$ be the children of $u_r$. From \Cref{obs: OPT_props}, exactly one of $u_1,u_2$ is contained in $\tau$. Assume w.l.o.g. that $u_1\in V(\tau)$. From \Cref{obs: OPT_props}, $\tau$ is the union of 
\begin{itemize}
\item a Steiner tree of instance $(S_1(u_r),T_1(u_r),\wn)_{d+1}$ that contains $u_1$ (denote by $\tau_1$);
\item a Steiner tree of instance $(S_2(u_r),T_2(u_r),\wn)_{d+1}$ that does not contain $u_2$ (denote by $\tau_2$); and
\item an edge connecting $\tau_1$ to $\tau_2$.
\end{itemize}
Note that instances $(S_1(u_r),T_1(u_r),\wn)_{d+1}$ and $(S_2(u_r),T_2(u_r),\wn)_{d+1}$ are identical to the instance $(S,T,\wn)_{d}$, so $w(\tau_1)=B(d)$ and $w(\tau_2)=A(d)$. Note that the minimum weight of an edge connecting $\tau_1$ to $\tau_2$ is the edge connecting $u_1$ to any leaf in $T_2(u_r)$, which has cost $2^{d-1}+2^{d}$. Therefore, $A(d+1)=A(d)+B(d)+2^{d-1}+2^d$.

On the one hand, let $\tau$ be an optimal Steiner tree of instance $(S,T,\wn)_{d+1}$ that contains $u_r$. From \Cref{obs: OPT_props}, both $u_1,u_2$ are contained in $\tau$. Indeed, $\tau$ is the union of 
\begin{itemize}
\item a Steiner tree of instance $(S_1(u_r),T_1(u_r),\wn)_{d+1}$ that does not contain $u_1$ (denote by $\tau_1$);
\item a Steiner tree of instance $(S_2(u_r),T_2(u_r),\wn)_{d+1}$ that does not contain $u_2$ (denote by $\tau_2$); and
\item an edge connecting $\tau_1$ to $u_r$ and  an edge connecting $\tau_2$ to $u_r$.
\end{itemize}
Via similar arguments, we can show that $w(\tau_1)=w(\tau_2)=A(d)$. Note that the minimum weight of an edge connecting $\tau_1$ to $u_r$ is the edge connecting any leaf in $T_1(u_r)$ to $u_r$, which has cost $2^{d}$. Therefore, $B(d+1)=2\cdot A(d)+2^{d+1}$.

We now use the inequality \ref{eq: recurrence} to complete the proof of \Cref{clm: no_cost}. From \ref{eq: recurrence}, we get that, for each $d\ge 1$, $A(d+1)=A(d)+2\cdot A(d-1)+5\cdot 2^{d-1}$. Using standard techniques and the initial values $A(0)=0$ and $A(1)=2$, we get that
\[A(d)=\bigg(\frac{5}{6}\bigg)\cdot d \cdot 2^{d}+ \bigg(\frac{-1}{9}\bigg)\cdot (-1)^d+\bigg(\frac{1}{9}\bigg)\cdot 2^d.\]
Therefore, from \ref{eq: recurrence}, we get that
\[
\begin{split}
B(d)=2\cdot A(d-1)+2^{d} & =2\cdot \left( \bigg(\frac{5}{6}\bigg)\cdot (d-1) \cdot 2^{d-1}+ \bigg(\frac{-1}{9}\bigg)\cdot (-1)^{d-1}+\bigg(\frac{1}{9}\bigg)\cdot 2^{d-1} \right)+2^{d}\\
& = \bigg(\frac{5}{6}\bigg)\cdot d \cdot 2^{d}+ \bigg(\frac{2}{9}\bigg)\cdot (-1)^{d}+\bigg(\frac{5}{18}\bigg)\cdot 2^{d}.
\end{split}
\]
Therefore, $A(d), B(d)\ge (5/6)\cdot d\cdot 2^d$. This completes the proof of \Cref{clm: no_cost}.
\end{proof}

From \Cref{clm: yes_cost} and \Cref{clm: no_cost}, we get that
\[
\frac{\stcost(\ino)}{\stcost(\iy)}\ge \frac{(5/3)\cdot d\cdot 2^{d-1}}{(d+1)\cdot 2^{d-1}}\ge \frac{5}{3}- \frac 1 d \ge \frac{5}{3}- \eps.
\]

We now complete the proof of \Cref{thm: lower-main} using the metrics $\wy,\wn$ defined above.

We construct a pair $\dset_{\sf Y},\dset_{\sf N}$ of distributions on metric Steiner Tree instances $(V',T',w')$ as follows. Set $V'$ is fixed and contains $n$ vertices. Let $\fset$ be the set of all one-to-one mappings from $V'$ to $V$. For each mapping $f\in \fset$, we define a pair of instances $\iy^f$ and $\ino^f$ as follows:
\begin{itemize}
\item $\iy^f=(V', f^{-1}(S), \wy^f)$, where $\wy^f$ is defined as: $\forall v,v'\in V'$, $\wy^f(v,v')=\wy(f(v),f(v'))$;
\item $\ino^f=(V', f^{-1}(S), \wn^f)$, where $\wn^f$ is defined as: $\forall v,v'\in V'$, $\wn^f(v,v')=\wn(f(v),f(v'))$;
\end{itemize}
where $f^{-1}(S)=\set{v\in V'\mid f(v)\in S}$. We then let $\dset_{\sf Y}$ be the uniform distribution on all instances in $\set{\iy^f\mid f\in \fset}$, and 
let $\dset_{\sf N}$ be the uniform distribution on all instances in $\set{\ino^f\mid f\in \fset}$.
Let $\dset$ be the distribution that sample an instance from $\dset_{\sf Y}$ with probability $1/2$, and sample an instance from $\dset_{\sf N}$ with probability $1/2$.

It is easy to verify that for each mapping $f\in \fset$, $\stcost(\iy^f)=\stcost(\iy)$ and $\stcost(\iy^f)=\stcost(\iy)$ hold, so any algorithm that with probability $0.51$ estimates the cost to within factor $(5/3-\eps)$ can correctly report a random instance from $\dset$ comes from $\dset_{\sf Y}$ or $\dset_{\sf N}$ with probability $0.51$.

Recall that the metrics $\wy$ and $\wn$ are identical on all pairs of vertices that are not both terminals. For each instance $\iy^f$, we say a pair $(v'_1,v'_2)$ of vertices in $V'$ is \emph{crucial} iff $v'_1,v'_2\in f^{-1}(S)$, and we say that the pair $(v'_1,v'_2)$ is \emph{discovered} iff the pair $(v'_1,v'_2)$ is queried by the algorithm.
From Yao's minimax principle \cite{yao1977probabilistic} and the above discussion, in order to distinguish between $\dset_{\sf Y}$ and $\dset_{\sf N}$ with probability $0.51$, it is necessary that the algorithm discovers a crucial pair on at least $0.01$-fraction of the instances in $\iy$. Therefore, the proof of \Cref{thm: lower-main} is concluded by the following lemma.

\begin{lemma}
Any deterministic algorithm that discovers a crucial pair on at least $0.01$-fraction of the instances in $\iy$ performs at least $\Omega(n^2/2^{2d})$ queries in expectation.
\end{lemma}
\begin{proof}
Since $f$ is a random one-to-one mapping from $V'$ to $V$, the set $f^{-1}(S)$ is a random size-$|S|$ subset of $V'$. Therefore, the probability that a single query discoveres a crucial pair is $\binom{|S|}{2}/\binom{n}{2}$, and it follows that the expected number of queries that is required to discovers a crucial pair with probability at least $\Omega(1)$ is $\Omega\big(\binom{n}{2}/\binom{|S|}{2}\big)=\Omega(n^2/2^{2d})$.
\end{proof}

\section{Algorithm for a $(2-\eps_0)$-Estimation of Steiner Tree Cost}
\label{sec: beat-2-upper}

In this section we provide the proof of \Cref{thm: beat-2-main}.
Specifically, we will construct an algorithm, that, takes as input a set $V$ of $n$ points, a set $T\subseteq V$ of $k$ terminals, and access to a metric $w$ on $V$, estimates the metric \ST cost to within a factor of $(2-\eps_0)$, for some universal constant $\eps_0$ that does not depend on $n$ and $k$, by performing $\tilde O(n^{12/7}+n^{6/7}\cdot k)$ queries.
This section is organized as follows.
First, in \Cref{subsec: with MST}, we give an algorithm that, assumes that the induced metric of $w$ on $T$ is known to us upfront, estimates the metric \ST cost to within a factor of $(2-\eps_0)$ with $\tilde O(n^{3/2}+n^{3/4}\cdot k)$ queries. The detail of some critical subroutine of the algorithm is provided in \Cref{sec: set cover}. Then in \Cref{subsec: without MST}, we show how to remove the assumption that the induced metric of $w$ on $T$ is given, while still attaining a slighly worse (while still sublinear) query complexity $\tilde O(n^{12/7}+n^{6/7}\cdot k)$.

\subsection{An Algorithm with the Terminal-Induced Metric Given Upfront}
\label{subsec: with MST}

In this subsection, we assume that the metric on terminals induced by $w$ is given to us upfront, and give an algorithm that estimates the metric \ST cost to within a factor of $(2-\eps_0)$.
We first give a high-level overview, and then describe the algorithm in detail and provide its analysis.

\subsubsection*{Overview of the algorithm}

We start by constructing a minimum spanning tree over the terminals, say $\tau^*$. Let $\mst$ denote the weight of $\tau^*$, so $\stcost(V,T,w)\ge \mst/2$.
The rest of the algorithm focuses on gathering ``local evidence'' to ascertain that $\stcost(V,T,w)$ is bounded away from $\mst$. If the algorithm fails to find the evidence to support this assertion, then we will be able to claim that $\stcost(V,T,w)$ is bounded away from $\mst/2$.

We now describe what constitutes this local evidence. At a high-level, it is some property of the metric $w$ that allows us to locally ``restructure'' the minimum terminal spanning tree maintaining connectivity among the terminals while reducing its total cost. To get some intuition for this process, let us consider a metric where all distances are $1$ or $2$. Suppose that the distances between all terminals are $2$, the distances between all Steiner vertices are $2$, and the distances between a terminal and a Steiner vertex is either $1$ or $2$. Clearly, $\mst=2k-2$. Assume that a Steiner vertex $v$ is at distance $1$ to three terminals $u_1,u_2,u_3$. Now if we remove two edges from $\tau^*$ such that terminals $u_1,u_2,u_3$ lie in different connected subtrees, and then add the edges $(v,u_1),(v,u_2),$ and $(v,u_3)$, then we get another Steiner Tree that now contains $v$. Note that, in this process we have deleted two edges of cost $2$ each and added three edges of cost $1$ each, so essentially the total cost decreases by $1$. We view this ``Steiner vertex $v$ connects to terminals $u_1$, $u_2$, $u_3$ via length-$1$ edges'' structure as a ``local evidence that separates $\stcost(V,T,w)$ from $\mst$''.

It is not hard to observe that, this type of evidence is both local and aggregatable, e.g., if a Steiner vertex $v$ is at distance $1$ to terminals $u_1$-$u_3$ and another Steiner vertex $v'$ is at distance $1$ to terminals $u_4$-$u_7$, then we can ``save a total of $(4-3)+(6-4)=3$ units of cost'' from $\mst$. The process of identifying the best way to aggregate these local cost-saving improvements is reminiscent of solving an instance of {\em Set Cover}. Specifically, if we define, for each Steiner vertex $v$, a set $W_v$ containing all terminals $u\in T$ with $w(u,v)=1$, then a good aggregation of the local evidence is a collection of a small number of sets $W_v$ that cover many terminals. In particular, if we denote $\wset=\set{W_v\mid v\notin T}$, then using similar ``local MST restructure''-type arguments, we can show that $\stcost(V,T,w)\le \mst-\Omega(k-\setcover(T,\wset))$, where $\setcover(T,\wset)$ is the minimum solution size of the Set Cover instance $(T,\wset)$. Therefore, our goal now is to estimate the value of $k-\setcover(T,\wset)$ to within an additive $\eps k$ factor (and some small multiplicative factor). We provide an $\tilde O(n^{3/2}+n^{3/4}\cdot k)$-query algorithm for this task in \Cref{sec: set cover}, and then show how to implement this algorithm when the terminal-induced metric is not given upfront, with a slightly worse query complexity $\tilde O(n^{12/7}+n^{6/7}\cdot k)$.

We next describe how the ideas outlined above for the special case of $(1,2)$-metric, can be extended to the general case. Let us consider the construction of the minimum spanning tree $\tau^*$ on $T$ using Kruskal's algorithm. 
Assume that the weight of every terminal-terminal edge is $(1+\eps)^i$ for some non-negative integer $i$. Then in the first round we add all weight-$1$ edges and obtain some connected components (called \emph{clusters}), and in the second round we add all weight-$(1\!+\!\eps)$ edges, and some clusters in the first round are merged into bigger clusters, etc. The main observation is that, in every round, we can use the Steiner vertices to locally restructure this cluster-merging step just as the special case. In particular, if there exists a Steiner vertex $v$ and three first-round clusters $U_1,U_2,U_3$, such that $U_1,U_2,U_3$ are merged in the second round, and $w(v,U_1),w(v,U_2),w(v,U_3)$ are close to $(1+\eps)/2$, then we can replace two weight-$(1\!+\!\eps)$ edges with three weight-roughly-$(1\!+\!\eps)/2$ edges, thereby saving the total cost by roughly $(1\!+\!\eps)/2$ without destroying the connectivity between the terminals in $U_1,U_2,U_3$. Therefore, the main framework of our algorithm is to compute the hierarchical structure of the terminal minimum spanning tree $\tau^*$ and use this set-cover-type algorithm at every ``level'' of $\tau^*$ to search for local evidence that there is a Steiner tree of cost significantly better than $\tau^*$.

There is one subtlety in the above algorithmic framework, which is that the cardinality-$2$ sets do not provide cost-saving. In particular, if we replace one terminal-terminal edge with two terminal-Steiner edge (of weight about half of that of a terminal-terminal edge), then the total cost will not decrease. More concretely, if we consider the metric $\wy$ defined in \Cref{sec: 5/3-lower}, which is the shortest-path metric induced by a complete binary tree with edge cost geometrically decreasing along with the levels, then when we construct Set Cover instances on different levels, we will only get cardinality-$2$ sets and ended up finding no local evidence at all, but in fact $\stcost(V,T,\wy)=\mst/2$ holds. To overcome this issue, we introduce another evidence-searching subroutine that goes beyond a single level of the hierarchical structure of $\tau^*$ while involving only $O(1)$ vertices. This is based on the observation that, in this very special case, although local evidence cannot be found at a single level, it can be found by looking at two consecutive levels and only focusing on sets of $O(1)$ vertices. It turns out that incorporating this subroutine into the above framework gives us the desired algorithm.

$\ $

We now describe the algorithm in detail. Recall that we are given an instance $(V,T,w)$ of the metric Steiner Tree problem and are allowed to perform queries to metric $w$. Also recall that the induced metric of $w$ on $T$ is also given to us upfront.
We use the following parameters: $\underline{\eps_0= 2^{-40}; \eps=2^{-20}}$.

\subsection*{Step 1. Computing an MST on $T$ and its hierarchical structure}

 We pre-process the instance $(V,T,w)$ as follows. 
Let $D=\max\set{w(u,u')\mid u,u'\in T}$.
While there exists a pair $u,u'\in T$ with $w(u,u')\le D/k^2$, we delete an arbitrary one of them from $T$ and $V$. We repeat this until all pairwise distances between vertices of $T$ are at least $D/k^2$. 
Let $T'$ be the resulting terminal set we get, and define $V'=(V\setminus T)\cup T'$.
We then scale the metric $w$ by defining another metric $w'$ on $V'$ such that for every pair $v,v'\in V$, $w'(v,v')=w(v,v')/\min\set{w(u,u')\mid u,u'\in T'}$, so now the distance (under $w'$) between every pair of terminals in $T'$ is at least $1$ and at most $k^2$. It is easy to verify that: (i) the metric \ST cost $\stcost(V',T',w')$ of instance $(V',T',w')$ is within factor $(1+O(1/k))$ of $\stcost(V,T,w)/\min\set{w(u,u')\mid u,u'\in T'}$; and (ii) every distance query to $w'$ can be simulated by a distance query to $w$.
Therefore, from now on we work with instance $(V',T',w')$, and we will construct an algorithm that with high probability estimates the value of $\stcost(V',T',w')$ to within a factor of $(2-2\eps_0)$.
Eventually, when the algorithm returns an estimate $X$ of $\stcost(V',T',w')$, we return $X\cdot\min\set{w(u,u')\mid u,u'\in T'}$ as the output estimate of $\stcost(V,T,w)$. 
It is easy to verify that the final output of the algorithm is with high probability a $(2-\eps_0)$-approximation of $\stcost(V,T,w)$.
For convenience, in the remainder of this section, we rename the vertex set $V'$, the terminal set $T'$ and the metric $w'$ by $V,T,w$, respectively.

Let $L=\ceil{\log_{1+\eps}k^2}$. For every pair $u,u'\in T$, we say that the edge $(u,u')$ is \emph{at level $i$} (or $(u,u')$ is an \emph{level-$i$} edge), iff $(1+\eps)^{i-1}\le w(u,u')< (1+\eps)^{i}$.
Clearly, every edge connecting a pair of terminals is at level at most $L$.
For each $1\le i\le L$, we define $H_i$ to be the graph on $T$ that contain all edges up to level $i$, and we define $H_0$ to be the empty graph on $T$. For each index $1\le i\le L$, we define $\sset_i$ as the collection of vertex sets of connected components of graph $H_{i-1}$. That is, each set in $\sset_i$ contains all vertices of some connected component of $H_{i-1}$. Clearly, each $\sset_i$ is a partition of $T$. 

Let $\sset=\bigcup_{1\le i\le L}\sset_i$. It is easy to verify that $\sset$ is a laminar family. That is, every pair $S,S'$ of sets in $\sset$, either $S\subseteq S'$, or $S'\subseteq S$, or $S\cap S=\emptyset$. We associate with $\sset$ a partition tree $\tset$ as follows. The vertex set of $\tset$ contains, for each set $S\in \sset$, a node $x_S$ representing the set $S$. The edge set of $\tset$ contains, for each pair $S,S'\in \sset$ such that $S\subseteq S'$ and $S,S'$ lie on consecutive levels, an edge $(x_S,x_{S'})$. In this case, we say that $x_{S'}$ is the parent node of $x_S$ (and $x_S$ is a child node of $x_{S'}$); similarly, we say that $S'$ is a parent set of $S$ (and $x_S$ is a child set of $S'$). Note that $\sset_L$ contains a single set, and its corresponding node in $\tset$ is designated as the root of $\tset$. It is easy to verify that $\tset$ is a tree.

Lastly, we compute a minimum spanning tree $\tau^*$ on $T$, and denote $\mst=w(\tau^*)$.
As the induced metric of $w$ on $T$ is given to us upfront, in this step we did not perform any additional queries.

\subsection*{Step 2. Finding local evidence using a set-cover-type subroutine}

We start by introducing a set-cover-type subroutine that we will use in this step.

\vspace{-8pt}
\paragraph{Algorithm $\algsetcover$.}
Let $(U,\wset)$ be an instance of the Set Cover problem, where $U$ is a collection of elements and $\wset$ is a collection of subsets of $U$. For convenience, throughout the paper, when we consider an instance $(U,\wset)$ of Set Cover, we always assume that $\wset$ contains, for each element $u\in U$, a singleton set $\set{u}$, so $|\wset|\ge |U|$, and if we denote by $\setcover(U,\wset)$ the size of the smallest set cover for the instance $(U,\wset)$, then $\setcover(U,\wset)\le |U|$.
The elements of $U$ and the number of sets in $\wset$ are known to us, but we do not know which elements each set of $\wset$ contains. We are allowed to perform queries to a \emph{membership oracle} of instance $(U,\wset)$, which is an oracle that, takes as input an element $u\in U$ and a set $W\in \wset$, returns whether or not $u$ belongs to $W$. A query to the membership oracle of instance $(U,\wset)$ is also called a \emph{membership query} to instance $(U,\wset)$.

For a collection $\wset$ of subsets of $U$, we denote by $\wset_{\ne 2}$ the collection that contains all sets in $\wset$ with size not equal to $2$.
For positive real numbers $X,Y,a,b$ with $a>1$, we say that $X$ is an {\em $(a,b)$-estimation of $Y$} iff $Y\le X\le aY+b$.

We use the following theorem, whose proof is deferred to \Cref{sec: set cover}. 

\begin{theorem}
\label{thm: sec cover}
There is a polynomial-time randomized algorithm called $\algsetcover$, that, given any instance $(U,\wset)$ of Set Cover and any constant $0<\eps<1$, with high probability, returns a $(4,\eps|U|)$-estimation of $\big(|U|-\setcover(U,\wset_{\ne 2})\big)$, by performing $O((|\wset|^{3/2}+|\wset|^{3/4}|U|) (\log|\wset|)^2/\eps^3 )$ membership queries to the instance $(U,\wset)$.
\end{theorem}

We remark that improving the query complexity of the result above will immediately improve the query complexity for the $(2-\eps_0)$-estimation algorithm. Also, note that the number of queries needed by a naive algorithm to solve the estimation problem above would be $O(|U||\wset|)$, the size of the input description. So the theorem above provides an estimation algorithm that is sublinear in the size of the input description whenever $|\wset|^{3/2} = o(|U||\wset|)$.

For each index $1\le i\le L$ and each pair $u,u'\in T$, we let $w_i(u,u')=w(u,u')$ if $(u,u')$ is a level-$i$ edge; otherwise we let $w_i(u,u')=0$. 
Consider the minimum spanning tree $\tau^*$ computed in Step $1$. For each $1\le i\le L$, we say that level $i$ is \emph{light} iff $w_i(\tau^*)<\mst/(L\log n)$ (that is, the total weight of all level-$i$ edges in $\tau^*$ is less than $\mst/(L\log n)$); otherwise we say that level $i$ is \emph{heavy}. 
For a set $E$ of edges in $E(\tau^*)$, we define $w_i(E)=\sum_{(u,u')\in E}w_i(u,u')$ and call $w_i(E)$ the \emph{level-$i$ weight} of set $E$.
The following observation is immediate.

\begin{observation}
$\sum_{i: \text{ level }i\text{ is light }}w_i(\tau^*)\le \mst/\log n$.
\end{observation}

Before we describe the set-cover-type subroutine in detail, we give some intuition. We intend to find local evidence on different levels separately. For each $0\le i\le L$ and for each set $S\in \sset_i$, we think of vertices in $S$ as already connected via edges up to level $i$. So we will search for better ways to connect different sets in $\sset_i$ via Steiner vertices. However, for Steiner vertex $v$ and set $S\in \sset$, naively it takes $O(|S|)$ queries to compute $w(v,S)$, which we cannot afford as $|S|$ can be as large as $k$. Therefore, for each $S\in \sset_i$, we will first compute a subset $\tilde S\subseteq S$ as its ``representative'', such that $|\tilde S|$ is small, and the values $w(v,S)$ and $w(v,\tilde S)$ are close for all Steiner vertices $v$.

We now describe the set-cover-type subroutine in detail.
For each level $i$ that is heavy, we construct an instance $I_i=(U_i,\wset_i)$ of Set Cover as follows. Recall that $\sset_i$ is the collection of all level-$i$ sets.
We first compute, for each index $0\le i\le L$ and for each set $S\in \sset_i$, a maximal subset $\tilde S$ of $S$, such that the distance between every pair of terminals in $\tilde S$ is at least $\eps\cdot (1+\eps)^i$.
We define $\tilde\sset_i$ to be the collection that contains all such sets $\tilde S$ with $|\tilde S|\le (L\log^2 n)/\eps$.
The ground set $U_i$ in instance $I_i$ is defined as $U_i=\set{x_S\mid \tilde S\in \tilde\sset_i}$. 
The collection $\wset_i$ contains, for each vertex $v\in V\setminus T$, a set $W_i(v)$ that is defined as 
$W_i(v)=\set{x_S\mid w(v,\tilde S)\le (3/5)\cdot (1+\eps)^{i}}$. We use the following simple observations.

\begin{observation}
\label{obs: query_simulation}
$|U_i|\le k$, $|\wset_i|\le n-k$, and every membership query in the instance $(U_i,\wset_i)$ can be simulated by at most $(L\log^2 n)/\eps$ distance queries to the metric $w$.
\end{observation}

\begin{observation}
\label{obs: small sets enough}
For each level $i$ that is heavy, $|U_i|\ge (1-O(1/\log n))\cdot|\sset_i|$.
\end{observation}
\begin{proof}	
Recall that $U_i=\set{x_S\mid \tilde S\in \tilde\sset_i}$. It suffices to show that, if level $i$ is heavy, then there are at most $O(1/\log n)$ fraction of sets $S$ in $\sset_i$ with $|\tilde S|>(L\log^2 n)/\eps$. Define $X=\bigcup_{S\in \sset_i}\tilde S$. Note that every pair of vertices in $X$ are at distance at least $\eps\cdot (1+\eps)^i$ in $w$. Therefore, $|X|\le  1+\mst/(\eps\cdot (1+\eps)^i)\le 2\cdot \mst/(\eps\cdot (1+\eps)^i)$. On the other hand, since level $i$ is heavy, $w_i(\tau^*)\ge \mst/(L\log n)$, and so $|\sset_i|\ge \mst/(L\cdot\log n\cdot (1+\eps)^{i-1})$. Altogether,
\[
\frac{L\log^2 n}{\eps}\cdot\bigg|\set{S\in \sset_i\mid |\tilde S|>\frac{L\log^2 n}{\eps}}\bigg|\le 
|X|\le \frac{2\cdot \mst}{\eps\cdot (1+\eps)^i}\le |\sset_i|\cdot \frac{2L\cdot\log n}{\eps\cdot(1+\eps)},\]
and it follows that there are at most $O(1/\log n)$ fraction of sets $S$ in $\sset_i$ with $|\tilde S|>(L\log^2 n)/\eps$.
\end{proof}

Then for each $0\le i\le L$, we apply the algorithm $\algsetcover$ to instance $I_i=(U_i,\wset_i)$ constructed above, and obtain an estimate $X_i$ of $\big(|U_i|-\setcover(U_i,(\wset_i)_{\ne 2})\big)$. If 
$\sum_{0\le i\le L}(1+\eps)^i\cdot X_i>2^{30}\cdot \eps_0\cdot\mst$, then we return $(1-\eps_0)\cdot\mst$ as an estimate of $\stcost(V,T,w)$. Otherwise, we go to the next step. 
Since the algorithm in \Cref{thm: sec cover} performs $(|U_i|^{}|\wset_i|^{3/4}+|\wset_i|^{3/2})\cdot(\log(|U_i|+|\wset_i|))^{O(1)}\le \tilde O(n^{3/2}+n^{3/4}k^{})$ queries, using \Cref{obs: query_simulation}, we get that the algorithm performs $\tilde O(n^{3/2}+n^{3/4}k^{})\cdot O(L\log^2 n)=\tilde O(n^{3/2}+n^{3/4}k^{})$ queries on the metric $w$ in this step.

\subsubsection*{Analysis of Step 2 when the algorithm returns $(1-\eps_0)\cdot\mst$ as the estimate of $\stcost(V,T,w)$}

Before we proceed to describe the next steps of the algorithm, we first show in this subsection that, if we collected enough local evidence in this step, then indeed $\stcost(V,T,w)$ is bounded away from $\mst$. Specifically, we will show that, if 
$\sum_{0\le i\le L}(1+\eps)^i\cdot X_i>2^{30}\cdot\eps_0\cdot\mst$, then $\stcost(V,T,w)\le (1-\eps_0)\cdot\mst$. Since $\stcost(V,T,w)\ge \mst/2$, our estimate in this case is indeed a $(2-2\eps_0)$-approximation of $\stcost(V,T,w)$.
For each $0\le i\le L$, we define $Y_i=|U_i|-\setcover(U_i,(\wset_i)_{\ne 2})$.
Define $L'=\ceil{\log_{1+\eps}2}$, so $L'$ is a constant between $2^{19}$ and $2^{20}$ from the definition of $\eps$.
We start by proving the following claim.

\begin{claim}
\label{clm: set cover construct steiner tree}
For each $0\le i\le L-1$, there exists a set $E_i$ of edges, each connecting a terminal in $T$ to a Steiner vertex in $V\setminus T$, such that (i) the vertex sets of the connected components in graph $H_{i-1}\cup E_i$ are exactly the sets in $\sset_{i+L'-1}$; and (ii) $w(E_i)\le \big(\sum_{i\le s< i+L'}w_s(\tau^*)\big)- (1/20)\cdot (1+\eps)^i\cdot |Y_i|$.
\end{claim}
\begin{proof}
Fix an index $0\le i\le L-1$ and consider the instance $I_i=(U_i,(\wset_i)_{\ne 2})$.
Let $\wset^*_i$ be an optimal set cover of instance $(U_i,(\wset_i)_{\ne 2})$. 
We compute a sub-collection $\widetilde \wset_i$ of $\wset^*_i$ as follows. 
We process the sets of $\wset^*_i$ in an arbitrary order. Upon processing each set in $\wset^*_i$, we add it into to $\widetilde \wset_i$ iff the set contains at least two elements that are not contained in all previously processed sets in $\wset^*_i$. Denote the resulting set by  $\widetilde \wset_i=\set{W_i(v_1),\ldots,W_i(v_r)}$, where the sets are indexed according to the order in which they are added to $\widetilde \wset_i$. For each $1\le j\le r-1$, define $U_i(v_j)=W_i(v_j)\setminus \big(\bigcup_{1\le t\le j-1}W_i(v_t)\big)$. From the above discussion, sets $U_i(v_1),\ldots,U_i(v_r)$ are mutually disjoint and each containing at least two elements.

On the one hand, we show that $\sum_{1\le j\le r}|U_i(v_j)|\ge 2\cdot |Y_i|$. In fact, in the process of iteratively processing the sets of $\wset^*_i$ to obtain a subcollection $\widetilde \wset_i$, every set that is not added into $\widetilde \wset_i$ contains exactly one element that does not lie in previously processed sets. Therefore, at least $Y_i$ sets are eventually added to $\widetilde \wset_i$. Note that $|U_i(v_j)|\ge 2$ for each $1\le j\le r$, we get that $\sum_{1\le j\le r}|U_i(v_j)|\ge 2\cdot |Y_i|$.

On the other hand, we construct the set $E_i$ of edges via the following iterative process. Throughout, we maintain a set $\hat E$ of edges, that is initialized to be the set of all edges in $\tau^*$ from level $i$ to level $i+L'$ (so the initial total weight of $\hat E$ is $\sum_{i\le s< i+L'}w_s(\tau^*)$). 
We will ensure that set $\hat E$ always satisfies the property (i) in the claim. That is, the vertex sets of the connected components in graph $H_{i-1}\cup \hat E$ are exactly the sets in $\sset_{i+L'-1}$.
We iteratively process Steiner vertices $v_1,\ldots, v_r$ while modifying set $\hat E_i$ as follows. Consider now the iteration of processing $v_j$ for some $1\le j\le r$. Denote $U_i(v_j)=\set{S_1,\ldots, S_p}$, where $S_1,\ldots, S_p\in \sset_{i-1}$. 
Clearly, sets $S_1,\ldots, S_p$ are subsets of the same set in $\sset_{i+L'-1}$, as each pair of them is at distance at most $(6/5)\cdot (1+\eps)^{i}<(1+\eps)^{i+L'-1}$ in $w$.
Moreover, from the definition of $W_i(v_j)$, for each $1\le q\le p$, there exists a terminal $u_q\in S_q$ such that $w(v_j,u_q)\le (3/5)\cdot (1+\eps)^{i}$.
We distinguish between the following two cases.

Case 1. $|U_i(v_j)|\ge 3$. We simply add edges $(v_j,u_1),\ldots,(v_j,u_p)$ into the set $\hat E$. Since initially set $\hat E$ contains all edges of $E(\tau^*)$ at level $i$, and since elements in $U_i(v_j)$ do not belong to any other set in $\set{U_i(v_1),\ldots,U_i(v_r)}$, it is easy to see that, we can delete $(p-1)$ edges from the current set $\hat E$ that are level-$i$ edges of $E(\tau^*)$, such that the resulting set $\hat E$ still satisfies property (i) in the claim.

Case 2. $|U_i(v_j)|= 2$. Since $|W_i(v_j)|\ge 3$, there must exist another set $S_0\in \sset_i$ such that element $x_{S_0}$ is contained in $W_i(v_j)$ and some previous set $W_i(v_{j'})$ (for some $j'<j$), and so there exists a terminal $u_0\in S_0$ with $w(v_j,u_0)\le (3/5)\cdot (1+\eps)^{i}$. We simply add edges $(v_j,u_0),(v_j,u_1),(v_j,u_2)$ into the set $\hat E$. For similar reasons, it is easy to see that we can delete $2$ edges from the current set $\hat E$ that are level-$i$ edges of $E(\tau^*)$, such that the resulting set $\hat E$ still satisfies property (i) in the claim.

In either case, we add $t_j$ edges into set $\hat E$ and delete $(t_j-1)$ from set $\hat E$, for some $t_j\ge 3$ (and in fact $t_j\ge |U_i(v_j)|$). Since the edges that are added to set $\hat E$ have weight at most $(3/5)\cdot (1+\eps)^{i}$, and the edges that are deleted from set $\hat E$ have weight at least $(1+\eps)^{i-1}$. In the iteration of processing $v_j$, the total weight of $\hat E$ decreases by at least 
$(t_j-1)\cdot (1+\eps)^{i-1}-t_j\cdot\big((3/5)\cdot (1+\eps)^{i}\big)\ge (1/40)\cdot t_j\cdot (1+\eps)^{i}$, as $t_j\ge 3$. Therefore, the accumulative decrease of the total weight of $\hat E$ has weight is at least \[\sum_{1\le j\le r}\frac{t_j\cdot (1+\eps)^{i}}{40}\ge \sum_{1\le j\le r}\frac{|U_i(v_j)|\cdot (1+\eps)^{i}}{40}
\ge  \frac{ |Y_i|\cdot (1+\eps)^{i}}{20}.\]
Denote the resulting set $\hat E$ by $E_i$, and the claim now follows.
\end{proof}

We now use \Cref{clm: set cover construct steiner tree} to complete the analysis of Step 2. For each index $0\le j\le L'-1$, we define $\alpha_j=\sum_{0\le i\le L:\text{ }i\equiv j (\text{mod } L')}Y_i\cdot (1+\eps)^i$. Clearly, 
$\sum_{0\le j< L'}\alpha_j=\sum_{0\le i\le L}(1+\eps)^i\cdot Y_i$, and so there exists some $0\le j^*\le L'-1$ such that $\alpha_{j^*}\ge (1/L')\cdot\sum_{0\le i\le L}(1+\eps)^i\cdot Y_i$.
We now define $E'$ as the set that contains (i) all edges of $\tau^*$ that are at level $0,1,\ldots,j^*-1$; and (ii) all edges of $E_{j^*}, E_{j^*+L'},E_{j^*+2L'},\ldots,$ that are given by \Cref{clm: set cover construct steiner tree}. From \Cref{clm: set cover construct steiner tree}, it is easy to verify that the graph induced by edges of $E'$ is a Steiner tree of instance $(V,T,w)$, and moreover, 
\[
\begin{split}
w(E') & \le w(\tau^*)-\frac{1}{20}\cdot\alpha_{j^*}\\
& \le w(\tau^*)-\frac{1}{20 L'}\cdot \sum_{0\le i\le L}(1+\eps)^i\cdot Y_i\\
& \le w(\tau^*)-\frac{1}{20 L'}\cdot \sum_{0\le i\le L}(1+\eps)^i\cdot \frac{X_i-\eps |U_i|}{4}\\
& = w(\tau^*)-\frac{1}{20 L'}\cdot\bigg(\sum_{0\le i\le L} \frac{(1+\eps)^i\cdot X_i}{4}-\sum_{0\le i\le L} \frac{(1+\eps)^i\cdot\eps|U_i|}{4}\bigg)\\
& \le w(\tau^*)-\frac{1}{20 L'}\cdot\bigg( \frac{2^{30}\cdot\eps_0\cdot\mst}{4}-\eps\cdot \mst\bigg) \le (1-\eps_0)\cdot w(\tau^*),
\end{split}
\]
according to the definition of $\eps,\eps_0$ and the fact that $2^{19}\le L'\le 2^{20}$. This shows that our estimate in this case is indeed a $(2-2\eps_0)$-approximation of $\stcost(V,T,w)$.

\subsection*{Step 3. Finding local evidence using a $4$-vertex subroutine}

In the third and last step, we focus on finding one specific type of local evidence, by querying distances related to groups of $4$ vertices.

Recall that we have computed a laminar family $\sset$ of subsets of terminals in $T$ and its partition tree $\tset$. We say that a node $x_S$ in $\tset$ is \emph{good} iff $x_S$ has exactly two children in $\tset$, and each child node of $x_S$ also has exactly two children in $\tset$. In this case, we also say that the corresponding set $S$ in $\sset$ is good.
Consider a good set $S\in \sset$. Let $S_1,S_2$ be its child sets, let $S_{11},S_{12}$ be the child sets of $S_1$, and let $S_{21},S_{22}$ be the child sets of $S_2$. We define the \emph{advantage} of set $S$, denoted by $\adv(S)$, as follows. 
We define $w^*(S)=w(S_{11},S_{12})+w(S_{21},S_{22})+w(S_{1},S_{2})$, that is, the total edge weight in $\tau^*$ that is used to connect the four sets $S_{11},S_{12},S_{21},S_{22}$ into a single set $S$.
We say that a set $Y$ \emph{represents} $S$, iff $Y$ contains exactly four terminals $u_{11},u_{12},u_{21},u_{22}\in T$, such that $u_{11}\in S_{11}$, $u_{12}\in S_{12}$, $u_{21}\in S_{21}$, and $u_{22}\in S_{22}$.
For a set $Y$ that represents $S$, we define $\adv(S,Y)=w^*(S)-\min_{v\in V\setminus T}\set{\stcost(Y\cup\set{v},Y,w)}$, so $\adv(S,Y)$ is the maximum cost reduction (local evidence) that can be achieved with the help of any single Steiner vertex.
We define $\adv(S)=\max_{Y}\set{\adv(S,Y)}$.
Intuitively, in this step we are searching for the benefit of utilizing one Steiner vertex to restructure a specific type of $2$-level local structure in $\tau^*$.


We now describe the algorithm in this step. 
We denote by $\sset_g$ the collection of all good sets in $\sset$. For each $0\le i\le L$, let $\sset^{i}_g$ be the set of all level-$i$ good sets. For each good set $S\in \sset^i_g$, similar to Step 2, we compute a maximal subset $\tilde S$ of $S$, such that the distance between every pair of terminals in $\tilde S$ is at least $\eps\cdot (1+\eps)^i$.
We then define 
$$A_i=\sum_{S\in \sset^i_g: |\tilde S|\le (L\log^2 n/\eps)}\adv(S)\cdot \mathbf{1}\big[\adv(S)\ge \eps^{3/4}(1+\eps)^i\big].$$ 
However, we are unable to compute $A_i$ using few queries, as the number of sets $S$ with $|\tilde S|\le (L\log^2 n/\eps)$ can be large, and for each such set, computing $\adv(S)$ takes $\Omega(n)$ queries since we need to try all Steiner vertuces. To get around this obstacle, we will compute, for each $i$, an estimate $B_i$ of $A_i$ as follows. Let $\hat \sset^i_g$ be the collection of all level-$i$ good sets $S$ with $|\tilde S|\le (L\log^2 n)/\eps$.
We sample $(\log n)/\eps^{10}$ sets in $\hat \sset^i_g$.
For each sampled set $S\in \hat \sset^i_g$, we first query all distances between any terminal in $\tilde S$ and any vertex in $V\setminus T$. We then try all four-terminal sets $Y$, such that $Y\subseteq \tilde S$ and $Y$ represents $S$ (note that there are at most $O((L\log^2 n/\eps)^4)$ such sets), and compute $\adv(S,Y)$ using the acquired distance information. We then let $\adv (S)$ be the maximum over all values $\adv(S,Y)$ that we computed.
We then let $B_i$ be the sum of $\adv (S)$ for all sampled sets $S$ such that $\adv (S)>(\eps^{3/4}/2)\cdot (1+\eps)^i$, namely
$$B_i=\sum_{S\text{ sampled}}\bigg((\eps^{3/4}/2)\cdot(1+\eps)^i\bigg)\cdot \mathbf{1}\big[\adv(S)\ge (\eps^{3/4}/2)\cdot(1+\eps)^i\big].$$ 
Finally, we compute
$\sum_{0\le i\le L}B_i\cdot |\hat \sset^i_g|/(\log n/\eps^{10})$. If it is greater than $5\eps^{3/4}\cdot \mst$, then we return $(1-\eps_0)\cdot\mst$ as an estimate of $\stcost(V,T,w)$. Otherwise, we return $\mst$ as an estimate of $\stcost(V,T,w)$. This completes the description of Step 3, and also completes the description of the whole algorithm.
In Step $3$, we performed in total $\big((\log n)/\eps^{10}\big)\cdot O((L\log^2 n/\eps)^4)\cdot O(n)=\tilde O(n)$ queries. Overall, the algorithm performs $\tilde O(n^{3/2}+n^{3/4}k^{})+\tilde O(n)=\tilde O(n^{3/2}+n^{3/4}k^{})$ queries.

\subsubsection*{Analysis of Step 3 when the algorithm returns $(1-\eps_0)\cdot\mst$ as the estimate of $\stcost(V,T,w)$} 

We now show that, if we collected enough local evidence in this step, then indeed $\stcost(V,T,w)$ is bounded away from $\mst$. Specifically, we will show that, with high probability, if $\sum_{0\le i\le L}B_i\cdot |\hat \sset^i_g|/(\log n/\eps^{10})>5\eps^{3/4}\cdot \mst$, then $\stcost(V,T,w)\le (1-\eps_0)\cdot\mst$.
Since $\stcost(V,T,w)\ge \mst/2$, this implies that our estimate in this case is indeed a $(2-2\eps_0)$-approximation of $\stcost(V,T,w)$.

We first show that, if $\sum_{0\le i\le L}A_i>(2\eps_0)\cdot \mst$ holds, then $\stcost(V,T,w)\le (1-\eps_0)\cdot \mst$. 
Note that we have $\sum_{0\le i\le L}A_i=\sum_{i\text{ even}}A_i+\sum_{i\text{ odd}}A_i$. We assume that $\sum_{i\text{ even}}A_i>\eps_0\cdot \mst$ (the case where $\sum_{i\text{ odd}}A_i>\eps_0\cdot \mst$ is symmetric). Consider now the minimum spanning tree $\tau^*$ computed in Step 1. We will iteratively modify $\tau^*$ by processing good sets in $\sset_g$ and eventually obtain a \St of instance $(V,T,w)$, such that the total cost decreases by at least $\sum_{i\text{ even}}A_i$.

We now formally describe the iterative modification process. 
Throughout the process, we will maintain a \St $\tau$ of instance $(V,T,w)$, that is initialized to be $\tau^*$.
Let $\sset^{e}_g$ be the set of all good sets $S$ at an even-index level, such that $|\tilde S|\le (L\log^2 n)/\eps$. In each iteration, we pick a set $S\in \sset^{e}_g$ that is at the lowest level, and after processing $S$ we discard it from $\sset^{e}_g$. We now describe the iteration of processing set $S$. Let sets $S_1,S_2,S_{11},S_{12},S_{21},S_{22}$ be defined as before. Before this iteration, each of these six sets induce a connected subgraph of the current tree $\tau$, and they are connected in $\tau$ by an edge $e_1$ connecting $S_{11}$ to $S_{12}$, an edge $e_2$ connecting $S_{21}$ to $S_{22}$, and an edge $e$ connecting $S_{1}$ to $S_{2}$. Clearly, the total cost of these three edges is at most $w^*(S)$, by the definition of $w^*(S)$. Consider now the set $Y$ that represents $S$ such that $\adv(S)=\adv(S,Y)$. Let $E'_Y$ be the set of edges in the \St that achieves $\adv(S,Y)$. In this iteration we simply replace edges $e,e_1,e_2$ with edges in $E'_Y$. It is easy to see that the resulting tree $\tau$ is still a \St of instance $(V,T,w)$, and the decrease of total weight is $w^*(S)-w(E'_Y)=\adv(S,Y)=\adv(S)$. Therefore, after processing all sets in $\sset^{e}_g$ in this way, we obtain a \St of instance $(V,T,w)$ with total cost at most $w(\tau^*)-\sum_{S\in \sset^{e}_g}\adv(S)=w(\tau^*)-\sum_{i\text{ even}}A_i\le (1-\eps_0)\cdot \mst$. This shows that $\stcost(V,T,w)\le (1-\eps_0)\cdot \mst$.

We now show that, if $\sum_{0\le i\le L}B_i\cdot |\hat \sset^i_g|/(\log n/\eps^{10})>5\eps^{3/4}\cdot \mst$, then $\sum_{0\le i\le L}A_i>(2\eps_0)\cdot \mst$ holds, completing the analysis of Step 3 when the output is $(1-\eps_0)\cdot\mst$. 
We start with the following simple observation.

\begin{observation}
\label{obs: rep good for adv}
For each good set $S\in \sset^i_g$ and each set $Y$ that represents $S$, there exists a set $\tilde Y\subseteq \tilde S$ that represents $S$, such that $\adv(S,Y)\le \adv(S,\tilde Y)+8\eps\cdot (1+\eps)^i$.
\end{observation}
\begin{proof}
Let $Y=\set{u_{11},u_{12},u_{21},u_{22}}$. Recall that $\tilde S$ is a maximal subset of $S$, such that the distance between every pair of terminals in $\tilde S$ is at least $\eps\cdot (1+\eps)^i$. So there exist vertices $\tilde u_{11}\in S_{11}$, such that $w(u_{11},\tilde u_{11})\le \eps\cdot (1+\eps)^i$.
Similarly, there exist vertices
$\tilde u_{12}\in S_{12}$, $\tilde u_{21}\in S_{21}$, $\tilde u_{22}\in S_{22}$ that are close to $u_{12},u_{21},u_{22}$, respectively. We simply let $\tilde Y=\set{\tilde u_{11},\tilde u_{12},\tilde u_{21},\tilde u_{22}}$. Assume that the set $Y$ achieves the cost $\adv(S,Y)$ via Steiner vertex $v$ and tree $\tau$. It is easy to observe that by replacing vertex $u_{ij}$ with $\tilde u_{ij}$, we obtain another Steiner tree $\tilde\tau$ that achieves advantage at least $\adv(S,Y)-8\eps\cdot (1+\eps)^i$. \Cref{obs: rep good for adv} now follows.
\end{proof}

Consider now the collection $\hat \sset^i_g$. Let $\sset'\subseteq \hat \sset^i_g$ contain all sets $S\in \hat \sset^i_g$ with $\adv(S)\ge \eps^{3/4}(1+\eps)^i$, and let $\sset''$ contain other sets. Since we have sampled $\log n/\eps^{10}$ sets in $\hat \sset^i_g$, from Chernoff Bound, 
\begin{itemize}
\item if $|\sset'|\ge \eps\cdot |\hat \sset^i_g|$, then with probability $(1-n^{-10})$ the number of sampled sets in $\sset'$ is within factor $(1+\eps)$ from $(\log n/\eps^{10})|\sset'|/|\hat \sset^i_g|$, so \[\frac{|\hat \sset^i_g|}{(\log n/\eps^{10})}\cdot\sum_{S\text{ sampled}, S\in \sset'}(\eps^{3/4}/2)\cdot(1+\eps)^i\le (1+\eps)\cdot \sum_{S\in \sset'}\adv(S);\]
\item if $|\sset'|< \eps\cdot |\hat \sset^i_g|$, then then with probability $(1-n^{-10})$, the number of sampled sets in $\sset'$ is at most $10\log n/\eps^9$, \[\frac{|\hat \sset^i_g|}{(\log n/\eps^{10})}\cdot\sum_{S\text{ sampled}, S\in \sset'}(\eps^{3/4}/2)\cdot(1+\eps)^i\le (10\eps)\cdot (\eps^{3/4}/2)\cdot(1+\eps)^i\cdot |\hat \sset^i_g|\le \eps\cdot w_i(\tau^*).\]
\end{itemize}
On the other hand, it is clear that
\[\frac{|\hat \sset^i_g|}{(\log n/\eps^{10})}\cdot\sum_{S \text{ sampled, } S\in \sset''}(\eps^{3/4}/2)\cdot(1+\eps)^i\le \eps^{3/4}\cdot w_i(\tau^*).
\]
Altogether, we get that, with probability $1-O(n^{-10})$, $A_i\ge (B_i-\eps^{3/4} \cdot w_i(\tau^*))/2$. Taking the union bound over all $0\le i\le L$, we get that, with probability $1-O(n^{-9})$,
$\sum_{0\le i\le L}A_i\ge \sum_{0\le i\le L}(B_i-\eps^{3/4} \cdot w_i(\tau^*))/2$. Therefore, if $\sum_{0\le i\le L}B_i\ge 5\eps^{3/4}\cdot \mst$, then $\sum_{0\le i\le L}A_i\ge 2\eps^{3/4}\cdot \mst\ge 2\eps_0\cdot \mst$.
This completes the analysis of Step 3 when the output is $(1-\eps_0)\cdot\mst$.

$\ $

Note that, if the algorithm did not return $(1-\eps_0)\cdot\mst$, then according to the algorithm, $\sum_{0\le i\le L}B_i< 5\eps^{3/4}\cdot \mst$. From the definition of $A_i$ and $B_i$ and similar arguments in the above analysis, we get that with high probability, $\sum_{0\le i\le L}A_i< 5\eps^{1/2}\cdot \mst$. Then from similar arguments in \Cref{obs: small sets enough}, we can then show that $\sum_{S\in \sset_g}w^*(S)\le O(\mst/\log n)$, and $\sum_{S\in \sset_g}\adv(S)\le 6\eps^{1/2}\cdot \mst$.

\subsection*{Analysis when the algorithm returns $\mst$ as the estimate of $\stcost(V,T,w)$}

Lastly, we show that, if the algorithm did not collect enough local evidence in Step 2 and Step 3, then $\stcost(V,T,w)$ is indeed bounded away from $\mst/2$. Specifically, we show that, if the algorithm returns $\mst$ as the estimate, then $\stcost(V,T,w)\ge \mst/(2-2\eps_0)$, and so in this case the estimate is indeed a $(2-2\eps_0)$-approximation of $\stcost(V,T,w)$.

Before diving into the details, we give some intuition. Consider the optimal \St $\tau$, and we will iteratively remove Steiner vertices from it such that eventually it becomes a spanning tree over terminals. In each iteration, we will try to replace some set $E$ of edges in the current tree with another set $E'$ of terminal-terminal edges, such that $w(E')\le (2-\Omega(\eps_0))\cdot w(E)$ holds and the resulting graph is still a Steiner Tree. Intuitively, if we cannot find sufficient local evidence in Step 2, then most Steiner vertices in $\tau$ can be eliminated such that the resulting tree $\tau$ satisfies that $w(\tau)\cdot (2-\Omega(\eps_0))\le \mst$.
However, it is also possible that $\tau$ behaves in a similar way as $\wy$ defined in \Cref{sec: 5/3-lower} and we cannot find Steiner vertices to eliminate in $\tau$. In this case, we will replace some set $E$ of edges in the current tree with a set $E'$ of terminal-terminal edges, such that $w(E')\le 2\cdot w(E)$ holds and simultaneously construct a set of four vertices that represents some set in $\sset$ as defined in Step 3, and achieves cost reduction comparable to $w(E')$. Eventually, we will collect sufficient four-vertex sets, indicating that the local evidence that should have been found by the $4$-vertex subroutine is large, contradicting the outcome of Step 3.

We now provide the complete proof.
Let $\tau_{\opt}$ be an optimal solution of instance $(V,T,w)$. We will iteratively modify tree $\tau_{\opt}$, such that eventually we obtain a spanning tree on $T$ whose total weight is at most $(2-2\eps_0)$ times the weight of $\tau_{\opt}$. Since such a tree has total weight at least $\mst$, we get that $w(\tau_{\opt})\ge \mst/(2-2\eps_0)$.

We now describe the tree-modification process.  Throughout, we maintain a \St $\tau$ of instance $(V,T,w)$, that is initialized to be $\tau_{\opt}$. 
Note that we can assume without loss of generality that $\tau_{\opt}$ does not contain degree-$2$ Steiner vertices (since such vertices can be suppressed without increasing the total weight of the tree), and whenever degree-$2$ Steiner vertices emerge in tree $\tau$, we immediately suppress them. We will also maintain two collections $\xset,\yset$ of sets of terminals in $T$, such that both $\xset$ and $\yset$ initially contain no sets, and (i) every set added into $\xset$ has size at least $3$, and will be denoted by $X_i(v)$ for some integer $0\le i\le L-1$ and some Steiner vertex $v$; and (ii) every set added into $\yset$ has size exactly $4$. 
In each iteration, we distinguish between the following cases.

\vspace{+8pt}

\textbf{Case 1.} {$\tau$ contains a Steiner vertex that is adjacent to at least three terminals.} Let $v$ be such a vertex. 
Let $u_1,\ldots,u_t$ be the terminals that are adjacent to $v$, such that distances $w(v,u_1)\le w(v,u_2)\le \cdots\le w(v,u_t)$. Intuitively, if the distances differ significantly, then we can replace the heaviest edge with a terminal-terminal edge; if the distances are almost the same and close to half of terminal-terminal distances, then they should contribute a set to the Set Cover instance on this level (defined in Step 2); if the distances are almost the same and significantly greater than half of terminal-terminal distances, then they can all be replaced by terminal-terminal edges, with the total weight increasing by a factor at most $(2-\Omega(\eps_0))$. Specifically, we distinguish between the following two cases.

\textbf{Case 1.1.} There exists a pair $i,j$ of indices such that $w(u_i,u_j)\le (2-4\eps_0)\cdot w(v,u_j)$. In this case, we simply replace the edge $(v,u_j)$ in $\tau$ with edge $(u_i,u_j)$. See \Cref{fig: case11} for an illustration.

\textbf{Case 1.2.} For every pair $i,j$ of indices, $w(u_i,u_j)> (2-4\eps_0)\cdot w(v,u_j)$. Denote $\ell=w(u,v_1)$. Observe that, in this case, $w(v,u_t)\le (1+5\eps_0)\cdot \ell$ must hold, since otherwise $w(u_1,u_t)\le w(v,u_1)+w(v,u_t)\le (1+\frac{1}{1+5\eps_0})\cdot w(v,u_t)\le  (2-4\eps_0)\cdot w(v,u_t)$, a contradiction to the assumption in this case. Also observe that, for every pair $i,j$, $w(u_i,u_j)> (2-4\eps_0)\cdot \ell$. We denote $i^*=\floor{\log_{1+\eps}\ell}$, and then define the set $X_{i^*}(v)=\set{u_1,\ldots,u_t}$ and add it into $\xset$ (note that $t\ge 3$). We then replace, for each $2\le i\le t$, edge $(v,u_i)$ with edge $(u_1,u_i)$. See \Cref{fig: case12} for an illustration.

\begin{figure}[h]
\centering
\subfigure[Case 1.1: before (left) and after (right).]{\scalebox{0.1}{\includegraphics{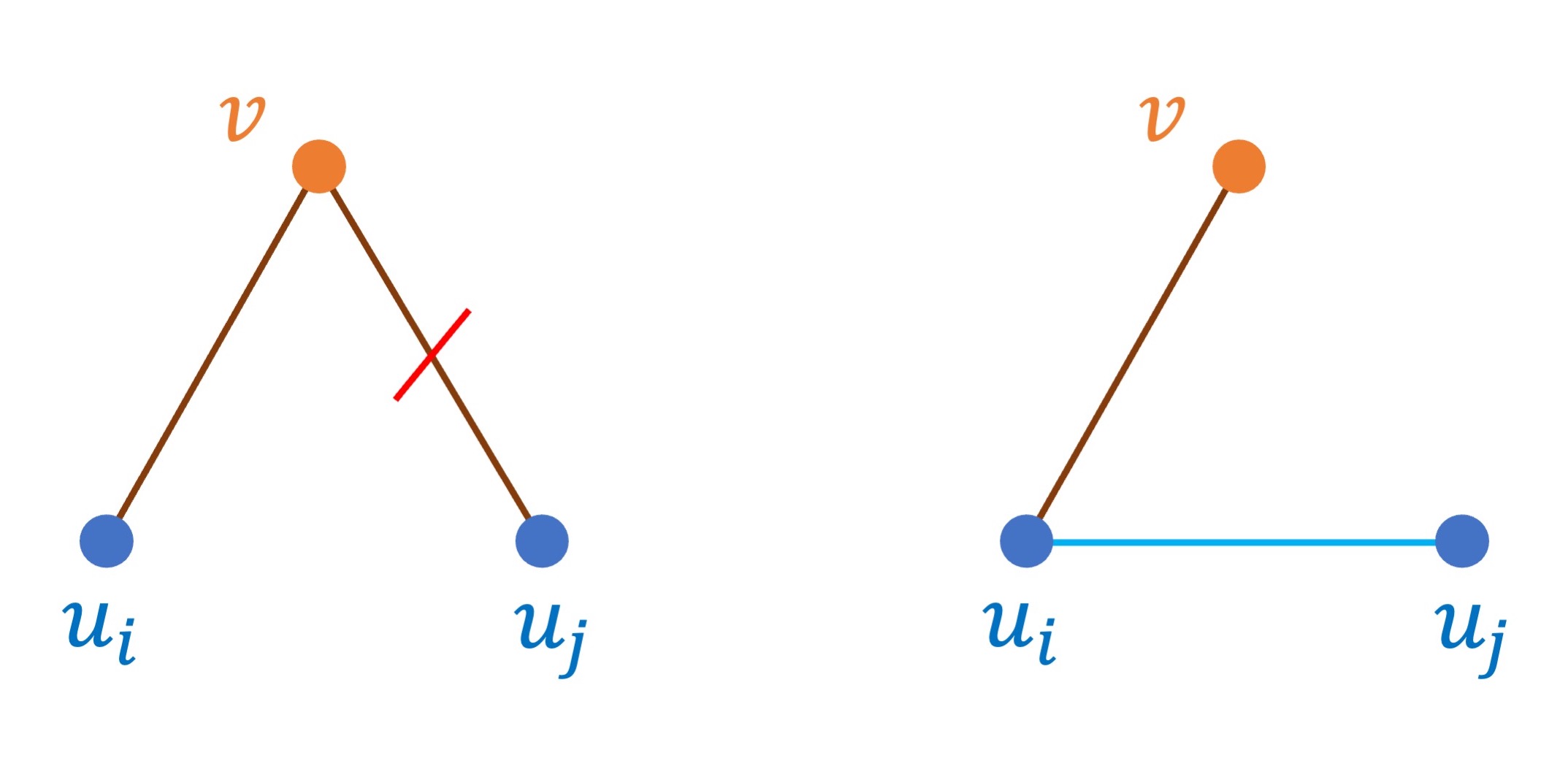}}\label{fig: case11}}
\hspace{0.5cm}
\subfigure[Case 1.2: deleted edges (red) and new edges (blue).]{
\scalebox{0.11}{\includegraphics{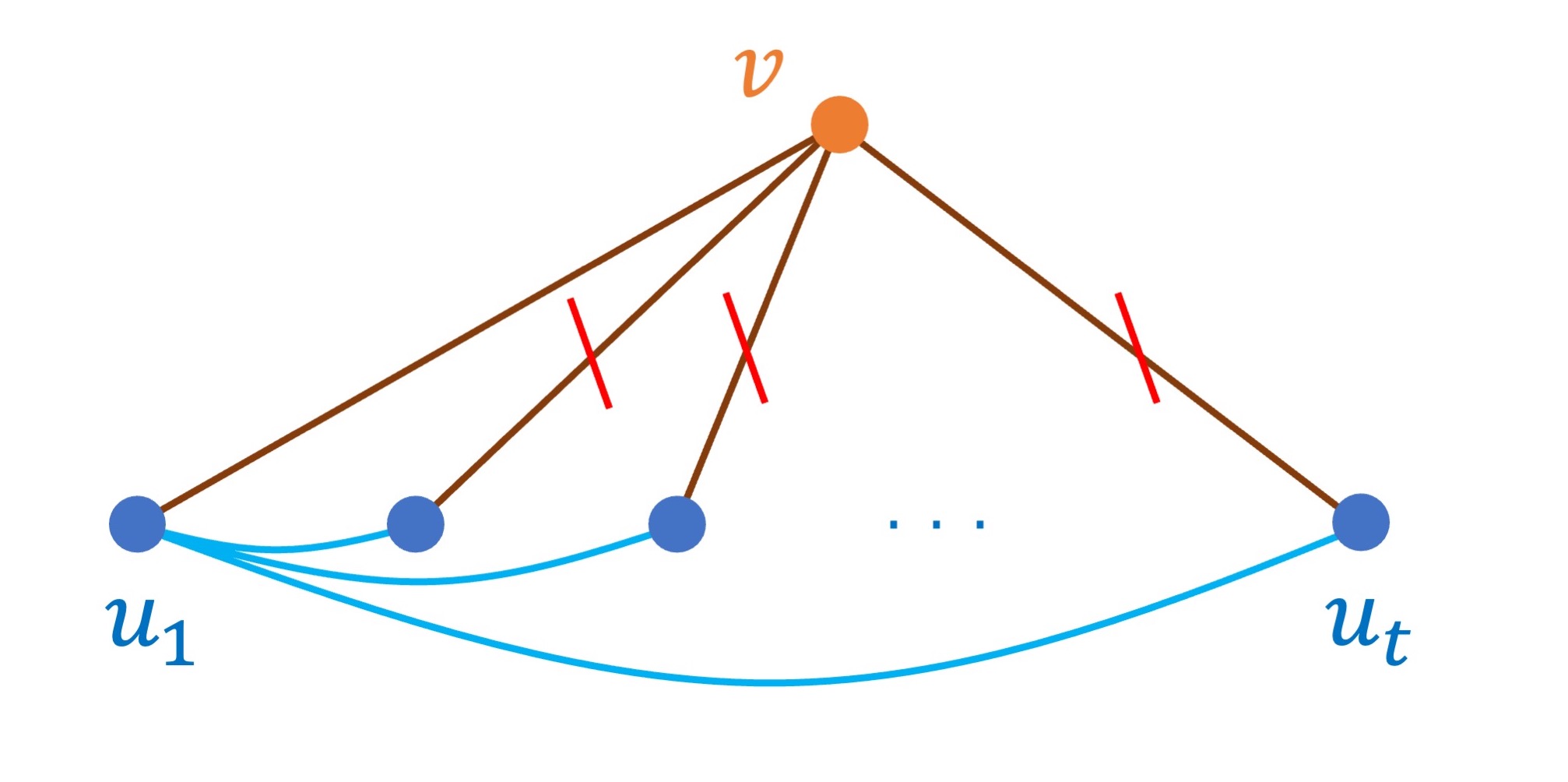}}\label{fig: case12}}
\caption{An illustration of edge replacement in Case 1.\label{fig: case_1}}
\end{figure}

Assume now that Case 1 does not happen, so every Steiner vertex is adjacent to at most two terminals. We root tree $\tau$ at an arbitrary Steiner vertex. Since $\tau$ does not contain degree-$2$ Steiner vertices, every height-$1$ Steiner vertex is incident to exactly two terminals (since otherwise it is either a leaf or a degree-$2$ Steiner vertex, a contradiction). 

Consider now any Steiner vertex $v$ of height $2$ in $\tau$. Let $u'_1,\ldots,u'_p$ be the terminals that $v$ is adjacent to, let $v_1,\ldots,v_t$ be the  height-$1$ Steiner vertices adjacent to $v$, and for each $1\le j\le t$, let $u^j_1,u^j_2$ be the two terminals adjacent to $v_j$, such that $w(v_j,u^j_1)\le w(v_j,u^j_2)$. See \Cref{fig: case2} for an illustration.
Since $v$ is not a degree-$2$ Steiner vertex in $\tau$, either $t\ge 2$, or $t= 1$ and $p\ge 1$. 

\begin{figure}[h]
	\centering
\includegraphics[scale=0.11]{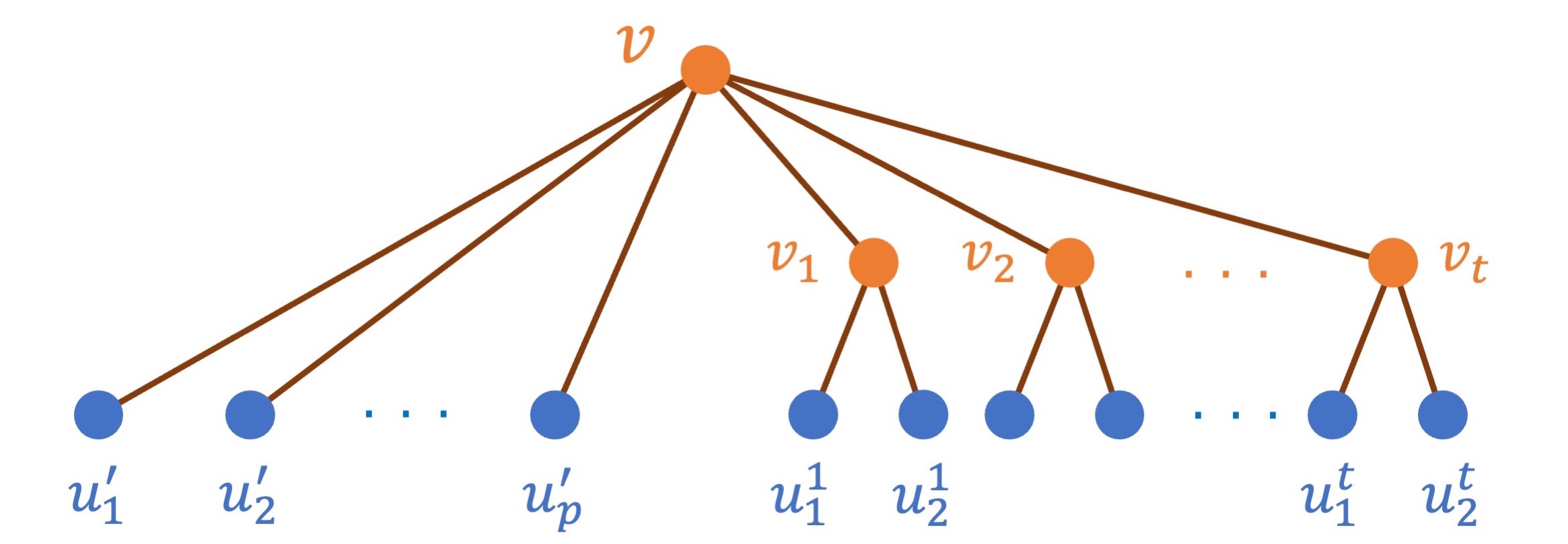}
	\caption{A schematic view of vertices and edges in Case 2.\label{fig: case2}}
\end{figure}

\textbf{Case 2.} {The tree-distances in $\tau$ between $v$ and terminals $u'_1,\ldots,u'_p, u^1_1,u^1_2,\ldots,u^t_1,u^t_2$ are not all within factor $(1+O(\eps_0))$.} Intuitively, in this case the subtree of $\tau$ rooted at $v$ is not balanced enough, and so we can always find some terminal-Steiner edge to replace with a terminal-terminal edge. In particular, we distinguish between the following six cases.

\textbf{Case 2.1.} There exists a pair $1\le i,j\le p$ such that $w(u'_i,u'_j)\le (2-4\eps_0)\cdot w(v,u'_j)$. Similar to Case 1.1, we replace edge $(v,u'_j)$ in $\tau$ with edge $(u'_i,u'_j)$ (see \Cref{fig: case21}).

\textbf{Case 2.2.} There exists an index $1\le j\le t$ such that $w(u^j_1,u^j_2)\le (2-4\eps_0)\cdot w(v,u^j_2)$. We replace edge $(v,u^j_2)$ in $\tau$ with edge $(u^j_1,u^j_2)$ (see \Cref{fig: case22}).

\textbf{Case 2.3.} There exist $1\le i\le p$, $1\le j\le t$ and $z\in \set{1,2}$, such that $w(u'_i,u^j_z)\le (2-4\eps_0)\cdot w(v,u'_i)$. We replace edge $(v,u'_i)$ in $\tau$ with edge $(u'_i,u^j_z)$ (see \Cref{fig: case23}).

\textbf{Case 2.4.} There exist indices $1\le i\le p$, $1\le j\le t$ and $z\in \set{1,2}$, such that $w(u'_i,u^j_z)\le (2-8\eps_0)\cdot \big(w(v,v_j)+w(v_j,u^j_z)\big)$. We replace edges $(v,v_j),(v_j,u^j_1),(v_j,u^j_2)$ in $\tau$ with edges $(u^j_1,u^j_2)$ and $(u'_i,u^j_z)$ (see \Cref{fig: case24}).

\textbf{Case 2.5.} There exist $1\le j,j'\le t$ and $z,z'\in \set{1,2}$, such that $w(u^j_z,u^{j'}_{z'})\le (2-8\eps_0)\cdot \big(w(v,v_j)+w(v_j,u^j_z)\big)$. We replace edges $(v,v_j),(v_j,u^j_1),(v_j,u^j_2)$ in $\tau$ with edges $(u^j_1,u^j_2)$ and $(u'_i,u^j_z)$, edge $(v_j,u^j_2)$ with edge $(u^j_1,u^j_2)$ (see \Cref{fig: case25}).

\textbf{Case 2.6.} There exist $1\le j\le t$ and $z\in \set{1,2}$, such that $w(v,u^j_z)\le (1-4\eps_0)\cdot \big(w(v,v_j)+w(v_j,u^j_z)\big)$. We replace edges $(v,v_j),(v_j,u^j_1),(v_j,u^j_2)$ in $\tau$ with edges $(u^j_1,u^j_2)$ and $(v,u^j_z)$.

\begin{figure}[h]
	\centering
	\subfigure[Case 2.1.]{\scalebox{0.09}{\includegraphics{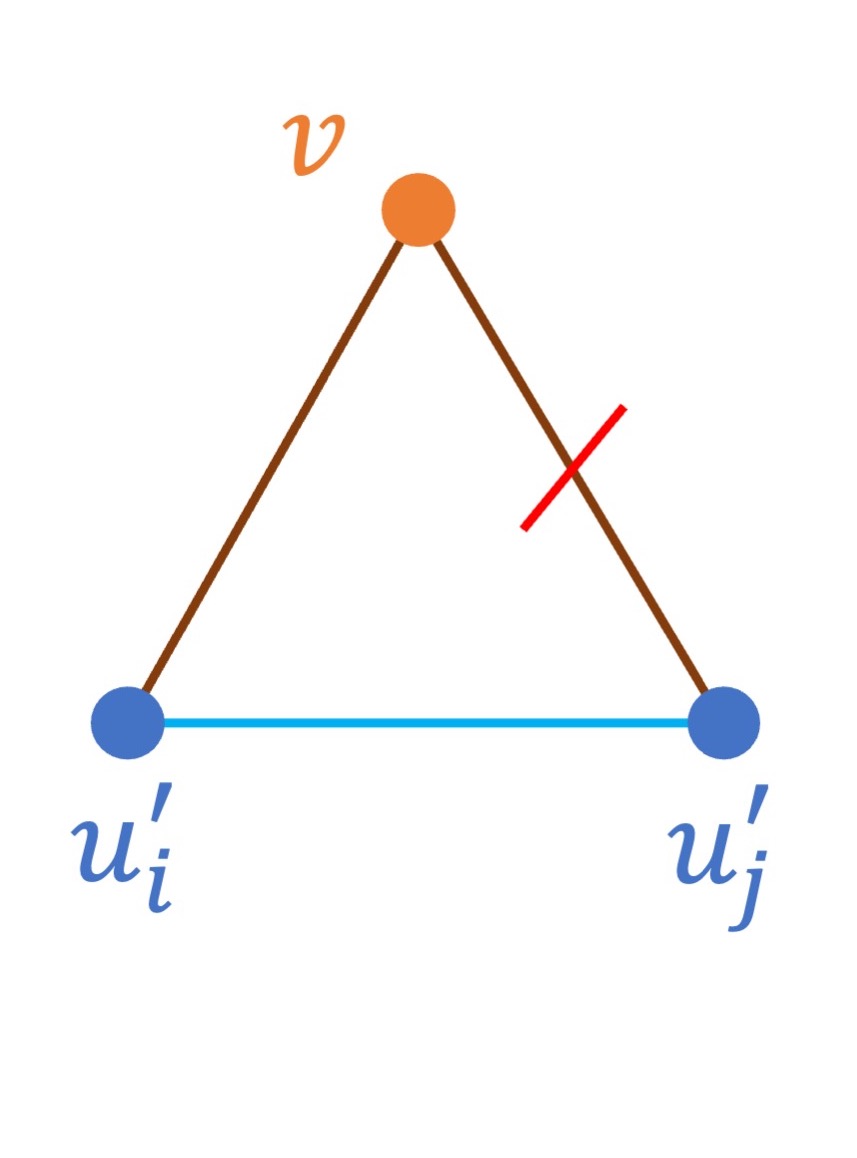}}\label{fig: case21}}
	\hspace{0.1cm}
	\subfigure[Case 2.2.]{
		\scalebox{0.09}{\includegraphics{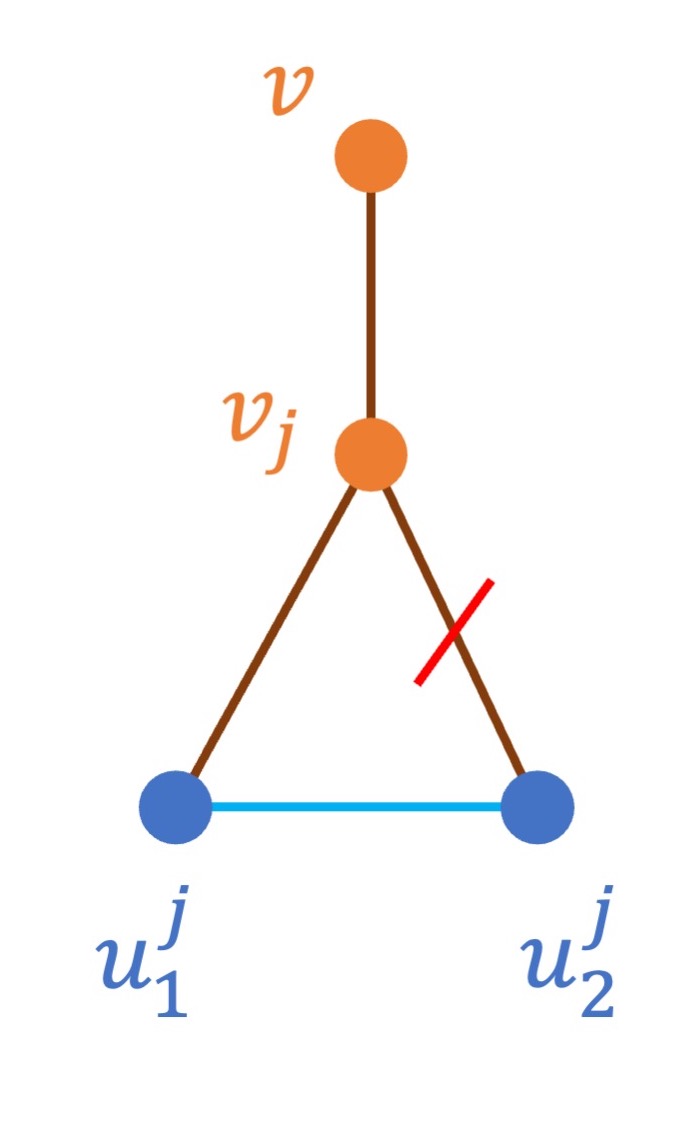}}\label{fig: case22}}
	\hspace{0.1cm}
	\subfigure[Case 2.3.]{
		\scalebox{0.09}{\includegraphics{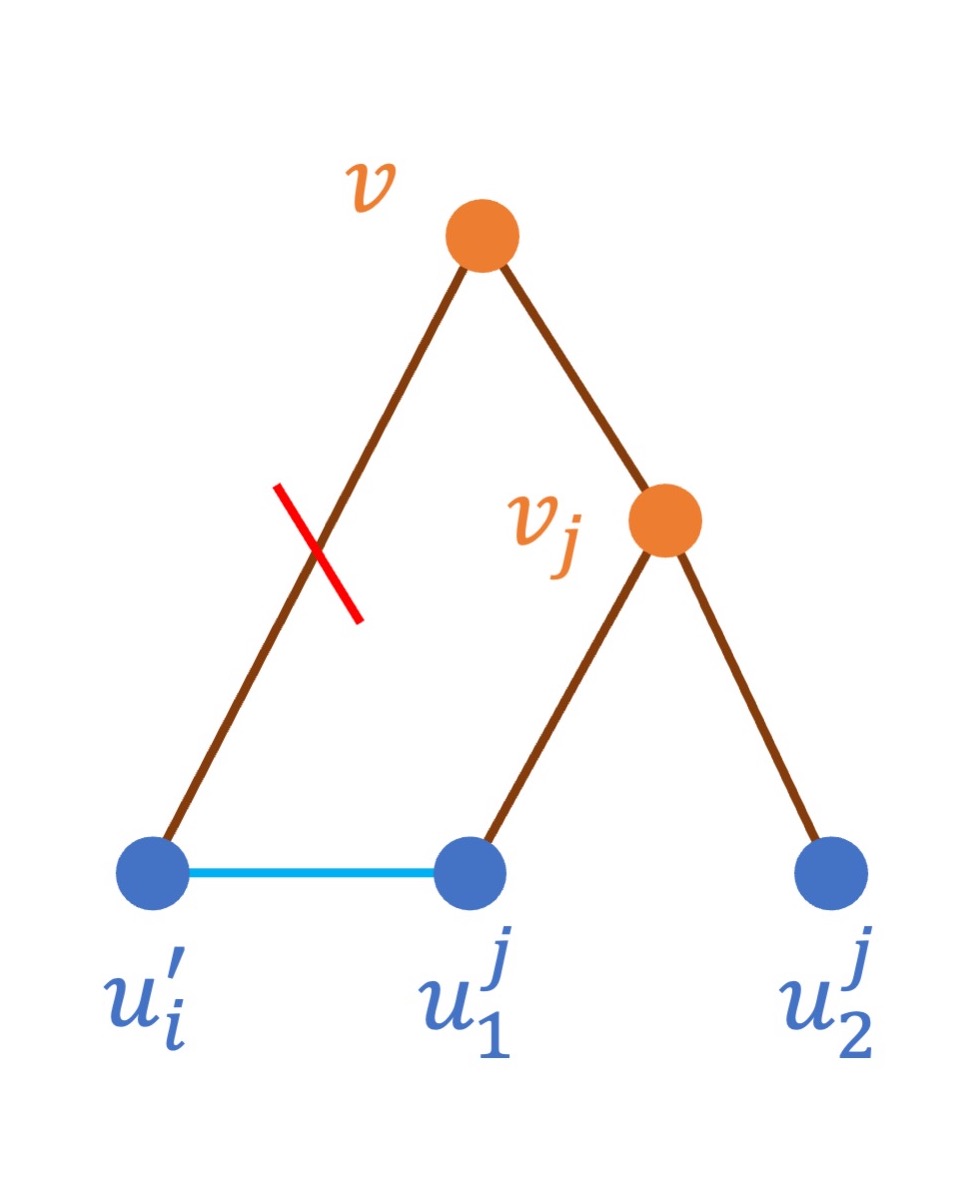}}\label{fig: case23}}
	\hspace{0.1cm}
	\subfigure[Case 2.4.]{
		\scalebox{0.09}{\includegraphics{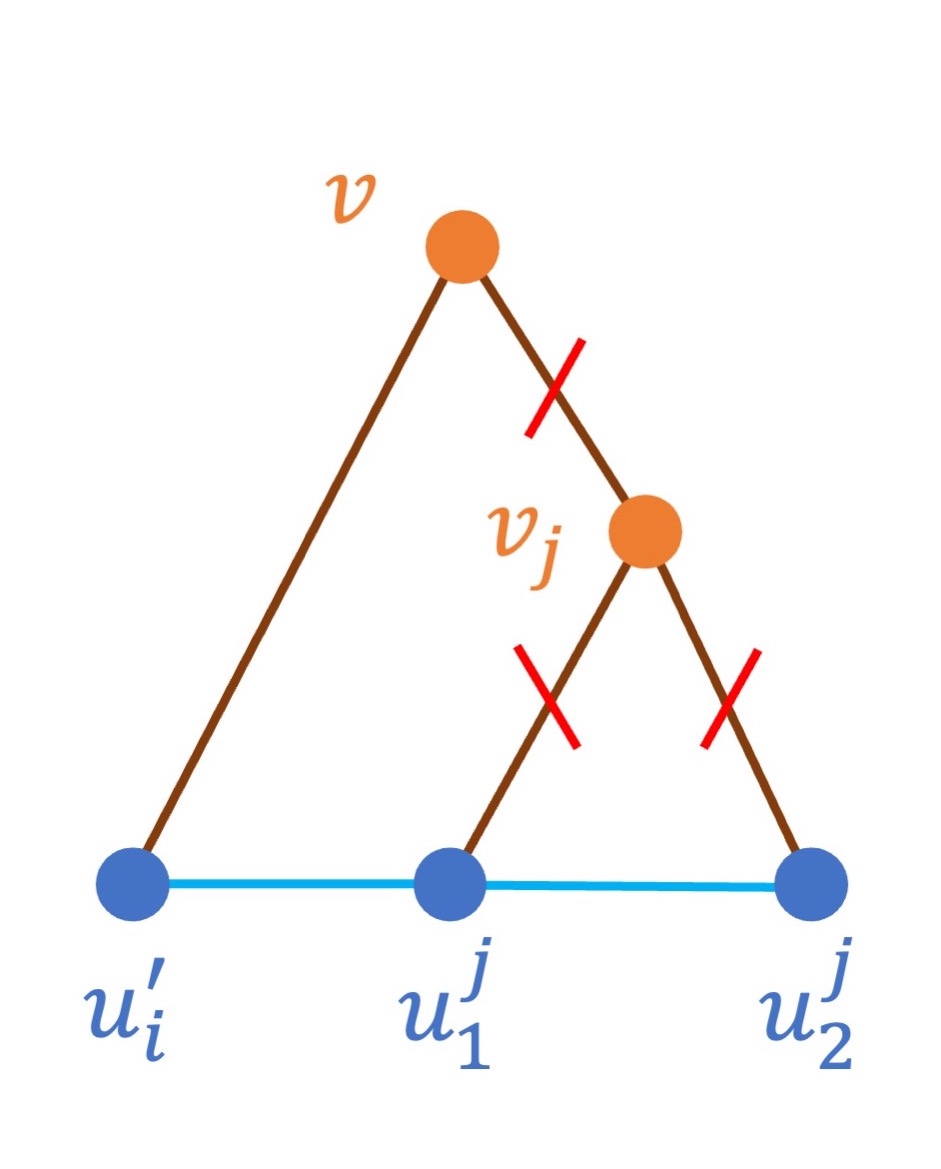}}\label{fig: case24}}
	\hspace{0.1cm}
	\subfigure[Case 2.5.]{
		\scalebox{0.09}{\includegraphics{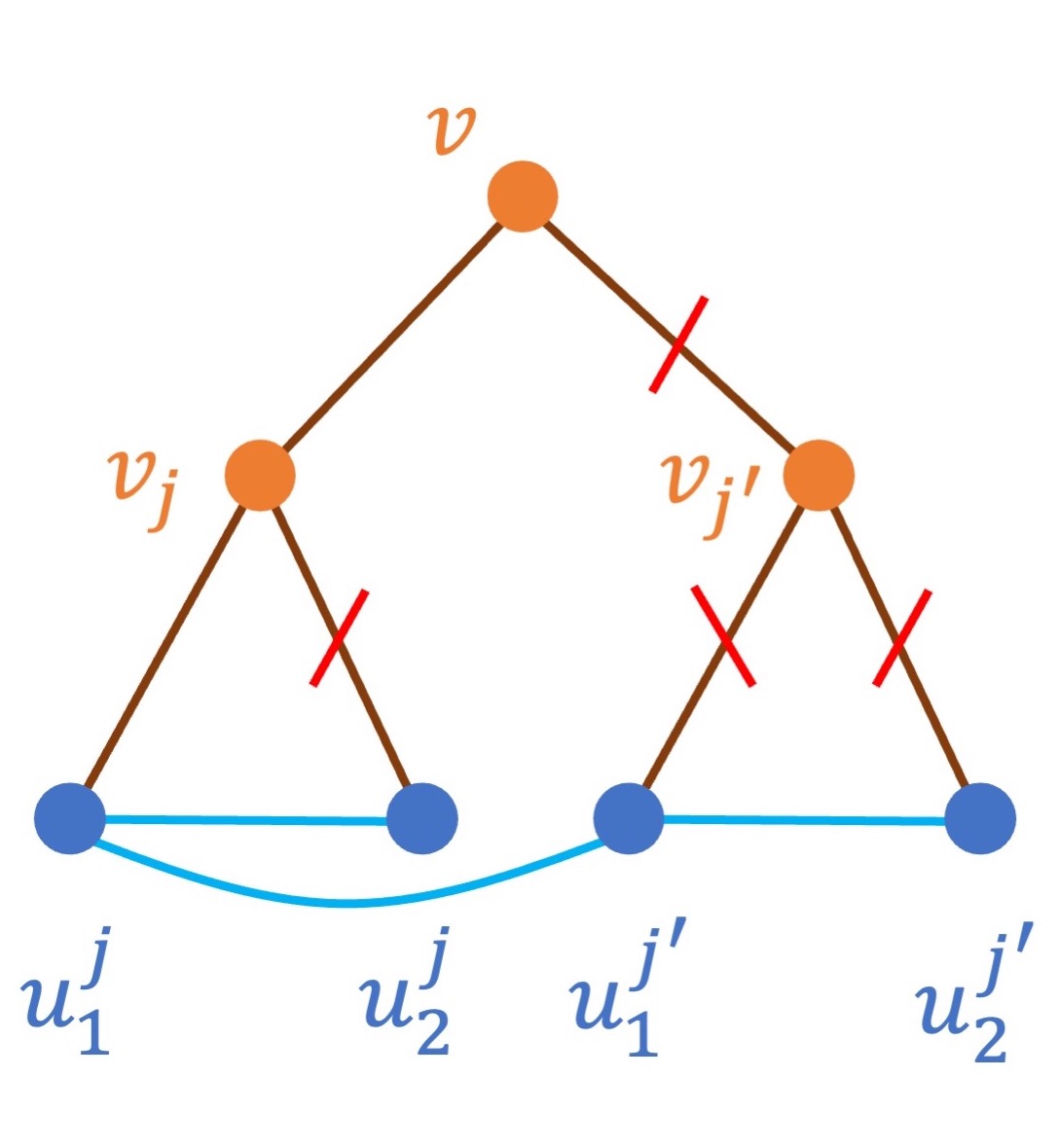}}\label{fig: case25}}
	\caption{An illustration of edge replacement in Case 2.\label{fig: case_2}}
\end{figure}

$\ $

Assume that Case 1 and 2 do not happen.
We denote $U=\set{u'_1,\ldots,u'_p, u^1_1,u^1_2,\ldots,u^t_1,u^t_2}$.
From the discussion in Case 2, it is easy to observe that the tree-distances in $\tau$ between $v$ and terminals $U$ are within factor $(1+20\eps_0)$ from each other.
Denote $\ell=\min\set{w(v,u)\mid u\in U}$. 
So for every terminal $u\in U$, $\ell\le w(v,u)\le (1+20\eps_0)\cdot \ell$. We denote $i^*=\floor{\log_{1+\eps}\ell}$, and let $u^*=\arg\min_{u\in U}\set{w(v,u)\mid u\in U}$.

\begin{figure}[h]
	\centering
	\includegraphics[scale=0.12]{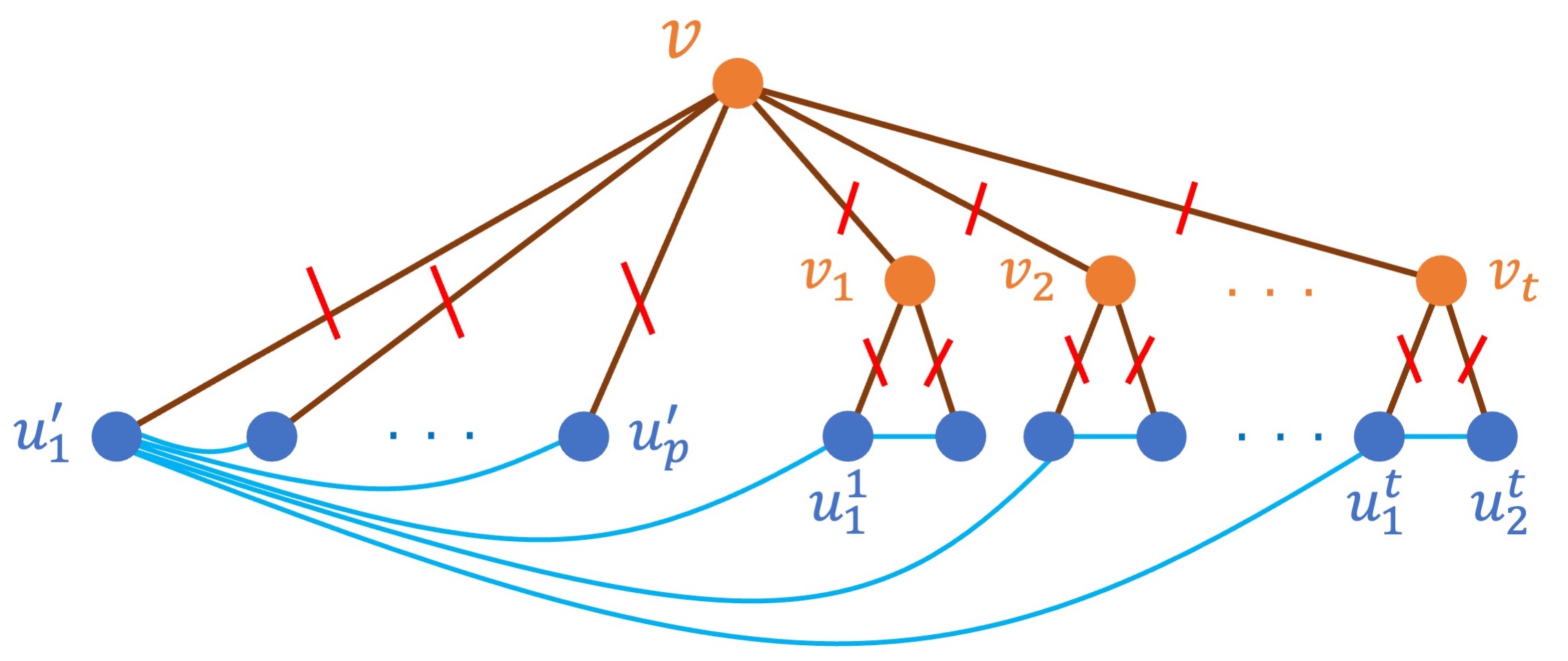}
	\caption{An illustration of edge replacement in Case 3.1 (assume that $u^*=u'_1$).\label{fig: case31}}
\end{figure}

\textbf{Case 3.} We say Case 3 happens if one of the following subcases happen.

\textbf{Case 3.1.} $p+t\ge 3$. 
We define the set $X_{i^*}(v)=\set{u'_1,\ldots,u'_p,u^1_1,\ldots,u^1_t}$ and add it into $\xset$.
Then for each $1\le j\le t$, we replace edge $(v_j,u^j_2)$ with $(u^j_1,u^j_2)$, and for each vertex $u\in U\setminus \set{u^*,u^1_2,\ldots,u^t_2}$, we replace edges in the $v$-$u$ path in $\tau$ with edge $(u^*,u)$ (see \Cref{fig: case31}). 

\textbf{Case 3.2.} $p=1$ and $t=1$. This case can be actually viewed as the special case of the next case by letting $v_2=u^2_1=u^2_2=u'_1$.

\textbf{Case 3.3.} $p=0$ and $t=2$. 
We assume without loss of generality that $u^*=u^1_1$. 

\textbf{Case 3.3.1.} We say that Case 3.3.1 happens if one of the following two cases happen:
\begin{itemize}
\item if $w(v,v_1)\le (50\eps_0)\cdot \ell$, then we define set $X_{i^*}(v)=\set{u^1_1,u^1_2,u^2_1}$ and add it into $\xset$;
\item if $w(v,v_2)\le (50\eps_0)\cdot \ell$, then we define set $X_{i^*}(v)=\set{u^1_1,u^2_1,u^2_2}$ and add it into $\xset$. 
\end{itemize}
In addition, in the above cases, we replace edge $(v_2,u^2_2)$ with $(u^2_1,u^2_2)$, edges $(v_2,v),(v_2,u^2_1)$ with edge $(u^2_1,u^1_1)$, and edge $(v_1,u^1_2)$ with edge $(u^1_1,u^1_2)$ (see \Cref{fig: case331}). 

\textbf{Case 3.3.2.}
Consider now the laminar family $\sset$ computed in Step 1. 
We say that a set $U'$ of terminals in $T$ is \emph{interfered} iff there is a set $S$ in $\sset$, such that both $S\setminus U', U'\setminus S\ne \emptyset$, and in this case we say that any vertex $u\in S\setminus U'$ is a \emph{witness}.
If the pair $u^1_1,u^1_2$ of terminals are interfered, then let $u$ be a witness, and assume without loss of generality that there exists a set $S\in \sset$ that contains $u^1_1$ and $u$ but not $u^1_2$ 
(the case where the pair $u^2_1,u^2_2$ of terminals are interfered is symmetric). 
If the set $\set{u^1_1,u^1_2,u^2_1,u^2_2}$ is interfered, then let $u$ be a witness, and assume in particular that there exists a set $S\in \sset$ that contains $u^1_1,u^1_2$ and $u$ but not $u^2_1,u^2_2$.
In both cases, we delete vertex $v_1$ and all its incident edges, and add edges $(u^1_1,u^1_2)$ and $(u^1_1,u)$ (see \Cref{fig: case332}).

If Case 3.3.2 does not happen, then the collection $\sset$ computed in Step 1 must contain sets $\set{u^1_1,u^1_2}$, $\set{u^2_1,u^2_2}$ and some set that contains all elements $u^1_1,u^1_2,u^2_1,u^2_2$. Since $\sset$ is a laminar family, there exists a minimum set in $\sset$ that contains all elements $u^1_1,u^1_2,u^2_1,u^2_2$, that we denote by $S$.

\textbf{Case 3.3.3.} $S=\set{u^1_1,u^1_2,u^2_1,u^2_2}$, and
either $w(v,v_1)\le (1-4\eps^{1/4}) \ell$ or $w(v,v_2)\le (1-4\eps^{1/4}) \ell$ holds. 
Assume without loss of generality that $w(v,v_1)\le (1-4\eps^{1/4}) \ell$.
Then we add the set $S$ into $\yset$, and then we replace edge $(v_2,u^2_2)$ with $(u^2_1,u^2_2)$, edges $(v_2,v),(v_2,u^2_1)$ with edge $(u^2_1,u^1_1)$, and edge $(v_1,u^1_2)$ with edge $(u^1_1,u^1_2)$. The illustration figure in this case is identical to that of Case 3.3.1 (see \Cref{fig: case331}).

\textbf{Case 3.3.4.}
$w(v,v_1),w(v,v_2)> (1-4\eps^{1/4})\cdot \ell$. In this case we simply replace edge $(v,u^1_2)$ with $(u^1_1,u^1_2)$ and edge $(v,u^2_2)$ with $(u^2_1,u^2_2)$. We call the operation particularly in this case a \emph{bad replacement}.

\begin{figure}[h]
	\centering
	\subfigure[Case 3.3.1.]{\scalebox{0.115}{\includegraphics{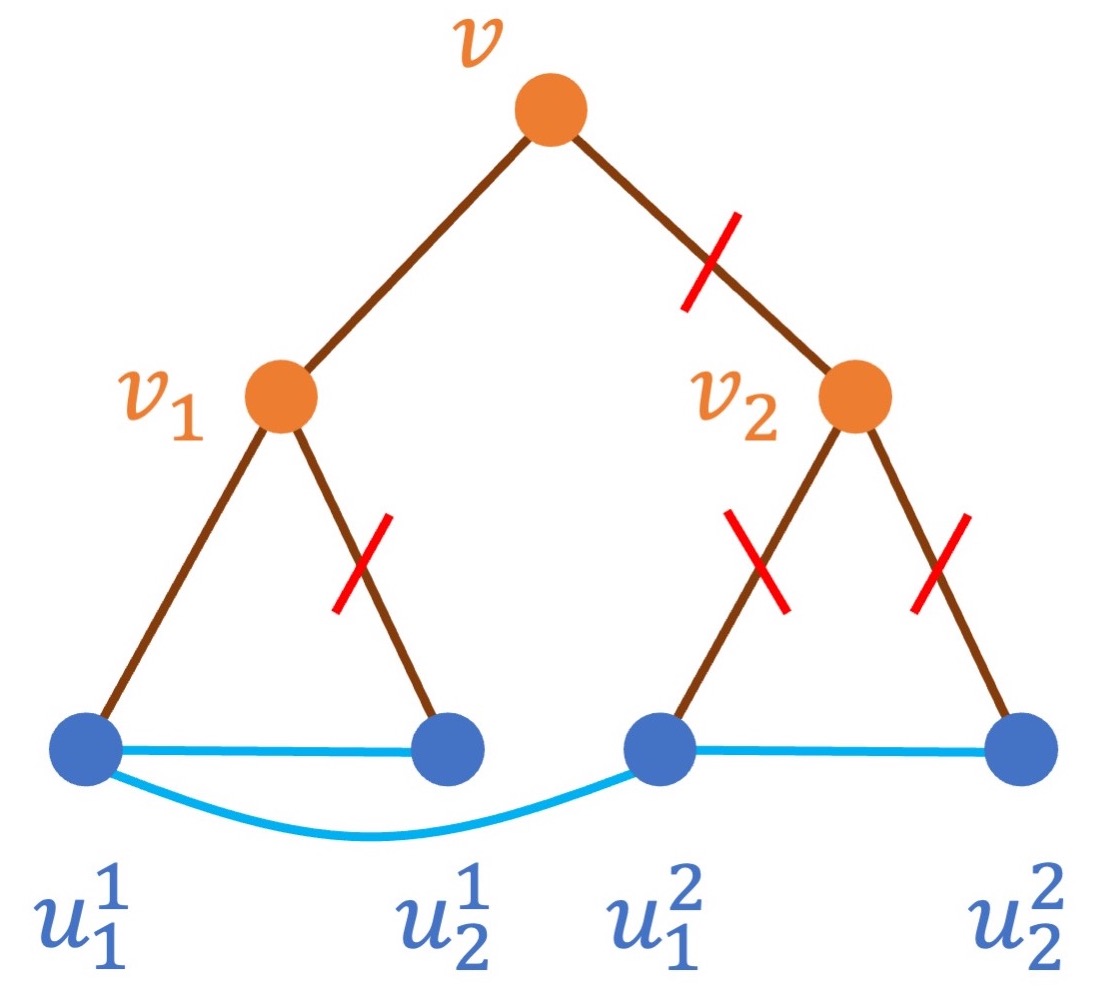}}\label{fig: case331}}
	\hspace{0.5cm}
	\subfigure[Case 3.3.2: old and new edges (left) and part of the tree $\tset$ (right).]{
		\scalebox{0.11}{\includegraphics{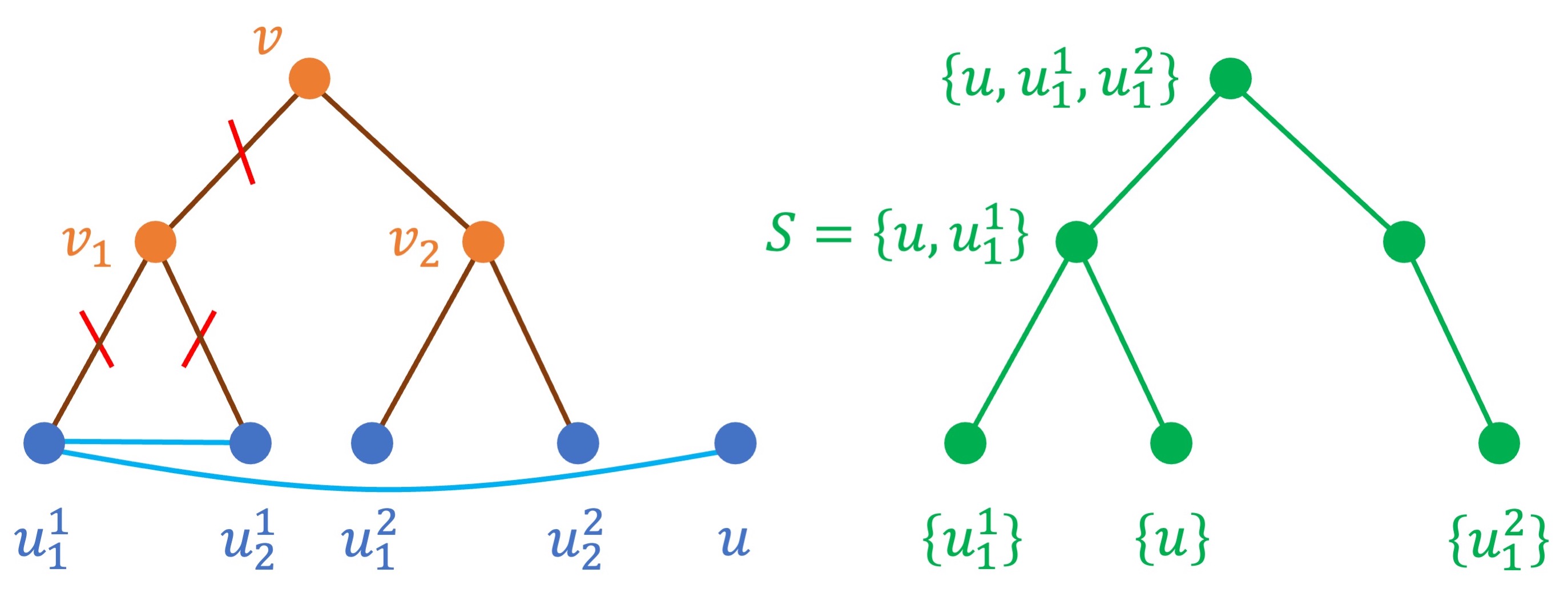}}\label{fig: case332}}
	\caption{An illustration of edge replacement in Case 3.3.1 and Case 3.3.2.\label{fig: case_3312}}
\end{figure}

The only possibility that Cases 2,3 do not happen is when:
\begin{itemize}
\item $v$ has two children $v_1,v_2$; $v_1$ has two children $u^1_1,u^1_2$; and $v_2$ has two children $u^2_1,u^2_2$;
\item $\ell\le w(v,u^1_1),w(v,u^1_2),w(v,u^2_1),w(v,u^2_2)\le (1+20\eps_0)\ell$; and
\item the set $\set{u^1_1,u^1_2,u^2_1,u^2_2}$ of vertices are not interfered in $\sset$, but $S\ne \set{u^1_1,u^1_2,u^2_1,u^2_2}$.
\end{itemize}
In this case, we say that $v$ is an \emph{unlucky} vertex. If there exists a height-$2$ non-unlucky vertex, then we process it using the operations described in Cases 2 and 3. We now consider the fourth and the last case, where all height-$2$ vertices in the current tree $\tau$ are unlucky.

\textbf{Case 4.} All height-$2$ vertices in $\tau$ are unlucky. 
If tree $\tau$ does not contain any height-$3$ vertices, then $\tau$ contains a unique height-$2$ vertex, but then this unique height-$2$ vertex is a degree-$2$ Steiner vertex in $\tau$, a contradiction. Therefore, tree $\tau$ contains height-$3$ vertices, and every height-$3$ vertex has at least two children.
Consider now any height-$3$ vertex $v^*$ and let $v,\hat v$ be two of its height-$2$ children. Since $v$ is unlucky, we let the vertices $v_1,v_2,u^1_1,u^1_2,u^2_1,u^2_2$ be defined as before, and we define the vertices $\hat v_1, \hat v_2, \hat u^1_1, \hat u^1_2, \hat u^2_1, \hat u^2_2$ similarly for $\hat v$. We also define set $S$ for $v$ as before.

Recall that $S\ne \set{u^1_1,u^1_2,u^2_1,u^2_2}$. From the construction of laminar family $\sset$, there exists a terminal $u$ in $S\setminus\set{u^1_1,u^1_2,u^2_1,u^2_2}$, such that $\min\set{w(u,u^1_1),w(u,u^1_2),w(u,u^2_1),w(u,u^2_2)}\le 2\ell\cdot(1+\eps)$. Assume without loss of generality that $w(u,u^1_1)\le 2\ell\cdot(1+\eps)$.

\textbf{Case 4.1.} If $w(v,v^*)\ge 10\eps\cdot \ell$ (the case where $w(\hat v,v^*)\ge 10\eps\cdot \ell$ is symmetric), then we delete from $\tau$ vertices $v_1,v_2,v$ and all its incident edges, and add new edges $(u^1_1,u^1_2),(u^2_1,u^2_2),(u^1_1,u^2_1)$ and $(u^1_1,u)$ (see \Cref{fig: case41}).

\textbf{Case 4.2.} If $w(v,v^*),w(\hat v,v^*)\le 10\eps\cdot \ell$. Denote $\ell= w(v^*,u^1_1)$ and $\hat \ell= w(v^*,\hat u^1_1)$, and assume without loss of generality that $\ell\le \hat\ell$. Then we delete all edges in the subtree rooted at $\tau^*$ expect for the $v^*$-$u^1_1$ path, and add edges $(u^1_1,u^1_2),(u^2_1,u^2_2),(u^1_1,u^2_1)$, edges  $(\hat u^1_1, \hat u^1_2),(\hat u^2_1, \hat u^2_2),(\hat u^1_1, \hat u^2_1)$ and edge $(u^1_1,\hat u^1_1)$ (see \Cref{fig: case42}). If $\hat\ell\le (1+10\eps)\cdot \ell$, then we further add set $X_{i^*}(v^*)=\set{u^1_1,u^2_1,\hat u^1_1,\hat u^2_1}$ into collection $\xset$, where $i^*=\floor{\log_{1+\eps}\ell}$.

\begin{figure}[h]
	\centering
	\subfigure[Case 4.1.]{\scalebox{0.11}{\includegraphics{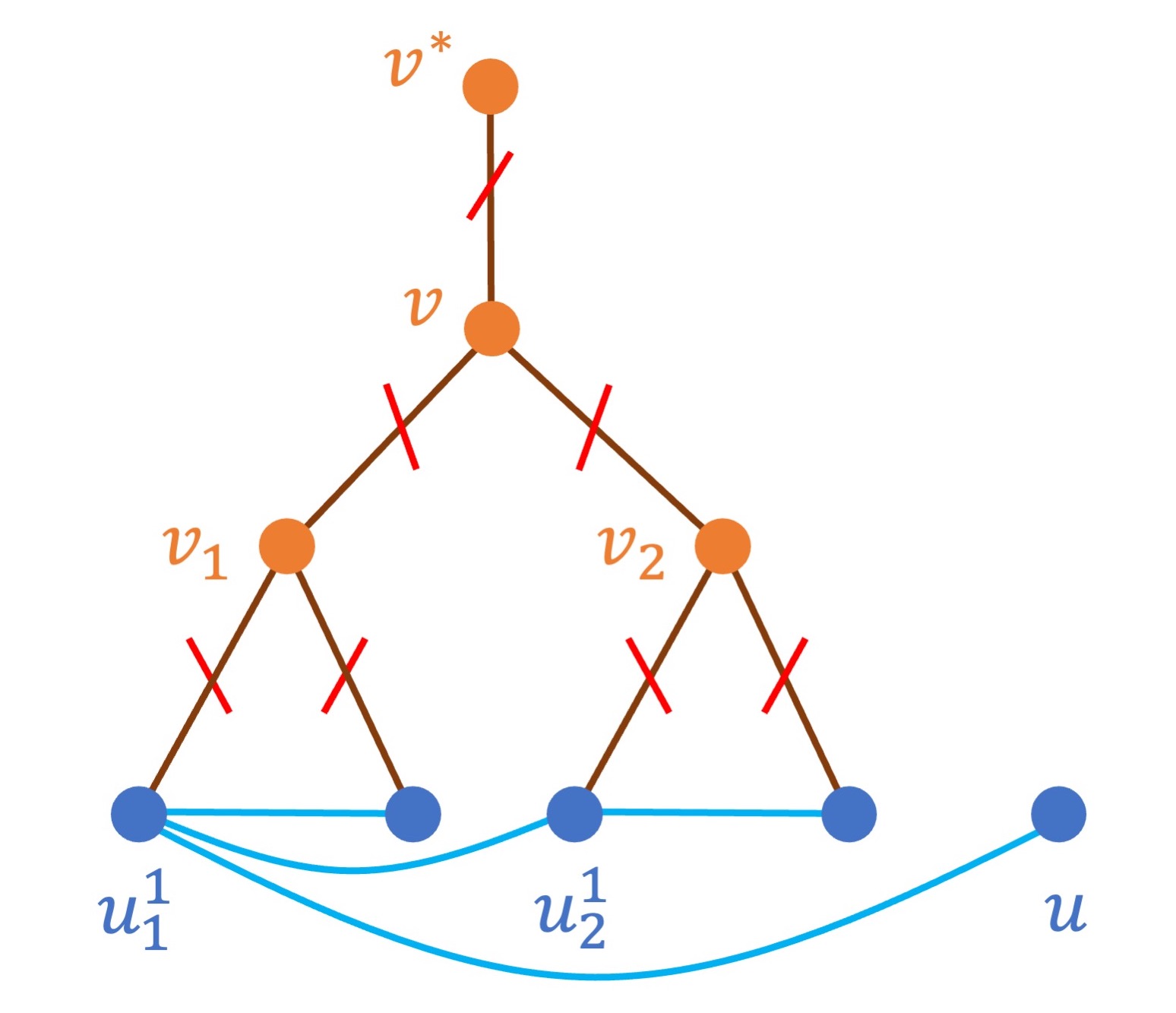}}\label{fig: case41}}
	\hspace{0.5cm}
	\subfigure[Case 4.2.]{
		\scalebox{0.1}{\includegraphics{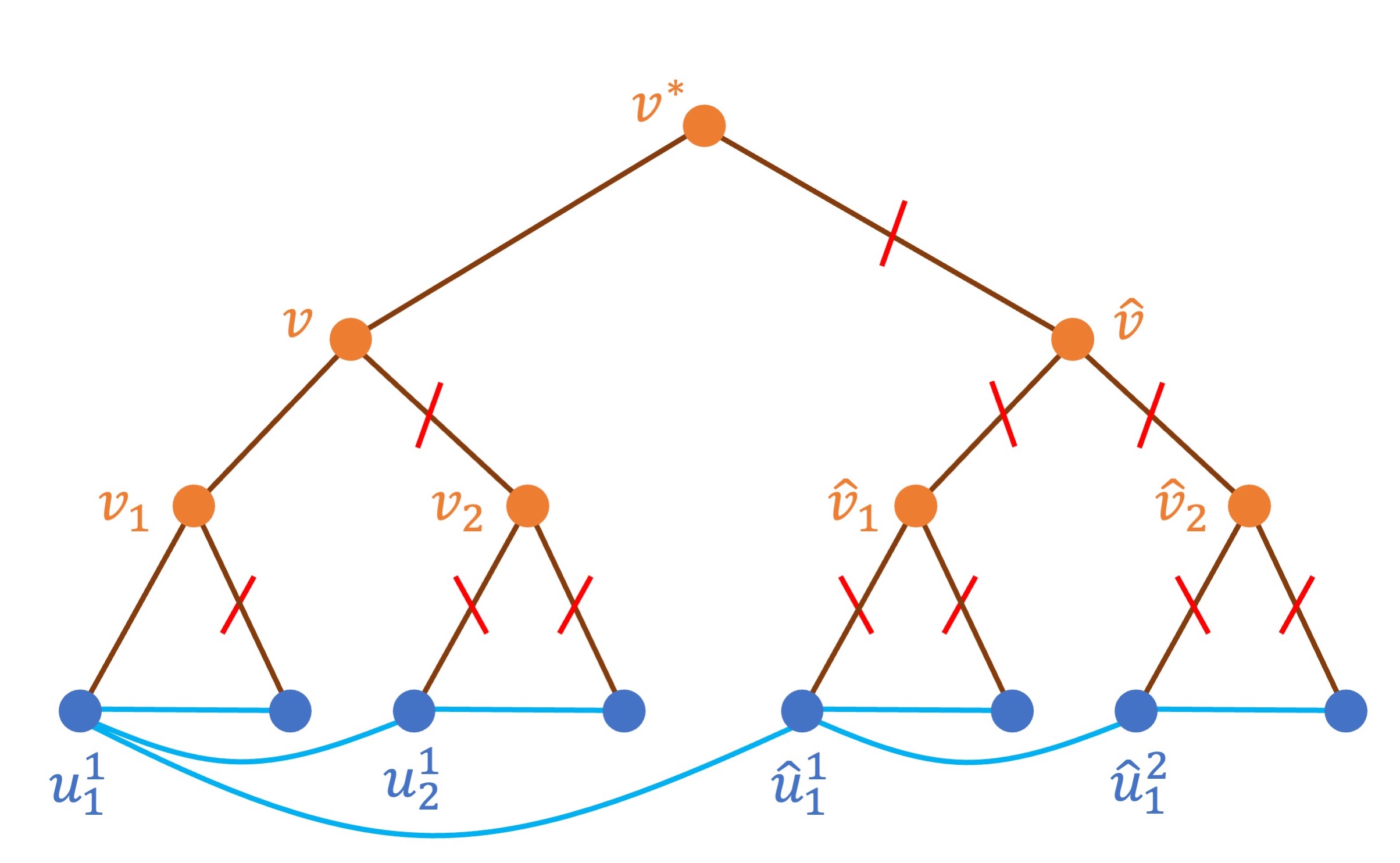}}\label{fig: case42}}
	\caption{An illustration of edge replacement in Case 4.\label{fig: case_4}}
\end{figure}

This finally completes the description of the tree-modification process.
Note that, in each of the cases described above, we replaced a set $E$ of edges that do not connect a pair of terminals in the current tree $\tau$ with a new set $E'$ of edges that connect a pair of terminals, such that $(\tau\setminus E)\cup E'$ is still a valid \St, such that either 
\begin{itemize}
\item $w(E')\le (2-4\eps_0)\cdot w(E)$; or
\item $w(E')\le 2\cdot w(E)$, and we have added a set $X_i(v)$ into $\xset$, such that $w(E)\le 10\cdot (1+\eps)^i$; or
\item $w(E')\le 2\cdot w(E)$, and we have added a set $Y$ into $\yset$, such that $\adv(Y)\ge \eps^{3/4}\cdot w(E)$; or
\item $w(E')\le 2\cdot w(E)$, and we have not added sets into $\xset$ or $\yset$, which may only happen in Case 3.3.4 where we performed a bad replacement.
\end{itemize}
 
First, it is easy to see that, the total cost of all edges where we perform bad replacements is at most $8\cdot\eps^{1/4}\cdot w(\tau_{\opt})$.
Second, from the construction of the Set Cover instances $\set{(\wset_i,U_i)}_{0\le i\le L-1}$ in Step 2, and using similar arguments in the proof of \Cref{obs: rep good for adv}, it is easy to show that when we add a set $X_i(v)$ into collection $\xset$, there is a set $W\in \wset_i$, such that, for each $u\in X_i(v)$, there is a terminal $u'\in W$, such that $w(u,u')\le \eps\cdot (1+\eps)^i$; and different sets $X_i(v)$ corresponds to different sets in $\wset_i$.
Therefore, according to the algorithm and the discussion above, the total weight of all edges in $\tau^*$ in which we either perform a bad replacement or add a set into $\xset$ or $\yset$ is at most
\[8\cdot\eps^{1/4}\cdot w(\tau_{\opt})+\frac{6\eps^{1/2}\cdot \mst}{\eps^{1/4}}+2^{20}\cdot\eps_0\cdot \mst\le \frac{\mst}{3}.\]
Therefore, if we denote by $\tau'$ the resulting tree we get from the above process, then 
\[\mst\le w(\tau')\le 2\cdot \frac{w(\tau_{\opt})}{3}+(2-4\eps_0)\cdot\frac{2\cdot w(\tau_{\opt})}{3}\le (2-2\eps_0)\cdot w(\tau_{\opt}).\]
It follows that $w(\tau_{\opt})\ge \mst/(2-2\eps_0)$. This completes the proof of the correctness of the algorithm.

\newcommand{\ASCOFF}{\textsf{AlgSetCovOff}}
\newcommand{\COVELE}{\textsf{CovEle}}
\newcommand{\COVSET}{\textsf{CovSet}}
\newcommand{\wlow}{\wset_{\textnormal{low}}}
\newcommand{\ulow}{U_{\textnormal{low}}}
\newcommand{\uhigh}{U_{\textnormal{high}}}
\newcommand{\tlow}{T_{\textnormal{low}}}
\newcommand{\thigh}{T_{\textnormal{high}}}

\subsection{Proof of \Cref{thm: sec cover}}
\label{sec: set cover}
In this section, we give a sublinear query algorithm for our Set Cover objective, and prove \Cref{thm: sec cover}.
Recall that we are given an instance $(U,\wset)$ of Set Cover, and the goal is to estimate the value of $\big( |U|-\setcover(U,\wset_{\ne 2})\big)$ where $\wset_{\ne 2}$ is the collection of sets in $\wset$ with size not equal to 2.
We note that the goal of estimating the number of sets needed to cover a universe has been considered from the perspective of sublinear query algorithms~\cite{har2016towards,indyk2018set,grunau2020improved}. However, these results do not apply to our setting as we need to estimate the difference between the universe size and the set cover size. 

We first give an algorithm that outputs an estimate of $\big( |U|-\setcover(U,\wset)\big)$, and then show how to modify the algorithm to prove \Cref{thm: sec cover}.

\subsubsection*{An Algorithm for Estimating $\big( |U|-\setcover(U,\wset)\big)$}
\label{sec: set_cover_general}

The main result of this subsection is the following theorem.

\begin{theorem} \label{thm: sec cover 1}
There is a polynomial-time randomized algorithm, that, given an instance $(U,\wset)$ of Set Cover and any constant $0<\eps<1$, with probability $1-O(n^{-2})$, returns a $(4,\eps|U|)$-estimation of $\big(|U|-\setcover(U,\wset)\big)$, by performing $O((|\wset|^{3/2}+|\wset|^{3/4}|U|^{}) \cdot(\log n)^2/\eps^3 )$ membership queries to the instance $(U,\wset)$.
\end{theorem}

Before we describe the algorithm in detail, we give a brief intuition. First we observe that, with high probability, all elements that appear in many (at least $\Omega(|\wset|\log |\wset|/\eps|U|)$) sets in $\wset$ can be covered by a random subset of $O(\eps |U|)$ sets in $\wset$, and so they can be ignored as we allow $\eps|U|$ additive error in the estimation. 
Assume for simplicity that all elements that appear in fewer than $o(|\wset|\log |\wset|/\eps|U|)$ sets in $\wset$.
Consider the optimal set cover $\wset^*$ and assume the sets in $\wset^*$ are arranged into a sequence. Now, if a set in $\wset^*$ covers $t$ elements that are not covered by the previous sets in the sequence, then it can be viewed as ``contributing'' $(t-1)$ to $\big(|U|-\setcover(U,\wset)\big)$.
We will show that, at a high-level, this ``contribution'' can be characterized as follows: if we consider the graph on $U$ where there is an edge $(u,u')$ iff $u,u'$ appear in the same set in $\wset$, then it can be shown that the size of any maximal matching in the graph is within a constant factor of $\big(|U|-\setcover(U,\wset)\big)$. We then focus on this graph and utilize the algorithm from \cite{Behnezhad21} to estimate the size of any maximal matching in it.

We now describe the algorithm in detail. 
For convenience, we denote $n=|\wset|$ and $k=|U|$.
Throughout this subsection, we use a parameter $\beta=\max\set{k/n^{3/4},1}$.
In other words, when $k\ge n^{3/4}$, $\beta=k/n^{3/4}$, otherwise $\beta=1$.

We define the \emph{frequency} of an element $u\in U$ in the collection $\wset$ to be the number of sets in $\wset$ that contain $u$.
We first partition the vertices into two subsets according to their frequency in $\wset$ as follows. Let $\tilde\wset$ be a random sub-collection of $\wset$ that contains $k/\beta$ sets. For every element $u\in U$, we compute the frequency $u$ in $\tilde\wset$ by performing membership queries on all pairs $(u,W)$ with $W\in \tilde\wset$. We then let $U_{\text{low}}$ be the set of all elements in $U$ with frequency at most $75 \beta n\log n / (\eps k)$ in $\tilde\wset$, and let $U_{\text{high}}=U\setminus U_{\text{low}}$. 
The total number of queries that are needed to compute this partition is at most 
$$k\cdot (k/\beta))=O(k^2/\beta)=O\bigg(\frac{k^2}{\max\set{k/n^{3/4},1}}\bigg)=O\bigg(k^2\cdot \min\set{n^{3/4}/k,1}\bigg)=\underline{O(n^{3/2}+kn^{3/4})}.$$

Let $N=50 n\beta\log n /(\eps k)$.
We use the following observations that follow from the Chernoff Bound in \Cref{lem: Chernoff}.

\begin{observation}
\label{obs: high low}
With probability $1-n^{-2}$, all elements in $U_{\textnormal{low}}$ have frequency at most $2N$ in $\wset$, and all elements in $U_{\textnormal{high}}$ have frequency at least $N$ in $\wset$.
\end{observation}

\begin{observation}
\label{obs: high free}
Let $\hat\wset$ be a random sub-collection of $\wset$ that contains $\eps k/(10\beta)$ sets. Then with probability $1-n^{-2}$, every element in $U_{\textnormal{high}}$ is contained in some set in $\hat \wset$.
\end{observation}

We now focus on the elements in $U_{\text{low}}$. We define $\wlow=\set{W\cap U_{\text{low}}\mid W\in \wset}$.
From \Cref{obs: high low} with probability $1-n^{-2}$, all elements in $\ulow$ have frequency at most $2N$ in $\wlow$. 
We define a graph $H$ as follows. Its vertex set is $\ulow$, and its edge set contains, for every pair $u,u'\in \ulow$, an edge $(u,u')$ iff there exists a set $W\in \wlow$ that contains both $u$ and $u'$. We prove the following lemma.

\begin{lemma} \label{lem:sc-mt}
	Let $M$ be any maximal matching in $H$. Then   \[\frac{\card{\ulow}-\setcover(\ulow,\wlow)}{2} \le \card{M} \le \card{\ulow}-\setcover(\ulow,\wlow).\]
\end{lemma}

\begin{proof}
	On the one hand, we can construct a set cover of size at most $\card{\ulow}-|M|$ as follows. For every pair $u,u'\in \ulow$ that is matched in $M$, we take any set that contains both $u,u'$; for every $u''\in \ulow$ that is not matched in $M$, we take the set $\set{u''}$. Clearly, we have taken at most $|\ulow|-|M|$ sets and they form a set cover, so $|M|\le |\ulow|-\setcover(\ulow,\wlow)$.
	On the other hand, consider an optimal set cover $\fset$ with $|\fset|=\setcover(\ulow,\wlow)$. Note that each set $W\in \fset$ contains at most one element that is not matched in $M$; otherwise there are $u,u'\in \ulow$ that are both not matched in $M$ while the edge $(u,u')$ belongs to $H$ by the definition of $H$, a contradiction to the maximality of $M$. Therefore, the number of vertices in $H$ that are unmatched in $M$ is at most $\setcover(\ulow,\wlow)$, which implies that $\frac{\card{\ulow}-\setcover(\ulow,\wlow)}{2} \le \card{M}$.
\end{proof}

We use the following result, which is implicit in Section $3$ of~\cite{Behnezhad21}. 

\begin{theorem}[\!\!\!\cite{Behnezhad21}\!] \label{thm:mat}
	Let $G$ be a graph on $Z$ vertices with average degree $\bar d$.
	If we are given an oracle that takes a vertex $v$ of $G$ as input and outputs all neighbors of $v$ in $G$, then there is an algorithm, that, given any $0<\eps <1$, with probability at least $1-n^{-2}$ estimates the size of some maximal matching in $G$ to within an additive factor of $\eps Z$, by performing  $O(\bar{d}\log Z /\eps^2)$ queries to the oracle. 
\end{theorem}

In order to use \Cref{thm:mat} to estimate the size of a maximal matching in $H$, we need to efficiently implement an oracle that, given any $u\in U$, finds all neighbors of $u$ in $H$, which we do next. 

\paragraph{A subroutine for finding all neighbors of a given vertex in $H$.}
From the definition of $H$, the neighbors of $u$ are the elements $u'$ in $\ulow$ that are contained in the same set as $u$. 
We find all neighbors of $u$ in $H$ as follows. We first perform membership queries on all pairs $(u,W)$ with $W\in \wlow$ (a total of $O(n)$ queries), and find all sets in $\wlow$ that contains $u$. Then for each set $W\in \wlow$ that contains $u$, we perform membership queries on all pairs $(u',W)$ with $u'\in \ulow$ (a total of at most $2N\cdot k$ queries as $u$ is contained in at most $2N$ sets), and figure out which elements are contained in $W$. The set of all neighbors of $u$ in $H$ is then obtained by taking the union of all sets $W$ that contain $u$. In the whole subroutine, we have performed $2N\cdot k= O(\beta n \log n/\eps)$ queries in total. 

From \Cref{thm:mat} with the algorithm described above serving as the oracle, we can with high probability estimate the size of a maximal matching to within an additive factor of $\eps k$ with $O(\bar{d} \cdot \beta n \log^2 n/\eps^3)$ membership queries, where $\bar d$ is the average degree in $H$. However, $\bar d$ can be as large as $k$, in which case we can only get an $\tilde{O}(k\cdot n/\eps^3)$ upper bound. To improve upon this, we will pre-process the instance $(\ulow,\wlow)$ before using the algorithm in \Cref{thm:mat}.

We partition the collection $\wlow$ into $\wset_1$ and $\wset_2$ as follows.

We set parameter $Q=\max\set{k\log^{-2} n /\eps n^{1/4}, \sqrt{n}\log^{-2} n /\eps}$.
That is, when $k\le n^{3/4}$, $Q=\sqrt{n}\log^{-2} n /\eps$, otherwise $Q=k\log^{-2} n /\eps n^{1/4}$.
Let $\tilde U$ be a size-$Q$ random subset of $\ulow$.
We let $\wset_1$ contain all sets $W\in \wlow$ with $|W\cap \tilde U|\le \log  n$, and let $\wset_2=\wlow\setminus \wset_1$.
Using Chernoff bound and similar arguments in the proof of \Cref{obs: high low}, we can show that with probability $1-n^{-2}$, every set in $\wset_1$ contains at most $(100k\log n/Q)$ elements, and every set in $\wset_2$ contains at least $(k\log n/100Q)$ elements. The number of queries that are needed for computing this partition is at most \underline{$O(nQ)=\tilde O_{\eps}(n^{3/2}+kn^{3/4})$}.

We now consider the instances $(\ulow,\wset_1)$ and  $(\ulow,\wset_2)$ separately. We define graphs $H_1,H_2$ for $(\ulow,\wset_1)$ and $(\ulow,\wset_2)$ respectively, in a similar way that $H$ is defined for instance $(\ulow,\wlow)$. We use the following observation.

\begin{observation}
\label{obs: matching size}
Let $M_1$ be any maximal matching in $H_1$, and let $M_2$ be any maximal matching in $H_2$. Then there exists a maximal matching in $H$ whose size is between $(\card{M_1}+\card{M_2})/2$ and $\card{M_1}+\card{M_2}$.
\end{observation}
\begin{proof}
Note that $H=H_1\cup H_2$. 
Assume without loss of generality that $|M_1|\ge |M_2|$.
We construct a matching $M$ in $H$ as follows. We start with $M=M_1$. We add all edges of $M_2$ that do not share endpoint with any edges in $M_1$. Then we greedily add edges in $H_2$ that do not share endpoint with any of the current edges in $M$, until no edge can be added. Since $M_1$ is a maximal matching in $H_1$, from our algorithm, it is easy to verify that the resulting matching $M$ is a maximal matching in $H$. Note that $|M|\ge |M_1|\ge (\card{M_1}+\card{M_2})/2$. 

It remains to show that $|M|\le \card{M_1}+\card{M_2}$. We denote by $M'_2$ the subset of edges in $M_2$ that do not share endpoints with edges in $M_1$. It suffices to show that in the last step we have greedily added at most $|M_2\setminus M'_2|$ edges into $M$. Consider an edge $e$ added to $M$ in this step. Since $e\in E(H_2)$ and $M_2$ is a maximal matching in $H_2$, $e$ shares an endpoint with some edge $e'\in M_2$, and $e'$ has not been added to $M$ in the previous step, so $e'\in M_2\setminus M'_2$. We say that $e'$ is the \emph{blocker} of $e$. In fact, every edge $e'\in M_2\setminus M'_2$ can be the blocker of at most one edge added in the last step, since if $e'=(u,u')$ and $u$ is an endpoint of some edge in $M_1$, then every edge $e$ blocked by $e'$ has to contain $u'$ as an endpoint, and thus there can be at most one such edge. 
\end{proof}

We first consider the instance $(\ulow,\wset_1)$. Recall that each set in $\wset_1$ has size at most $(100k\log n/Q)$, and every element in $\ulow$ is contained in at most $2N$ sets in $\wset_1$. Therefore, every element in $\ulow$, as a vertex in $H_1$, has at most $(100k\log n/Q) \cdot 2N = \tilde O(\beta n/Q)$ neighbors. We now apply the algorithm in \Cref{thm:mat} with parameter $\eps/10$ together with the subroutine for finding all neighbors of a given vertex described above to obtain an estimate $X_1$ of the size of a maximal matching in $H_1$. Since the average degree in $H_1$ is $\tilde O(\beta n/Q)$, the number of queries performed by the algorithm is 
$$O(\beta n \log n/\eps)\cdot \tilde O(\beta n/Q)\cdot (\log  n)/\eps^2=\tilde O(\beta^{2} n^{2}/Q)=
\tilde O\bigg(n^{2}\cdot\frac{\beta}{Q}\cdot  \beta\bigg)=\tilde O\bigg(n^{3/2}\cdot  \beta\bigg)=\underline{\tilde O(n^{3/2}+kn^{3/4})}.$$

We next consider the instance $(\ulow,\wset_2)$.
Recall that every set in $\wset_2$ has size at least $(k\log n/100Q)$ and every element is contained in at most $2N$ sets. Therefore,
$$\card{\wset_2} 
\le \frac{ k \cdot 2N }{ (k\log n/100Q)}
=\tilde O\bigg(\max\set{k,\frac{n^{3/2}}{k}}\bigg).$$ 
Therefore, if $k>n^{3/4}$, we can simply use another $\tilde O(n)$ queries to estimate $|\bigcup_{W\in \wset_2}W|$ to obtain an estimate of $\setcover(U,\wset_2)$ to within an additive $\eps k$ factor. If $k\le n^{3/4}$, we can then perform $k\cdot |\wset_2|=\tilde O(k^2+ n^{3/2})=\underline{\tilde O(n^{3/2})}$ queries to obtain the entire instance $(\ulow, \wset_2)$, compute the graph $H_2$, and then compute the size $X_2$ of a maximal matching in $H_2$.

Lastly, we return $X=|U_{\text{high}}|-\eps k/10+(X_1+X_2)/2$ as our final output. From the above discussion, the number of queries we have performed is $\tilde O(n^{3/2}+kn^{3/4})$.

We next analyze the probability that $X$ is a $(4,\eps  |U|)$-approximation of $|U|-\setcover(U,\wset)$.

\paragraph{Bad Event $\xi$.} We define $\xi$ to be the bad event that either of the following happens: (i) there exists some element in $U_{\text{low}}$ with frequency at more than $2N$ in $\wset$ or there exists some element in $U_{\text{high}}$ with frequency at most $N$ in $\wset$; (ii) there does not exist a subcollection of $\eps k/(10\beta)$ sets in $\wset$ that cover all elements in $U_{\text{high}}$; 
(iii) there exists some set in $\wset_1$ that contains more than $(100k\log n/Q)$ elements, or there exists some set in $\wset_2$ that contains fewer than $(k\log n/100Q)$ elements;
and (iv) in the application of the algorithm from \Cref{thm:mat} to $H_1$, the output $X_1$ is not an estimate of the size of any maximal matching in $H_1$ to within an additive error of $\eps k/10$. From \Cref{obs: high low}, \Cref{obs: high free}, \Cref{thm:mat} and the above discussion, $\Pr[\xi]=O(n^{-2})$.

The proof of \Cref{thm: sec cover 1} is concluded by the following claim.
\begin{claim}
\label{clm: main}
If event $\xi$ does not happen, then $X$ is a $(4,\eps |U|)$-estimation of $|U|-\setcover(U,\wset)$.
\end{claim}
\begin{proof}
Assume now that event $\xi$ does not happen. 
Let $M_1$ be a maximal matching in $H_1$, such that $X_1\le |M_1|\le X_1+\eps |U|/10$. Let $M_2$ be a maximal matching in $H_2$, such that $X_2=|M_2|$. Let $M$ be any maximal matching in $H$. From \Cref{obs: matching size},
\[
\frac{X_1+X_2}{2}\le \frac{|M_1|+|M_2|}{2}\le |M|\le |M_1|+|M_2|\le X_1+X_2+\frac{\eps |U|}{10}.
\]
Then from \Cref{lem:sc-mt},
\[
\frac{X_1+X_2}{2}\le |M|\le |\ulow|-\setcover(\ulow,\wlow)\le 2|M|\le 2(X_1+X_2)+\frac{\eps |U|}{5}.
\]

Let $\tilde \wset$ be a sub-collection of $\eps |U|/10$ sets in $\wset$ that cover all elements in $U_{\text{high}}$ (such a subcollection exists since event $\xi$ does not happen). Let $\tilde \wlow$ be an optimal set cover of instance $(\ulow,\wlow)$. Clearly, $\tilde \wset\cup \tilde \wlow$ is a feasible set cover of instance $(U,\wset)$. Therefore, $\setcover(U,\wset)\le \setcover(\ulow,\wlow)+\eps |U|/10$, and
\[
\begin{split}
X & =|U_{\text{high}}|-\frac{\eps |U|}{10}+\frac{X_1+X_2}{2}
\le |U_{\text{high}}|-\frac{\eps |U|}{10}+|M|\le |U_{\text{high}}|-\frac{\eps |U|}{10}+|\ulow|-\setcover(\ulow,\wlow)
\\
& 
= |U|-\setcover(\ulow,\wlow)-\frac{\eps |U|}{10}
\le |U|-\setcover(U,\wset).
\end{split}
\]
On the other hand, since $\setcover(\ulow,\wlow)\le \setcover(U,\wset)$,
\[
\begin{split}
 |U|-\setcover(U,\wset) & \le |U|-\setcover(\ulow,\wlow)=|U_{\text{high}}|+|\ulow|-\setcover(\ulow,\wlow)\le |U_{\text{high}}|+2|M| 
\\
& 
\le |U_{\text{high}}|+2(X_1+X_2)+\frac{\eps |U|}{5}
\le 4\cdot \bigg(|U_{\text{high}}|-\frac{\eps |U|}{10}+\frac{X_1+X_2}{2}\bigg)+\eps |U| =4X+\eps |U|.
\end{split}
\]
\end{proof}

\subsubsection*{Handling the cardinality-$2$ sets}

We follow the same algorithm as decribed in \Cref{sec: set_cover_general}. We first partition $U$ into $U_{\text{low}}$ and $U_{\text{high}}$, and then focus on instance $(\ulow,\wlow)$. We partition $\wlow$ into $\wset_1$ and $\wset_2$ and construct graphs $H_1$ and $H_2$. Note that when constructing these partitions, we do not (and are not able to) ignore the sets of size 2.
Since every set in $\wset_2$ contains more than $2$ elements, computing an estimate of the size of a maximal matching in $H_2$ can be done in the same way. 
When processing the instance $(\ulow,\wset_1)$, in the subroutine of finding all neighbors of an element $u\in \ulow$, we just need to ignore all sets of size $2$ in $\wset_1$ that contains $u$. 

Regarding the proof that the output is a $(4,\eps  |U|)$-approximation of $|U|-\setcover(U,\wset)$, \Cref{obs: high free} shows that with probability $1-n^{-2}$ there exists a subcollection $\tilde \wset$ of at most $\eps|U|/10$ sets in $\wset$ that cover all elements in $U_{\text{high}}$. We need to modify it to show that with probability $1-n^{-2}$ there exists a subcollection of $\eps|U|/5$ sets in $\wset_{\ne 2}$ that cover all elements in $U_{\text{high}}$, and this can be simply achieved by replacing all cardinality-$2$ sets in $\tilde \wset$ with singleton sets that contain all elements that are covered by the cardinality-$2$ sets in $\tilde \wset$. And the proof of \Cref{clm: main} now still goes through with this modification.

\newcommand{\BFS}{\mathsf{BFS}}
\newcommand{\FIND}{\mathsf{FIND}}

\subsection{Implementation without Knowing the Terminal-Induced Metric Upfront}
\label{subsec: without MST}

In this subsection, we complete the proof of \Cref{thm: beat-2-main} by showing that the algorithm described in \Cref{subsec: with MST} and \Cref{sec: set cover} can in fact be implemented without knowing the terminal-induced metric upfront, at the cost of a slightly worse query complexity $\tilde O (n^{12/7}+n^{6/7}\cdot k)$. 

Intuitively, since we do not have the terminal-induced metric upfront, we can no longer assume that we can start the process with a MST on terminals along with its hierarchical structure in our hand. However, as Step 2 and Step 3 in our algorithm only utilize local MST structure to find evidence, we do not really need the whole MST in order to implement them, but we can instead locally explore the MST structure (and its hierarchical structure) whenever needed.
Therefore, we start by introducing two BFS-type subroutines that are crucial for implementing Steps 2 and 3. 
These subroutines are similar to the ones used in \cite{czumaj2009estimating}.

\subsubsection*{Auxiliary BFS-type Subroutines}

Since we are not able to query the distances between each pair of terminals, we are not able to precisely recover the graph $H_i$ for each layer $i$. Thus, we will define a graph $\hat{H}_i$ which lies somewhere  ``in-between'' $H_i$ and $H_{i+1}$ but makes it easier to explore locally.

Let $\eps_1 = \eps/10$. We first construct a graph $\bar{H}_i$ as follows: the vertices are terminals, and there is an edge between each pair of terminals if their distance is less than $\eps_1(1+\eps)^i$. We assign a random rank on each terminal, and let $R_i$ be the lexicographically first maximal independent set based on this rank. We say the terminals in $R$ are the {\em representative terminals}, and for each terminal $u$, there is a representative terminal $u'$ at distance at most $\eps_1(1+\eps)^i$ away from $u$. We say {\em $u$ is represented by $u'$}. We define $\hat{H}$ as follows: any terminal is connected to its representative terminal, and for any two representative terminals in $R$, they are connected in $\hat{H}$ if and only if their distance is at most $(1+3\eps_1)(1+\eps)^i$. 
The following observation directly follows from these definitions:

\begin{observation} \label{obs:obs1}
    Any component in $H_i$ is also connected in $\hat{H}_i$, and any component in $\hat{H}_i$ is also connected in $H_{i+1}$.
\end{observation}

We define the {\em size of a component} in $\hat{H}$ as the number of representative terminals in $\hat{H}$. Next, we define another directed graph $H_i^C$ as follows. The vertices of $H_i^C$ are the components in $\hat{H}$, and for any two components $S_1$ and $S_2$, there is a directed edge from $S_1$ to $S_2$ if there is a representative terminal $u_1 \in S_1$ and a terminal $u_2 \in S_2$ such that their distance is at most $(1+3\eps_1)(1+\eps)^i$. We say that a {\em component $S_1$ can reach a component $S_2$} if there is a directed path from $S_1$ to $S_2$ in $H_i^C$. We have the following observation:

\begin{observation} \label{obs:obs2}
    If $S_1$ can reach $S_2$ in $H_i^C$, then $S_1$ and $S_2$ are inside the same component in $H_{i+1}$.
\end{observation}

We define $U_i$ to be the components $S$ in $\hat{H}_{i-1}$ such that the total size of the components that can be reached by $S$ is at most $100 L \log^2 n/\eps_1$. 

Throughout the process, we will always maintain the knowledge of whether a terminal is representative or not. At first, it is {\em unknown} for each terminal whether or not it is a representative. The subroutine $\FIND(u,i)$ takes as input a terminal $u$ and a parameter $i$ and outputs a representative $u'$ in the following manner. We query the distance between $u$ and all other terminals, if all terminals that are at most $\eps_1(1+\eps)^i$ away from $u$ have lower rank than $u$ or are known not to be a representative, then we output $u$. Otherwise let $u_1$ be the highest rank terminal among them, and we repeat this process on $u_1$, and continue in this manner, until we find a representative $u'$. We then mark $u'$ as a representative, and all terminals that are at most $\eps_1(1+\eps)^i$ away from $u'$ as non-representative. Note that $u'$ might not represent $u$, but $u$ and $u'$ must be in in the same component in $\hat{H}_i$. The following observation gives an upper bound on the running time of $\FIND$.

\begin{observation} \label{obs:find}
    With high probability, for any $u$ and $i$, $\FIND(u,i)$ uses $\tilde{O}(k)$ queries.
\end{observation}

\begin{proof}
    The exploration sequence is in fact a path in the DFS tree of the running of greedy parallel maximal independent set problem, and the depth of the DFS tree is $O(\log n)$ with high probability~\cite{FischerN20}. So with high probability, the length of the exploration sequence is $\tilde{O}(1)$ and the total number of queries we use is $\tilde{O}(k)$.
\end{proof}

Next, we give a BFS subroutine that takes as input a terminal $u$ and an integer $i$. Intuitively, the subroutine explores the neighborhood of $u$ in its level-$i$ connected component up to a poly-logarithmic depth. Throughout, every terminal is either marked \textsf{out}, or \textsf{in}, or \textsf{representative}, or \textsf{active}. Initially, terminal $u$ is marked \textsf{active} and all other terminals are marked \textsf{out}. 
The procedure $\BFS(u,i)$ proceeds in rounds. 
In each round, we first pick an \textsf{active} terminal $u'$ with the maximum rank (if there is no such terminal then the procedure is terminated). We run $\FIND(u',i)$ and let $\hat u$ be the output. We mark $\hat u$ as \textsf{representative}.
We then query the distance between $\hat u$ and all other terminals marked \textsf{out}, for each such terminal $u''$, if $w(\hat u, u'') < \eps_1(1+\eps)^i$, then we mark $u''$ \textsf{in}. 
If $\eps_1(1+\eps)^i \le w(\hat u,u'') < (1+3\eps_1)(1+\eps)^i$, we mark $u''$ \textsf{active}. 
This completes the description of a round.
If the procedure did not terminate after $100 L \log^2 n/\eps_1$ rounds, then we artificially terminate the procedure and mark all \textsf{active} terminals \textsf{in}.

It is easy to observe that, after the procedure $\BFS(u,i)$ terminates, all terminals marked \textsf{in} or \textsf{representative} are certified to lie in the a component that is reachable from the component that contains $u$ in $\hat{H}_i$
Clearly, if the procedure is not artificially terminated, then we have found all components in $\hat{H}_i$ that is reachable from the component containing $u$ with all representative terminals inside it. So we can check if the component that contains $u$ is inside $U_{i+1}$ or not. If the procedure is artificially terminated, then $u$ is contained in a component that is not in $U_{i+1}$.
As the procedure $\BFS(u,i)$ runs for at most $100 L \log^2 n/\eps$ rounds and each round takes at most $\tilde{O}(k)$ queries, its query complexity is $\tilde O(k/\eps_1)$. 

We prove the following lemma.

\begin{lemma} \label{lem:small}
    For any integers $i$ and $M$, either (i) level $i$ is light; or (ii) $\card{R_{i-1}} \le 2ML \log n/\eps_1$; or (iii) $\card{U_i} \ge M$ holds.
\end{lemma}

\begin{proof}
    Assume that (i) and (ii) do not hold. We will show that $\card{U_i} \ge M$ must hold. Since $\card{R_{i-1}} \le 2ML \log n/\eps_1$, the $2M L \log n/\eps$ representative terminals are at distance at least $\eps_1(1+\eps)^{i-1}$ from each other, so $\mst \ge (2M L \log n/\eps) \cdot \eps(1+\eps)^{i-1} = 2M L \log n (1+\eps)^{i-1}$. On the other hand, since level $i$ is not light, the total weight of all level-$i$ edges in $\mst$ is at least $\mst/(L \log n) \ge 2M (1+\eps)^{i-1}$. Note that $w_i(\mathcal{T^*})\le \card{\mathcal{S}_i} (1+\eps)^i$, so $\card{\mathcal{S}_i} \ge 2M /(1+\eps)$. Lastly, from \Cref{obs: small sets enough}, and the fact that any component of $\hat{H}_{i-1}$ that is not in $\card{U_i}$ is inside a large component in $H_i$ by \Cref{obs:obs2}. $\card{U_i} \ge \frac{2 (1-O(1/\log n)) M}{1+\eps} > M$.
\end{proof}

We now proceed to describe the simulation of Steps 2 and 3 of our algorithm in the previous subsections. Note that if $k=O(n^{6/7})$, then we can simply perform $k^2=O(n^{12/7})$ queries to obtain the terminal-induced metric and then perform Steps 2 and 3. The query complexity is $O(n^{12/7})+\tilde O(n^{3/2}+n^{3/4}\cdot k)=\tilde O(n^{12/7}+n^{6/7}\cdot k)$. Therefore, we assume from now on that $k=\Omega(n^{6/7})$. \footnote{The main reason that we can run on graph $\hat{H_i}$ instead of $H_i$ is the following: the algorithm we give in the previous section is in fact also true if we run layer $i$ algorithms on graph $H_{i+1}$ since the threshold we use for setting up the set cover instance, $3/5(1+\eps)^i$ can be an arbitrary number between $1/2(1+\eps)^i$ and $(1+\eps)^i$, which means it also works for $3/5(1+\eps)^{i+1}$. And since $\hat{H_i}$ is between $H_i$ and $H_{i+1}$ by \Cref{obs:obs1} and \Cref{obs:obs2}, the analysis still works if we run on $\hat{H_i}$.}

\subsubsection*{Simulation of Step 2}
We now describe how to simulate Step 2 of the algorithm described in \Cref{subsec: with MST}. 
Recall that, in Step 2, we have constructed, for each index $i$, an instance $(U_i,\wset_i)$ of Set Cover for finding local evidence at level $i$, and designed an algorithm called $\algsetcover$ for estimating the value of $|U_i|-\setcover(U_i, (\wset_i)_{\ne 2})$. In particular, the sets in $\wset_i$ correspond to Steiner nodes and the elements in $U_i$ correspond to level-$i$ connected components whose representative size is small (at most $L\log^2 n/\eps$), and a set $W\in \wset_i$ contains an element in $U$ iff the Steiner node that $W$ corresponds to is at distance at most $(3/5)\cdot(1+\eps)^i$ from some representative of the component that the element in $U$ corresponds to.

Fix an index $i$, and for convenience we denote $U=U_i$ and $\wset=\wset_i$.
As we do not have the MST on terminals, we do not know the level-$i$ connected components, which means we do not know the element in $U$ but can only make queries to locally explore them. The simulation of Step 2 (and Step 3) finds the best tradeoff between the query complexity of this additional local exploration task and the previous algorithmic steps.
In the remainder of this section, we use the parameter $M=n^{6/7}\log^2 n$. 

We first find the first $(2M L \log n/\eps_1)$ terminal in $R_{i-1}$ by greedy MIS algorithm. 
The query complexity of this step is $\tilde O(Mk)$.
If $\card{R_{i-1}} \le (2M L \log n/\eps)$, then we have already figured out all level-$i$ connected components together with the representatives in each component, and therefore we can now run the algorithm $\algsetcover$ described before to obtain an estimate of the value of $|U|-\setcover(U, \wset_{\ne 2})$, whose query complexity is $\tilde O(n^{3/2}+n^{3/2}\cdot k)$.
Assume from now on that $\card{R_{i-1}} > (2M L \log n/\eps)$

From \Cref{lem:small}, either level $i$ is light, or $|U|\ge M$. In order to determine which case happens, we will estimate the size of $U$. Specifically, we pick a random terminal $u$ and run the procedure $\BFS(i-1,u)$. If the level-$i$ connected component containing $u$ is small, then we set $X(u)$ to be the inverse of the size of this component, otherwise we set $X(u)=0$. 
It is easy to observe that $X$ is a random variable supported on $[0,1]$, and $\mathbb{E}[X]=|U|/k$.
Therefore, from Chernoff Bound, if we repeat the process for $100k \log n/M$ random sampled terminals, then with probability $1-n^{-10}$, we can estimate the value of $|U|$ to within an additive factor of $M/5$. Therefore, we can either correctly claim that $|U|<M$ or correctly claim that $|U|\ge M/2$. 
The query complexity is 
$\tilde O(k)\cdot (100k \log n/M) = \tilde{O}(k^2/M)= \underline{\tilde{O}(k\cdot n^{1/7})}$.
If the claim is $|U|<M$, then from \Cref{lem:small}, level-$i$ is light, and we will just ignore this layer by giving up local evidence on it. From now on we assume that $|U|>M/2$. 

We now simulate the algorithm $\algsetcover$ in \Cref{sec: set cover} with some modifications. 
We use another parameter $R=50 n^{1/7}\log n/\eps$. First we partition the terminals into subsets $\thigh$ and $\tlow$ such that (i) for every $u\in \thigh$, the number of Steiner nodes $v$ at distance at most $(3/5)\cdot (1+\eps)^i$ from $u$ is at least $R$; and
(ii) for every $u\in \tlow$, the number of Steiner nodes $v$ at distance at most $(3/5)\cdot (1+\eps)^i$ from $v$ is at most $2R$.
Such a partition can be computed with $\underline{\tilde O(k n/R)=\tilde O(k\cdot n^{6/7})}$ queries. 
Note that, for each terminal in $\thigh$, the element in $U$ that corresponds to the level-$i$ connected component that the terminal belongs to is contained in at least $R$ sets. Since $|U|>M/2$, we can show via similar arguments that $\eps |U|$ random sets will cover all these elements (as $M\cdot R>n/\eps$). Therefore, we can ignore all level-$i$ connected components that contain a terminal in $\thigh$.

Next, we partition the Steiner vertices into subsets $V_1$ and $V_2$, using another parameter $P=n^{2/7}$ such that (i) for every vertex $v\in V_1$, the number of terminals in $\tlow$ at distance at most $(3/5)\cdot (1+\eps)^i$ from $v$ is at most $100 P$; and (ii) for every vertex $v\in V_1$, the number of terminals in $\tlow$ at distance at most $(3/5)\cdot (1+\eps)^i$ from $v$ is at least $P$. Such a partition can be computed with $\tilde{O}(nk/P)=\underline{\tilde{O}(k\cdot n^{5/7})}$ queries. A similar argument shows 
\[\card{V_2} \le\frac{kR}{P}=\frac{k \cdot 50 n^{1/7} \log n/\eps}{n^{2/7}} = O(n^{6/7}\log n) = o(M)=o(|U|).\]

We now define $\ulow$ as the set of small level-$i$ connected components that consist of only terminals in $\tlow$. Let $\wset_1,\wset_2$ be the collections of sets naturally defined by Steiner nodes in $V_1, V_2$, respectively. We consider the set cover instances $(\wset_1,\ulow)$ and $(\wset_2,\ulow)$ separately. 

We first estimate the value of $|\ulow|-\setcover(\wset_2,\ulow)$. Since $\card{\wset_2} = o(\card{U})$, in order to obtain a $(2,\eps \card{U})$-estimate of $\card{U}-\setcover(\wset_2,\ulow)$, it is sufficient to estimate $|\bigcup_{W\in \wset_2}W|$ to within an additive factor of $\eps |U|$.
This can be done in a similar way as estimating $|U|$.
Specifically, we pick a random terminal $u\in \ulow$ and run the procedure $\BFS(u,i)$. If the level-$i$ connected component that contains $u$ is small, then we set $X(u)$ to be the inverse of the size of this component; otherwise we set $X(u)=0$. So the random variable $X$ is supported on $[0,1]$ and $\mathbb{E}[X]=|\bigcup_{W\in \wset_2}W|/|\tlow|$. From Chernoff bound, if we repeat the experiment for $100k\log n/(\eps M)$ times, then we can obtain an estimate of $|\bigcup_{W\in \wset_2}W|$ to within an additive factor of $\eps M$. The query complexity of this step is $\tilde{O}(k\cdot k/M)=\underline{\tilde{O}(k\cdot n^{1/7})}$.

We next estimate the value of $|U|-\setcover(\wset_1,\ulow)$. We proceed similarly as $\algsetcover$, by defining an auxiliary graph $H$ on $U$ and estimate its maximal matching size. 
Similar to \Cref{sec: set cover}, we only need to design a subroutine for finding all neighbors of a given element in $H$.
This can be done as follows. Note that an element in $H$ corresponds to a small level-$i$ connected component. We first find all Steiner nodes in $V_1$ that is close to (at distance at most $(3/5)\cdot(1+\eps)^i$ from) the component, and then find all terminals that are close to each of these Steiner nodes. Since any small components has at most $\tilde{O}(1)$ representatives, the number of such terminals is at most $R\cdot P=\tilde{O}((n^{1/7}/\eps) \cdot n^{2/7}) = \tilde{O}(n^{3/7}/\eps)$. For each of these terminals, we run the procedure $\BFS(\cdot,i)$ on it to figure out if it indeed lies in a small component or not, and if the answer is yes, the component that contains the terminal is counted as a neighbor in $H$. The query complexity for all $\BFS$ procedures is $\underline{\tilde{O}(n^{3/7}\cdot k/\eps)}$. Lastly, via similar arguments, we can show that the query complexity of implementing the maximal matching estimation algorithm on $H$ is $O(RP\cdot RPk)=\underline{\tilde O(n^{6/7}\cdot k)}$.


We note that it is immediate to generalize above procedure to handle the cardinality-$2$ sets, since when estimating $|U|-\setcover(\wset_1,\ulow)$, we have figured out all elements in each explored set (and will be able to discard the set whenever it contains exactly two elements in $U$).

\subsubsection*{Simulation of Step 3}
We now describe the simulation of Step 3, the four-vertex subroutine.
Recall that, with the knowledge of terminal-induced metric, we constructed a laminar family and its partitioning tree $\tau$ to represents its hierarchical structure, and we only aim to find a node $x_S\in V(\tau)$, such that $x_S$ has exactly two children in $\tau$ and each child of $x_S$ also has exactly two children in $\tau$. Consider such a node $x_S$ at level $i$.
Note that, if a child of $x_S$ splits (into two child nodes) at a lower-than-$(i-\log_{1+\eps} (1/\eps_0))$ level, then the advantage obtained within this child node is at most $\eps_0 (1+\eps)^i$, and it is safe to ignore it. Thus, we can run $\BFS$ on all terminals $S$ for all levels between $i-\log_{1+\eps} (1/\eps_0)$ and $i$ to figure out the hierarchical structure of $S$ between these levels and calculate $\adv(S)$. The query time is $\tilde{O}(k)$.

The additional steps needed in simulating Step 3 are similar to that of Step 2, we first find the first $(2ML \log n/\eps)$ terminals in $R_{i-1}$. If $\card{R_{i-1}} \le 2ML \log n/\eps$, than we already figured out all level-$i$ connected components, and we then simply proceed as before:  sample $O(\log n /\eps^{10})$ small components and calculate the advantage of them. Otherwise, similar to Step 2, we first estimates $|U|$. Either we correctly establish that $\card{U}<M$, in which case level $i$ is light and can be safely ignored; or we correctly establish that $\card{U}>M/2$. Now we sample $O(k\log n /(\eps^{10} M))$ terminals, and for each sampled terminal $u$, we first run $\BFS(u,i-1)$ to figure out the component $S$ that contains $u$. If $S$ is a small component, we calculate $\adv(S)$, and add it to $B_i$ with probability $1/|S|$. The process is the same as sampling $\log n /\eps^{10}$ small components in $U$ and sum up the advantage of them. So $B_i$ is a good approximation of $A_i$ in this case as well. The total query complexity is $\tilde{O}(nk/M) = \tilde{O}(n^{1/7}k)$.

$\ $

Altogether, the query complexity is $\tilde O(n^{12/7}+n^{6/7}\cdot k)$.

\section{An $\tilde \Omega(nk)$ Lower Bound for $(2-\eps)$-Approximate Steiner Tree}
\label{sec: proof of beat-2-lower-computing}

In this section, we provide the proof of \Cref{thm: beat-2-lower-computing} by showing that any randomized algorithm that computes a $(2-\eps)$-approximate Steiner Tree performs at least $\Omega(nk)$ queries in the worst case. Throughout this section, we assume that $k\le n/100$.

We first construct a distribution on metric Steiner Tree instances $(V,T,w)$ as follows. The vertex set $V$ and the terminal set $T$ are fixed (recall that $|V|=n$ and $|T|=k$). The terminal set is partitioned into $t=\floor{\frac{k}{\floor{1/\eps}}}$ sets $T=\bigcup_{1\le i\le t}T_i$, such that each set $T_i$ contains either $\floor{1/\eps}$ or $\ceil{1/\eps}$ terminals. This partitioning is fixed as well. The only randomized part is the weight-metric $w$. Let $X$ be a set chosen uniformly at random from all size-$t$ subsets of $V\setminus T$. We call vertices in $X$ \emph{crucial} vertices. 
We then choose a random one-to-one mapping from $X$ to $[t]$, and for each $i\in [t]$, we call the vertex in $X$ that is mapped to $i$ the \emph{$i$-crucial} vertex. The random metric $w$ is defined according to the random set $X$ as follows. For every $i\in [t]$, the weight between the $i$-crucial vertex in $X$ and every terminal in $T_i$ is $1$, and the weight between any other pair of vertices in $V$ is $2$.
It is easy to verify that $w$ always satisfies the triangle inequality.

We call edges that are incident to crucial vertices \emph{crucial edges}. Clearly, crucial edges form $t$ disjoint stars.
It is easy to verify that, although $w$ is random, any two realizations of $w$ are isomorphic, and so the metric Steiner Tree cost is always the same. In particular, the optimal Steiner tree contains all crucial edges and $(t-1)$ other edges connecting the star graphs formed by crucial edges, and so $\stcost(V,T,w)=k\cdot 1+(t-1)\cdot 2=k+2t-2$.
On the other hand, any Steiner tree that contains $k'$ crucial edges has to contain at least $(t-1)+(k-k')$ other edges in order to span all terminals, and its total cost is at least $k'\cdot 1+\big((t-1)+(k-k')\big)\cdot 2=2k-k'+2t-2$. Therefore, in order to compute a Steiner tree of cost $(2-4\eps)\cdot \stcost(V,T,w)$, which is a Steiner tree of cost $(2-4\eps)(k+2t-2)$, an algorithm has to finds at least $(2k+2t-2)-(2-4\eps)(k+2t-2)=4\eps k -2t\ge \eps k$ edges.

Observe that, from the construction of $w$, if a terminal $u\in T_i$ has a weight-$1$ edge connecting to a Steiner vertex, then that Steiner vertex is the $i$-crucial vertex in $X$ and every other terminal in $T_i$ is also connected to it by weight-$1$ edges. For ease of analysis, we consider the following distribution of instances $(V',T',w')$.
The terminal set $T'$ contains, for each $i\in [t]$, a terminal $u_i\in T_i$, and it is fixed. The vertex set $V'=T'\cup (V\setminus T)$ and is also fixed. The metric $w'$ is random and defined as follows. Let $X$ be a set chosen uniformly at random from all size-$t$ subsets of $V\setminus T$, and we then choose a random one-to-one mapping from $X$ to $[t]$. We define $i$-crucial vertices similarly as before. For every $i\in [t]$, the weight between the $i$-crucial vertex in $X$ and terminal $u_i$ is $1$, and the weight between any other pair of vertices in $V'$ is $2$. It is easy to observe that a size-$t$ set $X$ and a mapping from $X$ to $[t]$ defines a metric $(V,T,w)$ and a metric $(V',T',w')$, and a query in either metric can be simulated by a query in the other. Additionally, in order to find at least $\eps k$ crucial edges in $w$, it is necessary to find $(\eps k)/\ceil{1/\eps}\ge \eps^2k/2$ crucial edges in $w'$.

We say that an crucial edge $(u,x)$ is \emph{discovered} by an algorithm at some step iff the set of queries performed by the algorithm before this step uniquely identify the edge $(u,x)$ to be a crucial edge. We say that a terminal is discovered iff its (unique) incident crucial edge is discovered, otherwise we say it is \emph{undiscovered}.
From Yao's minimax principle \cite{yao1977probabilistic} and the above discussion, in order to prove \Cref{thm: beat-2-lower-computing}, it suffices to prove the following lemma.

\begin{lemma}
\label{clm: crucial edge}
Any deterministic algorithm that discovers in expectation at least $\eps^2 k/2$ crucial edges in the distribution on $w'$ defined above performs at least $\Omega(\eps^2 nk)$ queries in expectation.
\end{lemma}

In the remainder of this section, we provide the proof of \Cref{clm: crucial edge}. From now on, we only consider algorithms that perform queries to the metric $w'$ instead of the metric $w$. We follow the framework used in the proof of Lemma 5.3 in \cite{assadi2019sublinear}, which shows that any randomized algorithm that outputs an approximate maximal matching performs at least $\Omega(n^2)$ queries in the worst case.

Consider a deterministic algorithm and the sequence of queries it performs. We partition the sequence into \emph{phases} as follows. The first phase starts at the first query, a phase ends as soon as a crucial edge is discovered, and the next phase starts right after the previous phase ends. For each integer $j$, let $Z_j$ be the random variable denoting the number of queries performed in the $j$-th phase. In order to prove \Cref{clm: crucial edge}, it suffices to show that $\mathbb{E}[\sum_{1\le j\le \eps^2 k} Z_i]=\sum_{1\le j\le \eps^2 k}\mathbb{E}[Z_i]=\Omega(\eps^2 nk)$.

From the construction of $w$, the weight between any pair of Steiner vertices and the weight between any pair of terminals is always $2$, so the only potentially useful queries are the ones between a terminal and a Steiner vertex. 
We define the \emph{uncertainty} of a terminal $u\in T'$ at some step to be $(n-k)$ minus the number of queries involving $u$ performed by the algorithm so far.
We say that a phase is \emph{bad}, iff at the start of the phase, there exists an undiscovered terminal whose uncertainty is below $3n/4$; and we say that a phase is \emph{good}, iff at the start of the phase, the uncertainty of every undiscovered terminal is at least $3n/4$. The proof of \Cref{clm: crucial edge} is concluded by the following claims.

\begin{claim}
The expected number of queries in a good phase is at least $n/200$.
\end{claim}
\begin{proof}
We use the following lemma, which is similar to Lemma 5.4 in \cite{assadi2019sublinear}.
\begin{lemma}
\label{lem: random edge}
Let $H=(A,B,E)$ be a bipartite graph with $|A|=a$ and $|B|=b$ (where $b> 10a$), such that every vertex in $A$ has degree at least $2b/3$, then for every edge $e\in E$, the probability that $e$ is contained in a uniformly at random chosen $A$-perfect matching in $G$ is at most $2/b$ (here an $A$-perfect matching is a matching that matches all vertices of $A$).
\end{lemma}
\begin{proof}
It is easy to verify from Hall's Theorem that $H$ always contains a perfect matching. Consider an edge $(u,v)\in E$ where $u\in A$ and $v\in B$. Let $M$ be an $A$-perfect matching that contains edge $(u,v)$. We construct $b/2$ other $A$-perfect matchings as follows. Let $B'$ be the set of vertices in $B\setminus \set{v}$ that is adjacent to $u$ in $H$ but is not an endpoint of any edge in $M$. Since $\deg_H(u)\ge 2b/3$ and $b> 10a$, 
we get that $|B'|\ge 2b/3-1-b/10\ge b/2$. For each vertex $v'\in B'$, we define $M_{v'}$ to be the matching obtained from $M$ by replacing edge $(u,v)$ with edge $(u,v')$. Clearly, $M_{v'}$ is an $A$-perfect matching for every $v\in B'$, so we obtained a collection of at least $b/2$ other $A$-perfect matchings that do not contain edge $(u,v)$, that we denote by $\fset(M)$. On the other hand, for every pair $M,M'$ of distinct $A$-perfect matchings in $H$ that contains the edge $(u,v)$, the collections $\fset(M),\fset(M')$ are disjoint. This is because for every matching $\hat M\in \fset(M)$ and for every matching $\hat M'\in \fset(M')$, $\hat M$ and $\hat M'$ must differ on an edge not incident to $u$, as $M$ and $M'$ do. Therefore, the number of perfect $A$-matchings in $H$ is at least $(b/2)$ times the number of perfect $A$-matchings in $H$ that contains the edge $(u,v)$, and the lemma now follows.
\end{proof}
	
By definition, at the start of a good phase, the uncertainty of every undiscovered terminal is at least $3n/4$. Therefore, for the next $n/24$ queries, it is easy to verify that the graph induced by all unqueried edges satisfy the conditions of \Cref{lem: random edge}, and so the probability that any single query finds an edge of weight $1$ is at most $2/(n-k)\le 3/n$. Thus, with probability at least $(n/24)\cdot (3/n)=1/8$, none of the first $n/24$ queries of a good phase finds a weight-$1$ edge. This implies that the expected number of queries in a good phase is at least $(1/8)\cdot (n/24)\ge n/200$.
\end{proof}

\begin{claim}
If the algorithm performs less than $\eps^2 nk/20$ queries, then the number of bad phases is at most $\eps^2 k/4$.
\end{claim}
\begin{proof}
Since any useful query is between a terminal and a Steiner vertex, a useful query reduces the uncertainty of exactly one terminal by $1$. Therefore, in order to have $\eps^2 k/40$ bad phases, the reduction in the total uncertainty is at least $(\eps^2 k/4)\cdot (n-k-3n/4)\ge \eps^2 nk/200$, which implies that the algorithm performs at least $\eps^2 nk/200$ queries.
\end{proof}

From the above two claims, in order to find $\eps^2 k/2$ crucial edges, either the algorithm performs at least $\eps^2 nk/20$ queries, or the number of bad phases encountered by the query sequence is at most $\eps^2 k/4$. Thus, in the first $\eps^2 k/2$ phases, there are at least $\eps^2 k/4$ good phases, implying that the expected number queries made by the algorithm is at least $(\eps^2 k/4)\cdot (n/200)=\Omega(\eps^2 nk)$.
This completes the proof of \Cref{clm: crucial edge}, and therefore also completes the proof of  \Cref{thm: beat-2-lower-computing}.

\section{Upper and Lower Bounds for $\alpha$-Approximate Steiner Tree ($\alpha \ge 2$)}
\label{sec: >2-main}

In this section we provide the proof of \Cref{thm: >2-main}. 

\subsection{Upper Bound}
\label{subsec: >2_upper}

We start by presenting an $\tilde{O}(k^2/\alpha)$ query algorithm for any $\alpha\ge 2$. The query algorithm for any $\alpha \ge 2$ is very similar to the algorithm in the proof of Theorem 1 in \cite{chen2022sublinear}, which shows a one-pass $\tilde O(n/\beta)$ streaming algorithm for estimating the metric MST cost to within factor $\beta$ (for any $\beta>1$). 
Note that we may assume without lose of generality that $\alpha=\Omega(\log^2 n)$, as otherwise we can simply query all terminal-terminal weight (which is $k^2=\tilde O(k^2/\alpha)$ queries) and then compute the minimum spanning tree on $T$, which is a $2$-approximate Steiner Tree (and therefore an $\alpha$-approximate Steiner Tree as $\alpha\ge 2$).

\paragraph{Algorithm.} 
Let $\beta=\alpha/(100\log n)$.
We first choose a uniformly at random size-$\ceil{k/\beta}$ subset of $T$, and denote it by $T'$.
We then query all weights between pairs of terminals in $T'$, and use the acquired information to compute the minimum spanning tree $\tau'$ over $T'$. Lastly, for every terminal $u\notin T$, we query all weights between $u$ and terminals in $T'$, and let $f(u)=\arg_{u'\in T'}\min\set{w(u,u')}$. Finally, we output the tree $\tau$ defined as $\tau=\tau'\cup\set{(u,f(u))\mid u\in T\setminus T'}$.

On the one hand, observe that the algorithm only queries weights with one endpoint in $T'$ and the other endpoint in $T$, so the number of queries performed by the algorithm is at most $k\cdot (k/\beta)=\tilde O(k^2/\alpha)$. 
On the other hand, it is easy to verify that the algorithm always outputs a spanning tree on $T$.
The proof that the spanning tree $\tau$ output by the algorithm above is with high probability an $\alpha$-approximate Steiner Tree is similar to the analysis on pages 15-16 in \cite{chen2022sublinear}, and is deferred to \Cref{apd: Analysis of >2_upper}.

\subsection{Lower Bound}

We now prove an $\Omega(k^2/\alpha)$ lower bound for computing an $\alpha$-approximate Steiner Tree for any $\alpha \ge 2$. In fact, we will show that, if all vertices in the metric space are terminals (so there are $k$ terminals and no Steiner vertices), then for every $\alpha\ge 2$, computing a spanning tree of cost at most $\alpha$ times the minimum spanning tree cost requires at least $\Omega(k^2/\alpha)$ queries. Since the metric Steiner Tree cost is at most twice the minimum terminal spanning tree cost, this lower bound implies the lower bound in \Cref{thm: >2-main} (as we can construct a metric Steiner Tree instance, where all Steiner vertices are sufficiently far from all terminals and so any $\alpha$-approximate Steiner tree may only terminal-terminal edges, and is therefore a terminal spanning tree).

It is easy to observe that it suffices to consider the case where $\alpha\ge 2$ is an integer and $k$ is divisible by $100\alpha$.
We generate a metric space $w$ from the following distribution. The vertex set $V=T$ is fixed. Let $\pset$ be a partitioning of $T$, that is chosen uniformly at random from all partitioning of $T$ that partition $T$ into $t=k/100\alpha$ sets $T_1, \ldots, T_t$ of size $100\alpha$ each. For every pair $u,u'$ of terminals in the same part of $\pset$, $w(u,u)=1$, and for all other pairs $u,u'$, $w(u,u)=2\alpha$.

We call the weight-$1$ edges as \emph{crucial} edges.
It is easy to verify that, although metric $w$ is random, any two realizations of $w$ are isomorphic, and so the minimum spanning cost is always the same. In particular, the minimum spanning tree contains $(k-t)$ crucial edges and $(t-1)$ other edges connecting the trees incide each partition, and so $\mst(w)=(k-t)\cdot 1+(t-1)\cdot 2\alpha=1.02k-t-2\alpha$.
On the other hand, for any spanning tree that contains $k'$ crucial edges and $(k-k'-1)$ other edges, its total cost is at least $k'\cdot 1+(k-k'-1)\cdot 2\alpha=2\alpha k-(2\alpha-1) k'-2\alpha$. Therefore, in order for a spanning tree $\tset$ to approximate the minimum spanning tree to within factor $\alpha$, the tree $\tset$ has to contain at least 
\[\frac{2\alpha k-2\alpha-\alpha\cdot (1.02k-t-2\alpha)}{2\alpha-1}= \frac{0.98\alpha k-2\alpha+\alpha t+2\alpha^2}{2\alpha}\ge 0.01k\] 
crucial edges, and so there must be at least $0.01k$ vertices in $T$, such that $\tset$ contains a crucial edge incident to it.

We say that an edge is \emph{discovered} by an algorithm at some step iff the set of queries performed by the algorithm before this step uniquely identify the edge to be a weight-$1$ edge or a weight-$2\alpha$ edge.
We say that a vertex $u$ is \emph{settled} iff the algorithm has discovered a crucial edge incident to $u$; otherwise we say that it is \emph{unsettled}. 
Similarly, we say that a part $T_i$ in the partitioning $\pset$ is settled iff all its vertices are settled. 
From Yao's minimax principle \cite{yao1977probabilistic} and the above discussion, in order to prove the lower bound of \Cref{thm: >2-main}, it suffices to prove the following lemma.

\begin{lemma}
\label{lem: settled vertices}
Any deterministic algorithm that settles at least $0.01k$ vertices in the distribution on $w$ defined above performs at least $\Omega(k^2/\alpha)$ queries in expectation.
\end{lemma}

In the remainder of this section, we provide a proof of \Cref{lem: settled vertices}. We use a similar approach as in \Cref{sec: proof of beat-2-lower-computing}. 
Consider a deterministic algorithm and the sequence of queries it performs.
Assume without loss of generality that the algorithm terminates whenever it settles $0.01k$ vertices.
We partition the sequence of queries it makes as follows.
The first phase starts at the first query, a phase ends as soon as a previously unsettled vertex is settled, and the next phase starts right after the previous phase ends. For each $j\ge 1$, we let $Z_j$ be the random variable denoting the number of queries performed in the $j$-th phase. In order to prove \Cref{lem: settled vertices}, it suffices to show that $\sum_{1\le j\le 0.01 k}\mathbb{E}[Z_i]=\Omega(k^2/\alpha)$.

We classify all phases into good ones and bad ones as follows.
We say that an unsettled vertex $u$ is \emph{well-discovered}, iff the number of parts $T_i$ in $\pset$ such that some edge in $E(u,T_i)$ has been discovered is at least $t/10$.
We say that a phase is \emph{type-1 bad}, iff at the start of this phase there exists an unsettled part $T_i\in \pset$, and such that all discovered edges in $E(T_i,T\setminus T_i)$ touch at least $k/10$ vertices outside $T_i$. 
We say that a phase is \emph{type-2 bad} iff at the start of this phase there exists a well-discovered vertex. If a phase is neither type-1 bad nor type-2 bad, then we say it is \emph{good}. 
For ease of analysis, whenever the algorithm starts a type-1 bad phase due to some part $T_i$, we immediately reveal to the algorithm the weight of all edges with at least one endpoints in $T_i$, so the part $T_i$ is settled right away; and whenever the algorithm starts a type-2 bad phase due to some well-discovered unsettled vertex $u$, we immediately reveal a crucial edge incident on $u$, so $u$ and the other endpoint of the revealed edge are settled right away.
We prove the following claims.

\begin{claim}
The expected number of queries in a good phase is at least $k/1600\alpha$.
\end{claim}
\begin{proof}
Recall that the metric $w$ is defined based on a partitioning $\pset$ of $T$ into $t=k/100\alpha$ subsets of size $100\alpha$ each (that we call a \emph{valid} partitioning). We say that a valid partitioning $\pset$ \emph{joins} a pair $u,u'$ of vertices in $T$, iff $u,u'$ lie in the same part of $\pset$, otherwise we say that $\pset$ \emph{separates} the pair $u,u'$.
We say that a valid partitioning $\pset$ of $T$ is \emph{consistent} with the current queries, iff (i) every discovered crucial edge has both its endpoints lying in the same part of $\pset$ and every discovered edge that is not crucial has its endpoints in different partitions; (ii) every unsettled vertex $u$ is not well-discovered; and (iii) for every unsettled part $T_i$ in $\pset$, all discovered edges in $E(T_i,T\setminus T_i)$ touch at least $k/10$ vertices outside $T_i$. In other words, a partitioning $\pset$ is consistent iff it provides consistent answers for all queries made by the algorithm, and the algorithm is not a bad phase given the current queries. 
Intuitively, from the algorithm's viewpoint, the up-to-date distribution (according to the answers to the queries performed so far) of the underlying partitioning should be the uniform distribution on all partitionings that are consistent with the current queries.
We prove the following observation.
\begin{observation}
\label{obs: good phase}
At any time during the first $k/(10\alpha)$ queries of a good phase, for every pair $u,u'$ of vertices in $T$ such that $u,u'$ are not both settled and the edge $(u,u')$ has not been queried yet, the number of consistent partitionings that separate the pair $u,u'$ is at least $k/800\alpha$ times the number of consistent partitionings that joins the pair $u,u'$.
\end{observation}
\begin{proof}
Consider a pair $u,u'$ of vertices and a consistent partitioning $\pset$ that joins the pair $u,u'$. Assume without loss of generality that $u$ is unsettled. We construct a collection of other partitionings as follows. Let $T_i$ be the part in $\pset$ that contains $u$ and $u'$. Let $T'(u)$ be the set of all unsettled terminals $\hat u$ such that (i) no edge in $E(\set{\hat u}, T_i)$ has been discovered; and (ii) no edge in $E(\set{u}, T_{i'})$ has been discovered, where $T_{i'}$ is the part in $\pset$ that contains $\hat u$. 
We prove the following observation.
\begin{observation}
At any time during the first $k/(10\alpha)$ queries of a good phase, for every unsettled vertex $u$, $|T'(u)|\ge k/2$.
\end{observation}
\begin{proof}
Before this good phase, $u$ is also unsettled, so the number of vertices $\hat u$ such that $(u,\hat u)$ has been discovered is at most $(t/10)\cdot (100\alpha)\le k/10$. Let $T_i$ be the part that $u$ lies in, since $T_i$ is not settled before this phase, it did not create a bad phase before, and so the number of vertices in $T\setminus T_i$ that are touched by discovered edges in $E(T_i,T\setminus T_i)$ is at most $k/10$. On the other hand, note that the algorithm settles at most $0.49$ vertices before this phase, and has performed performed at most $k/(10\alpha)$ queries in this phase.
By definition, $T'(u)$ contains all vertices $\hat u$ that does not satisfy any of the above conditions, so $|T'(u)|\ge k-k/10-k/10-0.01k-k/(10\alpha)\ge k/2$.
\end{proof}

For each $\hat u\in T'(u)$, consider the partitioning $\pset(\hat u)$ obtained by exchanging the positions of vertices $u$ and $\hat u$ (that is, if originally $u\in T_i$ and $\hat u\in T_{i'}$, then we move $u$ to $T_{i'}$ and move $\hat u$ to $T_{i}$). Clearly, vertices $u,u'$ are separated in $\pset(\hat u)$. We prove the following observation.

\begin{observation}
For every $\hat u\in T'(u)$, $\pset(\hat u)$ is a consistent partitioning.
\end{observation}
\begin{proof}
Consider a vertex $\hat u\in T'(u)$. Let $T_i$ be the set that contains $u$ and let $T_{i'}$ be the set that contains $\hat u$. By definition of $T'(u)$, no edge in $E(\set{\hat u}, T_i)$ has been discovered, and no edge in $E(\set{u}, T_{i'})$ has been discovered. On the other hand, since $u$ and $\hat u$ are not settled, no edge in $E(\set{\hat u}, T_{i'})$ has been discovered, and no edge in $E(\set{u}, T_{i})$ has been discovered. Therefore, $\pset(\hat u)$ provides consistent answers with all queries made by the algorithm so far. On the other hand, it also implies that for every part $T_j$ in $\pset(\hat u)$, all discovered edges in $E(T_j,T\setminus T_j)$ touch the same set of vertices as the corresponding part in $\pset$. Lastly, it is also easy to verify that every unsettled vertex is still not discovered in $\pset(\hat u)$. Altogether, $\pset(\hat u)$ is a consistent partitioning with all current queries.
\end{proof}

For every $\hat u\in T'(u)$, we say that $\pset(\hat u)$ is a \emph{host} of $\pset$, so $\pset$ has at least $k/2$ hosts.
On the other hand, for every partitioning $\pset'$ that separates $u,u'$, there are at most $200\alpha$ partitioning $\pset$, such that $\pset$ joins $u,u'$ and $\pset'$ is a host of $\pset$ (since in $\pset'$, $u,u'$ belong to different parts, say $u\in T_i$ and $u'\in T_{i'}$ so it must be the case that either $u$ and some vertex in $T_{i'}$ are exchanged or $u'$ and some vertex in $T_{i}$ are exchanged, a total of at most $200\alpha$ possibilities).
Therefore, the number of consistent partitionings that separate the pair $u,u'$ is at least $(k/2)/(200\alpha)=k/400\alpha$ times the number of consistent partitionings that joins the pair $u,u'$.
\end{proof}

From \Cref{obs: good phase}, among the first $k/800\alpha$ queries in a good phase, the probability that any single query finds a crucial edge is at most $400\alpha/n$, and so with probability $1/2$, none of them finds a crucial edge. This implies that the expected number of queries in a good phase is at least $(k/800\alpha)\cdot (1/2)=k/1600\alpha$.
\end{proof}

\begin{claim}
If the algorithm performs less than $k^2/(10^7\alpha)$ queries, then the number of type-1 bad phases is at most $k/(10^5\alpha)$, and the number type-2 bad phases is at most $k/10^3$.
\end{claim}
\begin{proof}
Assume for contradiction that the number of type-1 bad phases is more than $k/(10^5\alpha)$. Then there are at least $k/(10^5\alpha)$ indices $i$, such that the number of edges in $|E(T_i,T\setminus T_i)|$ discovered by the algorithm is at least $k/10$. On the other hand since the algorithm never settles more than $0.01k$ vertices, it never discovers more than $0.01k$ crucial edges. Therefore, the number of queries performed by the algorithm is at least \[\frac{1}{2}\cdot \frac{k}{10^5\alpha}\cdot \bigg(\frac{k}{10}-\frac{k}{10^5\alpha}\cdot 100\alpha-0.01k\bigg)\ge \frac{k^2}{10^7\alpha},\] a contradiction.
Assume for contradiction that the number of type-2 bad phases is more than $k/10^3$, then there are more than $k/10^3$ well-discovered vertices. 
Therefore, the number of queries performed by the algorithm over all phases is at least $(k/10^3)\cdot (t/10-k/(10^5\alpha))\cdot (1/2)\ge k^2/(10^7\alpha)$, a contradiction.
\end{proof}

From the above two claims, we get that, in order to settle $0.01k$ vertices, either the algorithm performs at least $k^2/(10^7\alpha)$ queries, or the number of type-1 bad phases of its query sequence is at most $k/(10^5\alpha)$ (which settles at most $0.001k$ vertices in total) and the number of type-2 bad phases of its query sequence is at most $k/10^3$ (which settles at most $0.002k$ vertices in total), and so among the first $0.01k$ phases, there are at least $0.007k$ good phases, implying that the expected number queries made by the algorithm is at least $0.007k\cdot (k/1600\alpha)=\Omega(k^2/\alpha)$.
This completes the proof of \Cref{lem: settled vertices}, and therefore also completes the proof of the lower bound in \Cref{thm: >2-main}.

\newcommand{\DY}{\mathcal{D}_{\textnormal{\textsf{Y}}}}
\newcommand{\DN}{\mathcal{D}_{\textnormal{\textsf{N}}}}
\newcommand{\SP}{special\xspace}
\newcommand{\NM}{regular\xspace}
\newcommand{\IM}{secret\xspace}
\newcommand{\DD}{\mathcal{D}}
\newcommand{\isety}{\iset_{\sf Y}}
\newcommand{\isetn}{\iset_{\sf N}}

\section{An $\tilde \Omega(n+k^{6/5})$ Query Lower Bound for $(2-\eps)$-Estimation}
\label{sec: beat-2-lower}

In this section we provide the proof of \Cref{thm: beat-2-lower-main}, by showing that, for any $0<\eps<1$, any randomized algorithm that with probability $2/3$ estimates the Steiner Tree cost to within a factor of $2-\eps$ performs $\Omega(n+k^{6/5})$ in the worst case.
We will prove a lower bound of $\Omega(n)$ in \Cref{sec: n lower} and a lower bound of $\Omega(k^{6/5})$ in \Cref{sec: k^{6/5} lower}. Combined together, they complete the proof of \Cref{thm: beat-2-lower-main}. Throughout the section, we assume that $1/(2\eps)<k<n/2$.

\subsection{An $\Omega(n)$ Lower Bound}
\label{sec: n lower}

We construct a distribution $\dset$ of metric Steiner Tree instance $(V,T,w)$ as follows. The set $T$ contains $k$ terminals and the vertices in $V\setminus T$ are denoted by $v_1,\ldots,v_{n-k}$. 
We define $w_0$ as the metric on $V$ where $w_0(v,v')=2$ for all pairs $v,v'\in V$.
For each $1\le j\le n-k$, we define a metric $w_j$ as follows:
\begin{itemize}
\item For each terminal $u\in T$, $w_j(u,v_j)=1$.
\item For every other pair $v,v'\in V$, $w_j(v,v')=2$.
\end{itemize}
In other words, the weight-$1$ edges in $w_j$ form a star graph where $v_j$ is its center and all terminals are its leaves. It is easy to verify that $w_j$ is indeed a metric. We then define the instance $I_j=(V,T,w_j)$, and the distribution $\dset$ is defined as follows: $\Pr[I_0]=1/2$, and for each $1\le j\le n-k$, $\Pr[I_j]=1/(2(n-k))$.

Clearly, $\stcost(I_0)=2(k-1)$, and for each $j$, $\stcost(I_j)=k$ as the star graph formed by all weight-$1$ edges is a Steiner Tree of cost $k$. Since $k>1/(2\eps)$, $\stcost(I_0)/\stcost(I_j)>2-\eps$ for all $j$, and so estimating the metric Steiner Tree cost of a random instance sampled from $\dset$ to within factor $(2-\eps)$ is equivalent to determining whether or not the random instance is $I_0$.

In order to correctly determining if a random instance sampled from $\dset$ is $I_0$ with probability at least $2/3$, it is necessary that the algorithm discovers a weight-$1$ edge on at least $(1/3)$-fraction of the instances $I_1,\ldots,I_{n-k}$.
From Yao's minimax principle \cite{yao1977probabilistic}, the following claim implies a $\Omega(n-k)=\Omega(n)$ lower bound on the query complexity of any randomized algorithm, as $k\le n/2$.

\begin{claim}
Any deterministic algorithm that discovers a weight-$1$ edge on at least $1/3$-fraction of the instances $I_1,\ldots, I_{n-k}$ performs at least $\Omega(n)$ queries in expectation.
\end{claim}

\begin{proof}
Observe that, in each of the instances $I_1,\ldots, I_{n-k}$, all weight-$1$ edges are incident to one Steiner vertex, that we call the \emph{secret vertex}, and the secret vertex is $v_1,\ldots,v_{n-k}$ with probability $1/(n-k)$ each.
Since a single query can only check one Steiner vertex (by querying any edge connecting it to a terminal), in order to find the secret vertex on at least $1/3$-fraction of the instances $I_1,\ldots, I_{n-k}$, any algorithm needs to perform at least $(n-k)/3=\Omega(n-k)$ queries in expectation.
\end{proof}

\subsection{An $\Omega(k^{6/5})$ Lower Bound}
\label{sec: k^{6/5} lower}

We will construct a pair $\DY,\DN$ of distributions on metric Steiner Tree instances, such that $\DY$ is only supported on instances $(V,T,w)$ with $\stcost(V,T,w)$ close to $Z$, while $\DN$ is only supported on instances instances $(V,T,w)$ with $\stcost(V,T,w)$ close to $2Z$, where $Z$ is some function of $k$ that will be defined later.
We let $\dset=(\DY+\DN)/2$ be the average distribution of $\DY$ and $\DN$. 
We show that, in order to report correctly with probability at least $2/3$ whether an random instance sampled from $\dset$ comes from $\DY$ or $\DN$, any randomized algorithm has to perform at least $\tilde{\Omega}(k^{6/5})$ queries, thereby proving \Cref{thm: beat-2-lower-main}.


We now proceed to define the distributions $\DY$ and $\DN$.
We first define an \emph{auxiliary instance} $(V,T,w)$ as follows. For convenience, we let set $T$ contain $k+k^{2/5}/\eps$ terminals instead of $k$ terminals. As we will see, this does not influence our lower bound, as $(k+k^{2/5}/\eps)^{6/5}=\Theta(k^{6/5})$.

\begin{itemize}
\item The set $T$ is partitioned into subsets $T=\big(\bigcup_{1\le j\le d}S_i\big)\cup \big(\bigcup_{1\le i\le d'}T_i\big)$, where $d=k^{2/5}$, $d'=k^{3/5}$; for each $1\le j\le d$, $|S_j|=1/\eps$, and for each $1\le i\le d'$, $|T_i|=k^{2/5}$.
\begin{itemize}
\item We call sets $S_1,\ldots,S_d,T_1,\ldots,T_{d'}$ \emph{groups}.
\item We denote $S=\bigcup_{1\le j\le d}S_j$, we call terminals in $S$ \emph{special} terminals, and we call terminals in $T\setminus S$ \emph{regular} terminals.
    For each $1\le j\le d$, we denote $S_j=\set{s_{j,1},\ldots,s_{j,1/\eps}}$. Note that $\card{S} = k^{2/5}/\eps$.
\end{itemize}  
\item The set $V\setminus T$ contains $n-(k+k^{2/5}/\eps)$ Steiner vertices and is further partitioned into $d'+1$ subsets $V\setminus T=\bigcup_{0\le i\le d'}V_i$, where for each $1\le i\le d'$, $|V_i|=k^{2/5}/\eps$.
\item The metric $w$ is defined as follows:  
\begin{itemize}
\item For every pair $u,u'$ of regular terminals that belong to the same group, $w(u,u')=0$.
\item For each $1\le i\le d'$, we fix an arbitrary perfect matching $M_i$ between terminals in $S$ and vertices in $V_i$. For every matched pair $u,v'$, $w(u,v)=1$.
\item The weight between every other pair of vertices in $V$ is $2$.
\end{itemize}
\end{itemize}

\begin{figure}[h]
\centering
\includegraphics[scale=0.14]{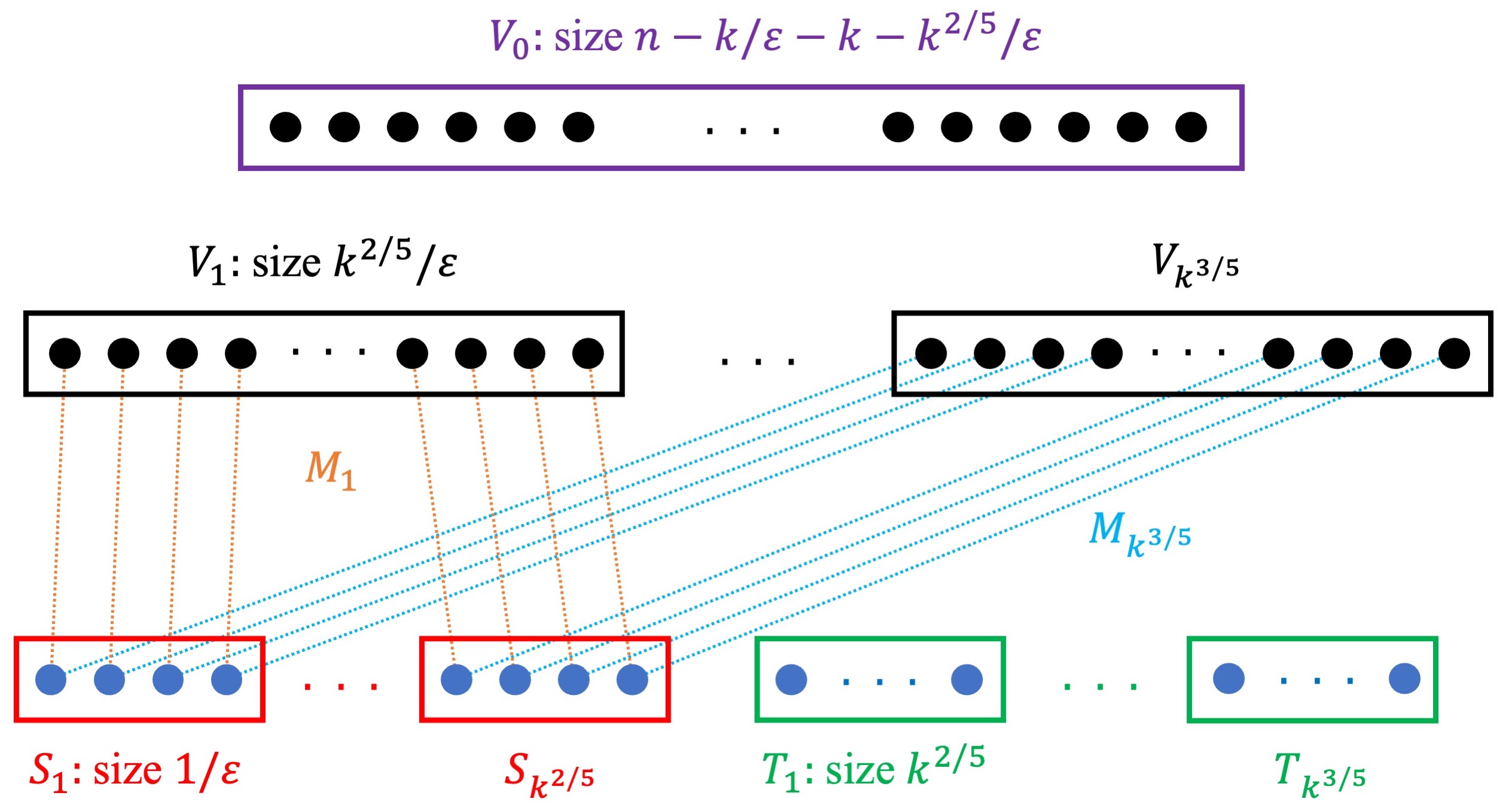}
\caption{An illustration of the metric $w$. All terminals are shown in dark blue and all Steiner vertices are shown in black. Matchings $M_1,\ldots,M_{k^{3/5}}$ are shown in dashed lines and the matched pairs have weight $1$. Pairs of teminals in the same $T_i$ (green box) have weight $0$. All other pairs have weight $2$.
\label{fig: metric_w}}
\end{figure}
See \Cref{fig: metric_w} for an illustration. It is easy to verify that $w$ is indeed a metric.

We now use the auxiliary instance defined above to construct distributions $\DY$ and $\DN$.
Every instance with non-zero probability in either $\DY$ or $\DN$ has the same vertex set $\hat V$ and the same terminal set $\hat T$, where $|\hat V|=|V|$ and $|\hat T|=|T|$. The set $\hat V\setminus \hat T$ of Steiner vertices is further partitioned into subsets $\hat V\setminus \hat T=\bigcup_{0\le i\le d'}\hat V_i$, where $V_0=\hat V_0$, and for each $1\le i\le d'$, $|\hat V_i|=|V_i|$.
We say that a one-to-one mapping $f: \hat V\to V$ is \emph{valid}, iff $f$ maps terminals in $\hat T$ to terminals in $T$, $f$ maps every vertex in $\hat V_0$ to itself in $V_0$, and for each $1\le i\le d'$, $f$ maps vertices in $\hat V_i$ to vertices in $V_i$.
Let $\fset$ be the set of all valid mappings.

We first define the distribution $\DN$.
For each mapping $f\in \fset$, we define an instance $I_f=(\hat V, \hat T, \hat w_f)$, where the metric $\hat w_f$ is defined as follows: for every pair $v,v'\in \hat V$, $\hat w_f(v,v')=w(f(v),f(v'))$.
The distribution $\DN$ is simply defined to be the uniform distribution over all instances in $\isetn=\set{I_f\mid f\in \fset}$.

We now define the distribution $\DY$. Consider a mapping $g: [d]\to [d']$. 
We first define another auxiliary metric $w_{g}$ on $V$ as follows.
\begin{itemize}
\item For each $1\le j\le d$, we consider the matching $M_{g(j)}$ between $S$ and $V_{g(j)}$. If we denote by $v^*_{g(j)}$ the Steiner vertex in $V_{g(j)}$ matched with terminal $s_{j,1}$, then for each $u\in S_j$, $w_{g}(u,v^*_{g(j)})=1$, and the weight in $w_{g}$ between $u$ and its matched Steiner vertex in $M_{g(j)}$ is $2$.
\item For every other pair $v,v'\in V$, $w_{g}(v,v')=w(v,v')$.
\end{itemize}
\begin{figure}[h]
\centering
\includegraphics[scale=0.14]{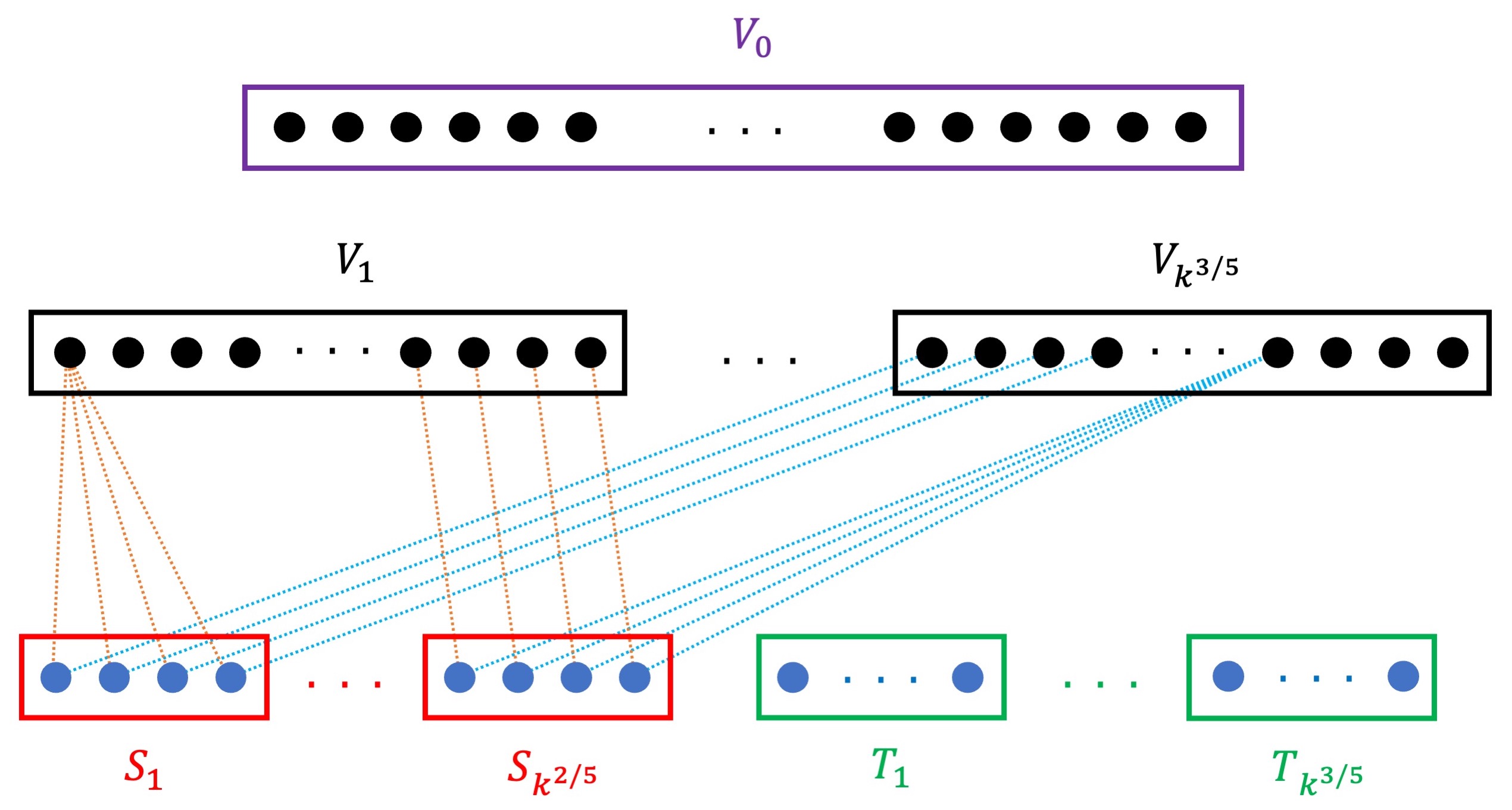}
\caption{An illustration of the metric $w_g$ (where $g(1)=1$ and $g(k^{2/5})=k^{3/5}$). The only difference between metrics $w$ and $w_g$ are the weight-$1$ pairs, which are shown in dashed lines.
\label{fig: metric_wg}}
\end{figure}
See \Cref{fig: metric_wg} for an illustration. It is easy to verify that $w_g$ is a metric. The only difference between metrics $w$ and $w_g$ are the weight-$1$ pairs. In particular, for each $S_i$, there is a star graph consisting of weight-$1$ edges in $w_g$ that spans all terminals in $S_i$, while there is no such graph in $w$. These star graphs are the reason that the costs $\stcost(V,T,w)$ and $\stcost(V,T,w_g)$ are roughly separated by factor $2$.

For every mapping $f\in \fset$, we define a metric Steiner Tree instance $I_{(f,g)}$ as $I_{(f,g)}=(\hat V, \hat T, \hat w_{(f,g)})$ where the metric $\hat w_{(f,g)}$ is defined as: for every pair $v,v'\in \hat V$, $\hat w_{(f,g)}(v,v')=w_g(f(v),f(v'))$.
The distribution $\DY$ is simply defined to be the uniform distribution on all instances in $\isety=\set{I_{(f,g)}\mid f\in \fset, g\in \gset}$, where $\gset$ is the collection of all mappings from $[d]$ to $[d']$.

On the one hand, it is easy to verify that every instance with non-zero probability in $\DN$ is isomorphic to $(V,T,w)$, and so it has the same Steiner Tree cost as $(V,T,w)$. 
Similarly, it is easy to verify that every instance with non-zero probability in $\DY$ is isomorphic to $(V,T,w_g)$ for an arbitrary $g\in \gset$ (and in fact all instances $\set{(V,T,w_g)}_{g\in \gset}$ are isomorphic to each other), and so it has the same Steiner Tree cost as $(V,T,w_g)$. We next show that $\stcost(V,T,w)$ is roughly two times $\stcost(V,T,w_g)$ for every $g$. Recall that $\isetn=\set{I_f\mid f\in \fset}$ and $\isety=\set{I_{(f,g)}\mid f\in \fset, g\in \gset}$.

\begin{claim} \label{prop:st-cost}
    Any instance in $\isetn$ has Steiner Tree cost at least $2k^{2/5}/\eps$. Any instance in $\isety$ has Steiner Tree cost at most $k^{2/5}(1/\eps+4)$. Hence the cost of any instance in $\isetn$ is more than $(2-8\eps)$ times the cost of any instance in $\isety$.
\end{claim}

\begin{proof}
    For any instance in $\isetn$, any Steiner node has weight $1$ to at most $1$ terminals, so the Steiner Tree cost equals the spanning tree of terminals, which is $2(k^{2/5} + k^{2/5}/\eps -1) > 2k^{2/5}/\eps$. For any instance in $\isety$, any group of \SP terminnal $S_i$ have weight $1$ to a common Steiner node, thus we can connect them at a cost of $1/\eps$. Thus the Steiner Tree cost is at most $(1/\eps) k^{2/5} + 2(k^{2/5}+k^{2/5}-1) < (1+4\eps)(k^{2/5}/\eps)$. Moreover, it is easy to verify that the ratio of these costs is more than $(2-8\eps)$ for any $\eps > 0$.
\end{proof}

We then define distribution $\dset=(\DY+\DN)/2$, and consider the following problem: Given an instance sampled from $\dset$, estimate its value (Steiner Tree cost) to within a factor of $(2-8\eps)$.
From \Cref{prop:st-cost}, the problem is equivalent to the problem of determining a random instance (sampled from $\dset$) comes from $\DY$ or $\DN$.
If a randomized algorithm reports correctly with probability at least $2/3$, then we say that the algorithm \emph{distinguishes between $\DY$ and $\DN$}.
Therefore, in order to prove \Cref{thm: beat-2-lower-main}, it suffices to prove that any algorithm that distinguishes between $\DY$ and $\DN$ performs $\tilde{\Omega}(k^{6/5})$ queries in the worst case.

The remainder of this section is dedicated to the proof that any algorithm that distinguishes between $\DY$ and $\DN$ performs $\tilde{\Omega}(k^{6/5})$ queries in the worst case. Before we give the detailed proof, we provide some intuition.
From the construction of $\DY$, $\DN$, distinguishing between $\DY$ and $\DN$ is essentially distinguishing between the metric $w$ and the metric $w_g$ (for any $g$), where the identities of vertices are randomized.
The main difference between metrics $w$ and $w_g$ is that, in $w_g$, there exist Steiner vertices that are connected to more than one (actually $1/\eps$) special terminals with weight-$1$ edges. We call such Steiner vertices \emph{secret vertices}.
We now argue from a high level that finding a secret vertex requires $\Omega(k^{6/5})$ queries.
We call edges of weight $0$ or $1$ \emph{crucial} edges, since if a terminal $u$ is found incident to a crucial edge, then we can immediately tell if $u$ is \SP or \NM.

On the one hand, note that in the metric $w_g$, there are $k^{2/5}$ secret vertices. So if we sample a random Steiner vertex $v$ from $\bigcup_{1\le i\le d'}V_i$ (as the vertices in $V_0$ are irrelevant in distinguishing between metrics $w$ and $w_g$), then the probability that $v$ is a secret vertex is $O(k^{-3/5})$ as $|\bigcup_{1\le i\le d'}V_i|=\Omega(k)$, and so it takes $\Omega(k^{3/5})$ random samples to get a secret vertex. However, in order to certify that $v$ is indeed a secret vertex, we need to find at least $2$ crucial edges incident to it, which takes $\Omega(k)$ queries as every Steiner vertex is only incident to $O(1)$ crucial edges both in $w$ and $w_g$. Altogether, it takes $\Omega(k^{3/5})\cdot \Omega(k)=\Omega(k^{8/5})$ queries to discover a secret vertex in this way.

On the other hand, note that there are $k^{2/5}$ special terminals (in both $w$ and $w_g$). So if we sample a random terminal $u$ from $T$, then the probability that $u$ is a special terminal is $O(k^{-3/5})$ as $|T|=\Omega(k)$, and so it takes $\Omega(k^{3/5})$ random samples to get a special terminal.
In order to certify that $u$ is indeed a special terminal, we need to find a crucial edge incident to it. Since each special terminal is only incident to $k^{3/5}$ crucial edges (in both $w$ and $w_g$), this takes $\Omega(n/k^{3/5})=\Omega(k^{2/5})$ queries.
Altogether, it takes $\Omega(k^{3/5})\cdot \Omega(k^{2/5})=\Omega(k)$ queries to discover a special terminal.
If the algorithm performs $o(k^{6/5})$ queries, it is only able to discover $o(k^{1/5})$ special vertices. As the identities of the terminals are randomized, from the Birthday Paradox, with high probability all discovered special vertices come from different groups in $S_1,\ldots,S_d$, so even if the algorithm has queried all edges incident to these discovered terminals, with high probability it will not find any vertex that is incident to more than one discovered special terminal, and so with high probability it will not find any secret vertex.

$\ $

In the remainder of the section, we formalize the ideas described above, in a way similar to \Cref{sec: proof of beat-2-lower-computing} and \Cref{sec: >2-main}. 
We begin with some definitions.
For convenience, we will think of the algorithm as performing queries on $w$, but it does not know the identities of the vertices (that is, it does know the set $T$ and the partition $(V_0,V_1,\ldots,V_{d'})$ of $V\setminus T$, but it does not know which special terminal in set $S_{j}$ is $s_{j,i}$, for any $j,i$).
Over the course of the algorithm, we say a terminal $u$ is \emph{settled} at some step, iff the set of queries performed thus far uniquely identify $u$ to be a \SP terminal or a \NM terminal, i.e., we have discovered a crucial edge incident to it; otherwise we say it is \emph{unsettled}. 
For convenience of the analysis, whenever a terminal $u$ is settled, if $u$ is a regular terminal, then we immediately reveal to the algorithm which $T_i$ group it belongs to; if it is a \SP terminal, we will immediately reveal to the algorithm all crucial (weight-$1$) edges incident to it. 

For each \SP terminal $u\in S$, we denote by $V(u)$ be set of Steiner nodes that are connected to $u$ by some matching in $\set{M_1,\ldots,M_{d'}}$. Clearly, in $w$, vertex $u$ is connected to all vertices in $V(u)$ via crucial (weight-$1$) edges, while in $w_g$, this is not always the case.
We say that an edge is \emph{discovered} if we can uniquely identify the weight of the edge. For any $0 < \alpha < 1$, we say an unsettled terminal $u$ is \emph{$\alpha$-well-discovered} if at least one of the following four conditions is satisfied: 
\begin{properties}{P}
\item There are at least $\alpha k^{2/5}$ groups $T_i$ such that we have discovered at least one edge in $E(u,T_i)$. \label{query_prop_1}
\item There are at least $\alpha k^{2/5}$ \SP terminals $u'$ such that we discovered an edge in $E(u,V(u'))$. \label{query_prop_2}
\item $u$ is a \NM terminal in the $T_i$, and there are at least $\alpha k$ terminals $u'$ such that some edge in $E(u',T_i)$ has been discovered. \label{query_prop_3}
\item $u$ is a \SP terminal, and there are at least $\alpha k$ terminals $u'$ such that some edge in $E(u',V(u))$ has been discovered. \label{query_prop_4}
\end{properties}

The following lemma is the main technical tool for the proof of the $\Omega(k^{6/5})$ lower bound. Intuitively, it shows that, if the algorithm perform $o(k^{6/5})$ queries, then not only it settles or $(\eps/200)$-well-discovers very few terminals, but also it has very limited knowledge upon the edges that it has not queried.

\begin{lemma} \label{lem:settle-sp}
Let $\alg$ be any deterministic algorithm that performs at most $\eps^2 k^{6/5}/10^9$ queries. Then:
\begin{itemize}
\item If the input to $\alg$ is a random instance from $\DN$, then with probability at least $9/10$, the number of terminals that are either settled or $(\eps/200)$-well-discovered is at most $\eps k^{4/5}/10^4$, and at most $\eps k^{1/5}/10^4$ of them are \SP terminals; and conditioned on the query sequence and its answers over the course of the algorithm, for every terminal that is neither settled nor $(\eps/200)$-well-discovered, the probability that it is a special terminal is at most $2/k^{3/5}$.
\item If the input to $\alg$ is a random instance from $\DY$, then with probability at least $9/10$,
\begin{itemize}
\item the number of terminals that are either settled or $(\eps/200)$-well-discovered is at most $\eps k^{4/5}/10^4$;
\item at most $\eps k^{1/5}/10^4$ of them are \SP terminals; and
\item all these special terminals belong to different $S_i$ groups.
\end{itemize}
Moreover, conditioned on the query sequence and its answers over the course of the algorithm, for every unqueried edge between a Steiner vertex and a not-$(\eps/100)$-well-discovered terminal, the probability that the edge is a crucial (weight-$1$) edge and the Steiner vertex is a secret vertex is at most $2/k^{8/5}$.
\end{itemize} 
\end{lemma}


We provide the proof of \Cref{lem:settle-sp} in \Cref{sec: proof of structural_1} and \Cref{sec: proof of structural_2}, after we complete the proof of the $\Omega(k^{6/5})$ lower bound using it.

\subsubsection{Completing the Proof of the $\Omega(k^{6/5})$ Lower Bound}

Recall that $\dset=(\DY+\DN)/2$.
In this subsection we use \Cref{lem:settle-sp} to prove the following lemma. 

\begin{lemma}
\label{lem: k^6/5}
Any randomized algorithm that, given a random instance sampled from $\dset$, reports correctly with probability at least $2/3$ that the instance comes from $\DY$ or $\DN$, performs at least $\eps^2 k^{6/5}/10^9$ queries in the worst case.
\end{lemma}



From Yao's minimax principle \cite{yao1977probabilistic}, it suffices to consider only deterministic algorithms that report correctly on at least $2/3$-fraction (in $\dset$) of the instances.

We define a \emph{transcript} to be the union of a query sequence and all its answers. We can define $\alpha$-well-discovered vertices and settled terminals with respect to a transcript similarly.

Recall that each instance in $\isetn$ is determined by a mapping $f\in \fset$ and each instance in $\isety$ is determined by a mapping $f\in \fset$ and a mapping $g\in \gset$, so $|\isety|=|\isetn|\cdot |\gset|$.
Let $\sigma$ be the transcript produced by the algorithm when given a random instance from $\dset$. Since the input is randomized, $\sigma$ is a random variable.
We say that $\sigma$ is \emph{consistent} with an instance $I$ iff all answers to the queries in $\sigma$ are matched with the corresponding weight-values in $I$.
In order to prove \Cref{lem: k^6/5}, it suffices to show that, if $\sigma$ always contains at most $\eps^2 k^{6/5}/10^9$ queries, then with high probability,
the ratio between the number of instances in $\isety$ that are consistent with $\sigma$ (that we call \emph{consistent instances in $\isety$}) and the number of instances in $\isetn$ that are consistent with $\sigma$ (that we call \emph{consistent instances in $\isetn$}) is still roughly $|\gset|$. 
We prove this by showing that (i) for each consistent instance $I_f\in \isetn$, almost all mappings $g\in \gset$ give a consistent instance $I_{(f,g)}$; and (ii) for almost all consistent instances $I_{(f,g)}\in \isety$, the instance $I_f\in \isetn$ is also consistent.

We let $T^*$ be the subset of terminals that are either settled or $(\eps/200)$-well-discovered by $\sigma$ (so $T^*$ is a random variable as well). We will focus on queries incident to vertices in $T^*$ and $S\setminus T^*$.

From \Cref{lem:settle-sp}, there are at most $1/10$-fraction of the consistent instances in $\isetn$ such that the desired properties (the number of settled or $(\eps/100)$-well-discovered terminals is low, etc) do not hold; and there are at most $1/10$-fraction of the consistent instances in $\isety$ such that the desired properties do not hold. We call these instances \emph{bad} instances, and call all other consistent instances \emph{good} instances. 

Consider now any mapping $g\in \gset$, we define the following instance $I'_{(f,g)} \in \isety$ defined by $g$ and the $T^*$ which is slightly different from $I_{(f,g)}$. For each $1 \le j \le d$, we consider the matching $M_{g(j)}$ between $S_j$ and $V_{g(j)}$. If $S_j$ contains more than 1 terminals in $T^*$, then $I_{(f,g)}$ is not well defined. If $S_j$ contains one such terminal and assume it is $s^*_j$, then we first exchange the matching node of $s^*_j$ and $s_{j,1}$ in $M_{g(j)}$. If $S_j$ contains no such terminal, then we do not change the matching. After modifying the matching, for each $u \in S_j$, we make $w_g(u,s^*_{g(j)}) =1$ and the weight in $w_g$ between $u$ and its matched Steiner vertex in $M_{g(j)}$ is $2$. Note that after the change, we guarantee that we do not change the crucial edges incident on any terminal in $T^*$.





\paragraph{From consistent instances in $\isetn$ to consistent instances in $\isety$.}
From \Cref{lem:settle-sp}, for every terminal $u \in T\setminus T^*$, the probability (conditioned on $\sigma$ and its answers) that $u$ is \SP terminal is at most $2/k^{3/5}$. Then from Markov's Bound, with probability at least $0.98$ (i.e., on at least $0.98$-fraction of the good instances in $\isetn$), the number of queries performed on edges in $E(S\setminus T^*,V\setminus T)$ is at most $\eps k^{3/5}/100$. We now show that, for such a good instance $I_f\in \isetn$, almost all mappings $g$ give consistent instance $I'_{(f,g)}\in \isety$.

Now for a random mapping $g \in \gset$. If we view the algorithm as performing queries on $w_g$ (without knowing the identities of the vertices), then from the perspective of the algorithm, the partitioning of the \SP terminals into groups $S_1, \dots, S_{k^{2/5}}$ is random. Since $T^*$ contains at most $\eps k^{1/5}/10^4$ \SP terminals, with probability at least $1/10^4$, $|S_j \cap T^*|\le 1$ for all $1\le j\le d'$ hold. Moreover, if $z$ queries are incident to terminals in $S_j \setminus T^*$, then with probability at least $z/k^{3/5}$, no query has been performed between $S_j \setminus T^*$ and $V_{g(j)}$. Thus with probability at least $1-\eps/100$, for each $1\le j\le d'$, no query has been performed on $E(S_j \setminus T^*,V_{g(j)})$. If this happens, then it is easy to verify that $I'_{(f,g)}$ is a well defined and consistent instance in $\isety$, as the difference between $w$ and $w_g$, in particular a subset of edges in $\bigcup_{1\le j\le d'}E(S_j \setminus T^*,V_{g(j)})$, has not been queried at all.


Altogether, there are at least $0.98$-fraction of the good instances in $\isetn$, such that for each such $I_f$, at least $(1-\eps/100-1/10^4)\le (1-10^{-2})$-fraction of mappings $g\in \gset$ can give consistent instances $I'_{(f,g)}$ in $\isety$.

\paragraph{From consistent instances in $\isety$ to consistent instances in $\isetn$.}
From \Cref{lem:settle-sp}, for every unqueried edge between a Steiner vertex and a not-$(\epsilon/100)$-well-discovered terminal, the probability that the edge is a crucial (weight-$1$) edge and the Steiner vertex is a secret vertex is at most $2/k^{8/5}$.
Therefore, for every $u \in S\setminus T^*$ and every $1\le i\le d'$, the probability that $u$ is a \SP terminal and belongs to a group $S_j$ with $g(j)=i$ is at most $(2/k^{8/5})\cdot (k^{2/5})=2/k^{6/5}$. 
For each $1\le i\le d'$, if we denote by $z_i$ the number of terminals $u\in S\setminus T^*$ such that some edge from $E(u,V_i)$ has been queried, then 
$$\Pr\Big[\exists j, \text{ s.t. }g(j)=i,\text{ and some edge in $E(S_j \setminus T^*,V_i)$ has been queried}\Big]\le \frac{2z_i}{\eps k^{6/5}}.$$
Since the algorithm performs at most $\eps^2 k^{6/5}/10^9$ queries, from Markov's Bound, with probability at least $1-10^{-6}$ (i.e., on at least $(1-10^{-6})$-fraction of the good instances in $\isety$), for all pairs $(i,j)$ with $g(j)=i$, no queries has been perform on $E(S_j \setminus T^*,V_i)$. 
For each consistent instance $I'_{(f,g)}\in \isety$ with the above property, it is easy to verify that the corresponding instance $I_{f}$ in $\isetn$ is a consistent instance in $\isetn$, as the difference between $w$ and $w_g$, in particular a subset of edges in $\bigcup_{1\le j\le d'}E(S_j \setminus T^*,V_{g(j)})$, has not been queried at all.

Altogether, there are at least $(1-10^{-6})$ fraction of the good instances in $\isety$, such that for each $I'_{(f,g)}$ of them, the corresponding instance $I_f$ is a consistent instance in $\isetn$.

Form the above discussion, 
the number of consistent instances in $\isety$ is at least $0.9\cdot 0.98 \cdot (1-10^{-2}) |\gset| > 0.85 |\gset|$ times the number of consistent instances in $\isetn$, and it is at most $|\gset|/(0.9\cdot(1-10^{-6}))\le 1.12|\gset|$ times the number of consistent instances in $\isetn$. Therefore, the algorithm reports correctly with probability at most $\max\set{1/(1+0.85),1.12/(1+1.12)}\le 2/3$.

\subsubsection{Proof of \Cref{lem:settle-sp} for $\DN$}
\label{sec: proof of structural_1}

In this subsection we prove of the first half of \Cref{lem:settle-sp}. 
We start with the following claim.

\begin{claim} \label{prop:DN-random}
Let $\sigma$ be any transcript, and let $\DN(\sigma)$ be the uniform distribution on all instances in $\isetn$ that are consistent with $\sigma$. Then
\begin{itemize}
\item for every unqueried edge between a pair of terminals such that at least one is unsettled and not $(\eps/100)$-well-discovered, the probability in $\DN(\sigma)$ that the edge is a crucial (weight-$0$) edge is at most $2/k^{2/5}$; and
\item for every unqueried edge between a Steiner vertex and a unsettled and not-$(\eps/100)$-well-discovered terminal, the probability in $\DN(\sigma)$ that the edge is a crucial (weight-$1$) edge is at most $2/k$.
\end{itemize}
\end{claim}

\begin{proof}
We first prove the first property. 
Consider a pair $u,u'$ of terminals where terminal $u$ is unsettled and not $(\eps/100)$-well-discovered. Let $I, I'$ be instances in $\isetn$ that are consistent with $\sigma$, such that in $I$, $u$ and $u'$ belong to the same $T_j$ group. We say that $I'$ is \emph{host} by $I$, iff $I'$ can be obtained from $I$ by exchanging the role of $u$ with another unsettled \NM terminal $u''$ that is not in the same group as $u$ in $I$. It is clear that in any such instance $I'$, the answer of the same query will $2$ (that is, the edge is not a crucial edge). Suppose $u$ and $u'$ are in group $T_j$. The number of such instances $I'$ equals the number of terminal $u''$ such that (i) $u''\notin T_j$, and no edge in $E(u'',T_j)$ has been queried; and (ii) no edge between $u$ and the group $u''$ belongs to has been queried. Since $u$ is not $(\eps/100)$-well-discovered, Property~\ref{query_prop_1}, the number of $u''$ that violate (ii) is at most $\eps k/100$, and from Property~\ref{query_prop_3}, the number of $u''$ that violate (i) is at most $\eps k/100$. 
Therefore, $I$ hosts at least $49k/50$ instances $I'$. 
On the other hand, for each instance $I'$, the number of instances that hosts it is at most $k^{3/5}$, since $|T_j|\le k^{3/5}$. Altogether, over all instances in $\isetn$ that are consistent with $\sigma$, at most $50/(49k^{2/5})\le (2/k^{2/5})$-fraction of them have $(u,u')$ as a crucial edge.

We now prove the second property. Consider an unsettled and not-$(\eps/100)$-well-discovered terminal $u$ and a Steiner vertex $v$. 
Let $I, I'$ be instances in $\isetn$ that are consistent with $\sigma$, such that in $I$, $u$ is a \SP terminal and $v \in V(u)$.
We say that $I'$ is host by $I$ if $I'$ can be obtained from $I$ by exchanging the role of $u$ with a unsettled \NM terminal $u'$ in $I$. The number of such instances is the number of unsettled \NM terminals $u'$ in $I$ such that (i) no edge in $E(u',V(u))$ has been queried; and (ii) no edge between $u$ and the group that contains $u'$ has been discovered. Since $u$ is not $(\eps/100)$-well-discovered, from Property~\ref{query_prop_3}, the number of \NM terminals that violate (i) is at most $\eps k/100$, and Property~\ref{query_prop_1}, the number of \NM terminals that violate (ii) is at most $\eps k^{3/5} \cdot k^{2/5}/100 = \eps k/100$. Therefore, $I$ hosts at least $49k/50$ instances. On the other hand, any instance $I'$ is host by at most one instance since at most one crucial edge is incident to$v$. 
Altogether, over all instances in $\isetn$ that are consistent with $\sigma$, at most $50/49k\le (2/k)$-fraction of them have $(u,v)$ as a crucial edge.
\end{proof}

In the remainder of this subsection, we will refer to $(\eps/200) $-well-discovered vertices as \emph{well-discovered vertices}, for convenience.
We call queries between two terminals \emph{regular} queries, and queries between a terminal and a Steiner node \emph{special} queries.
We say a query is \emph{good} iff it discovers a previously-unknown crucial edge. In other words, either (i) the query is a \NM query, such that at least one endpoint is not well-discovered, and the answer is $0$, or (ii) the query is a \SP query such that the terminal is not $(\eps/100)$-well-discovered, and the answer is $1$. 
From \Cref{prop:DN-random}, the probability that a \NM query is good is at most $2/k^{2/5}$, and the probability that a \SP query is good is at most $2/k$. Therefore, if the algorithm performs at most $\eps^2 k^{6/5}/10^9$ queries, then from Markov's Bound, with probability $1/50$, the number of good \NM queries is at most $\eps^2 k^{2/5}/10^6$, and the number of good \SP query is at most $\eps^2 k^{3/5}/10^6$.

For every settled terminal, if it is not settled by a good query, it must become well-discovered before it become settled. Thus if we can upper bound the number of terminals that are well-discovered without but not settled at some step, then we can upper bound the number of settled or well-discovered terminals as well. 
For ease of analysis, we always assume that there are at most $\eps k^{4/5}/10^4$ terminals and $\eps k^{1/5}/10^4$ \SP terminals that are settled or well-discovered, and we think of the algorithm as being immediately terminated once this condition no longer holds.

There are four possiblities for a terminal to become well-discovered, and we analyze them seperately. 

\textbf{Possibility 1: through \ref{query_prop_3}.} Let $u$ be such a vertex, so $u$ is a \NM terminal and the group that $T_i$ contains it has $\eps k/200$ incident edges discovered. For each such edge, it is either discovered by a query or because the its other endpoint is settled or well-discovered. Since there are at most $\eps k^{4/5}/10^4$ settled or well-discovered terminals, there are at least $0.004 \eps k$ edges incident on $T_i$ that are discovered by queries. Thus, if the total number of queries is at most $\eps^2 k^{6/5}/10^9$, then there are at most $\eps k^{1/5}/10^6$ groups such that the terminals in this group is well-discovered through \ref{query_prop_3}. Therefore, the number of terminals that become well-discovered through \ref{query_prop_3} is at most $\eps k^{4/5}/10^6$. Note that all these terminals are \NM terminals. 

\textbf{Possibility 2: through \ref{query_prop_4}.}
Let $u$ be such a vertex, so $u$ is a \SP terminal and the set $V(u)$ has $\eps k/200$ incident edges discovered. By the same argument, for at most $\eps k^{4/5}/10^4$ of these edges, the other endpoint is a settled or well-discovered terminal, and all the others are discovered by queries. Therefore, the number of such terminals is at most $\eps k^{1/5}/10^6$, and all of them are \SP terminals. 

\textbf{Possibility 3: through \ref{query_prop_1} but not \ref{query_prop_3}/\ref{query_prop_4}.} Let $u$ be such a vertex, so there are at least $\eps k^{2/5}/200$ groups $T_i$ such that some edge in $E(u,T_i)$ has been discovered. Note that each such edge is discovered either by a query, or because we have discovered all terminals in $T_i$. By previous analysis, there are at most $\eps k^{1/5}$ of them. Therefore, at least $0.004 \eps k^{2/5}$ edges incident on $u$ are queried. 

\textbf{Possibility 4: through \ref{query_prop_2} but not \ref{query_prop_3}/\ref{query_prop_4}.} Let $u$ be such a vertex, so there are at least $\eps k^{2/5}/200$ \SP terminals $u'$ such that some edge in $E(u,V(u'))$ has been discovered. Note that such an edge is discovered either by a query, or because $u'$ has already been settled. Therefore, at least $0.004\eps k^{2/5}$ edges incident to $u$ are queried.

From the analysis in Possibilities $3$ and $4$, for any terminal that becomes well-discovered through \ref{query_prop_1} or \ref{query_prop_2} but not \ref{query_prop_3} or \ref{query_prop_4}, at least $0.004\eps k^{2/5}$ of its incident edges have been queried. Therefore, there are at most $\eps k^{4/5}/(4 \cdot 10^6)$ such vertices. 
However, such terminals could be either \NM or \SP, and we still need to upper bound the number of such \SP terminals. Let $u$ be any such terminal, and consider the moment when exactly $0.004 \eps k^{2/5}$ of its incident edges are queried. From \Cref{prop:DN-random}, for any Steiner node $v$, the probability that $(u,v)$ is a crucial (weight-$1$) edge is at most $2/k$. It follows that the probability that $u$ is a \SP terminal is at most $(2/k) \cdot (k / k^{3/5}) = 2/k^{3/5}$. 
By Markov's Bound, the number of \SP terminals that are well-discovered is at most $\eps k^{1/5}/10^5$ with probability at least $1/40$.

Altogether, with probability $1-1/40-1/50\ge 9/10$, the number of settled or well-discovered terminals is at most $(\eps^2 k^{4/5}+\eps k^{4/5}+\eps k^{1/5}+\eps k^{4/5})/10^6< \eps k^{4/5}/10^4$, and the number of settled or well-discovered \SP terminals is at most $(\eps^2 k^{1/5}+\eps k^{1/5})/10^6+\eps k^{1/5}/10^5 < \eps k^{1/5}/10^4$.

\subsubsection{Proof of \Cref{lem:settle-sp} for $\DY$}
\label{sec: proof of structural_2}


In this subsection we provide the proof of the second half of  \Cref{lem:settle-sp}

Recall that Steiner vertices that are incident to more than one crucial edges are called \emph{secret vertices}. 
We assume for now that, over the course of the algorithm,
\begin{itemize}
\item for every Steiver vertex, we have discovered at most one crucial edge incident to it;
\item no \IM vertex has more than $\eps k/10^4$ of its incident edges discovered;
\item the number of terminals that are settled or well-discovered is at most $\eps k^{4/5}/10^4$; and
\item the number of \SP terminals that are settled or well-discovered is at most $\eps k^{1/5}/10^4$.
\end{itemize}
We will show at the end of this subsection that, if any of the above condition is not satisfied, then the algorithm has to perform at least $\eps^2 k^{6/5}/10^9$ as well.

We start with the following claim, whose is very similar to \Cref{prop:DN-random}, and is omitted here.

\begin{claim} \label{prop:DY-random}
Let $\sigma$ be any transcript, and let $\DY(\sigma)$ be the uniform distribution on all instances in $\isety$ that are consistent with $\sigma$. Then
\begin{itemize}
\item for every unqueried edge between a pair of terminals such that at least one is unsettled and not $(\eps/100)$-well-discovered, the probability in $\DY(\sigma)$ that the edge is a crucial (weight-$0$) edge is at most $2/k^{2/5}$; and
\item for every unqueried edge between a Steiner vertex and a unsettled and not-$(\eps/100)$-well-discovered terminal, the probability in $\DY(\sigma)$ that the edge is a crucial (weight-$1$) edge is at most $2/k$.
\end{itemize}
\end{claim}

Using \Cref{prop:DY-random} and the same arguments in the proof of \Cref{lem:settle-sp} for $\DN$, we can prove that the number of terminals that are settled or well-discovered is at most $\eps k^{4/5}/10^4$, and at most $\eps k^{1/5}/10^4$ of them are \SP terminals.
In order to prove that all settled or well-discovered \SP terminals belong to different $S_j$ groups, we use the following two claims.

\begin{claim} \label{prop:DY1}
Let $u^*$ be a special terminal that is either settled or $(\eps/200)$-well discovered.
Then for every other not-$(\eps/100)$-well-discovered terminal $u$ and any other Steiner vertex $v$, the probability that $u$ is a \SP terminal in the same $S_j$ group as $u^*$ and $(u,v)$ is a crucial edge is at most $2/k^{7/5}$.
\end{claim} 

\begin{proof}
Let $I,I'$ be instances in $\isetn$, let $u^*$ be a special terminal that is either settled or $(\eps/200)$-well discovered in both $I$ and $I'$, let $u$ be a special vertex in the same $S_j$ group as $u^*$ in $I$, and let $v$ be a Steiner vertex such that $(u,v)$ is a crucial edge in $I$.
We say $I'$ is \emph{host} by $I$, iff there are two terminals $u', u''$, such that $u'$ is \SP and $u''$ is  \NM, and $I'$ can be obtained from $I$ by giving the role of $u'$ to $u$, giving the role of $u$ to $u''$ and giving the role of $u''$ to $u'$. 
The vertex $u'$ can be any not-$(\eps/100)$-well-discovered terminal in $S_j$ with no edge to $V(u)$ discovered. Since $u$ is not $(\eps/100)$-well-discovered, and we have assumed that each \IM Steiner vertex has at most $\eps k/10^4$ of its incident edges discovered, there are at least $0.99k$ choices for $u'$. 
On the other hand, $u''$ can be any terminal that has no edge discovered to $V(u')$ and the \IM Steiner node of group $S_j$ in $I$. By the same arguement, there are at least $0.99k$ such terminal. Since we have assumed that at most $\eps k^{1/5}/10^4$ \SP terminals are $(\eps/100)$-well-discovered, an instance $I$ hosts at least $0.99^2k^2 \cdot (k^{2/5}/\eps) > 0.98 k^{12/5}/\eps$ instances $I'$. On the other hand, for any instance $I'$, the terminal $u'$ must be the terminal such that $v \in V_{u'}$, and the terminal $u''$ must in the same group as $u^*$. So there are at most $k/\eps$ instance $I$ that hosts $I'$. Altogether, the probability that all events happend is at most $(k/\eps)/(0.98 k^{12/5}/\eps)\le 2/k^{7/5}$.
\end{proof}

\begin{claim} \label{prop:DY2}
Let $u$ be a not-$(\eps/100)$-well-discovered terminal and let $v$ be a Steiner vertex. Then the probability that $(u,v)$ is a crucial (weight-$1$) edge and $v$ is an \IM vertex is at most $2/k^{8/5}$.
\end{claim}

\begin{proof}
Recall that we have assumed that $v$ has at most one crucial edge connecting to a settled or a well-discovered terminal discovered. By definition, any $(\eps/100)$-well discovered terminal is also a well discovered terminal. So there $v$ $v$ has at most one crucial edge connecting to a settled or a $(\eps/100)$-well-discovered terminal discovered. We distinguish between the following two cases. 


\textbf{Case 1:} There does not exist a settled or $(\eps/100)$-well discovered terminal $u'$, such that $(u',v)$ is a crucial edge and has been discovered. For any consistent instance $I$ such that $u$ and $v$ has weight $1$ and $v$ is an \IM Steiner node. Suppose $u \in S_i$ and let $S_i = \{s_{i,1},\dots,s_{i,1/\eps}\}$, and suppose $v \in V(s_{i,1})$. By assumption, no terminal in $S_i$ is $(\eps/100)$-well discovered. We say an instance $I'$ is hosted by $I$ if all \SP terminals not in $S_i$ and their weights to the Steiner nodes are the same as $I$, and the set $V(s_{i,t})$ in $I'$ is identical to the set $V(s_{i,t})$ in $I$ for all $1\le t\le 1/\eps$. 

We first count the number of such consistent instances $I'$. For each $1 \le t \le 1/\eps$, we define $T^*_t$ as the set of \NM terminals that with no edges to $V(s_{i,t})$ discovered. Since no terminal in $S_i$ is $(\eps/100)$-well discovered, every set $T^*_t$ has size at least $(1-\eps/100)k$. For each $1 \le t \le 1/\eps$, consider any group of terminals $t_1 \dots, t_{1/\eps}$ such that $t_{j'} \in T'_{j'} \cap T'_{j}$. Any $t_{j'}$ can have weight one to all Steiner nodes in $V_{s_{j'}}$ and $V_{s_j}$, which means we can make any terminal in $V_{s_{j}}$ an \IM terminal. Thus such group can construct at least $k^{3/5}/\eps$ consistent instances $I'$. On the other hand, since any $\card{T_{j'}} \ge (1-\eps/100)k$, the number of such group is at least $((1-49\eps/50)k)^{1/\eps} > (49/50)\cdot k^{1/\eps}$. Thus, $I$ hosts at least $(49/50)k^{1/\eps} \cdot (k^{3/5}/\eps)$ instances $I'$.

Now we count how many instance can host an instance $I'$. For any instance $I'$, an instance $I$ that hosts it can only change the terminals in $S_i$ that has weight $1$ to $v$, and so one of the terminal in $S_i$ should be $u$. Moreover, every terminal in $S_i$ should has weight $1$ to $v$. The total number of such instance is at most $k^{1/\eps} \cdot 1/\eps$ since $u$ could replace any terminal in $S_i$. Thus, the probability that $v$ and $u$ has weight $1$ and $v$ is an \IM Steiner node is at most $(k^{1/\eps-1} \cdot (1/\eps))/((49/50) k^{1/\eps} k^{3/5} \cdot (1/\eps)) < 2/k^{8/5}$.
    
\textbf{Case 2:} There does not exist a settled or $(\eps/100)$-well discovered terminal $s_{j^*}$, such that $(s_{j^*},v)$ is a crucial edge and has been discovered. 
Note that $s_{j^*} \neq u$. For any instance $I$ such that $u$ and $v$ has weight $1$. We say an instance $I'$ is hosted by $I$ as the same definition as the first case, except that now $u$ must still in $S_i$, and moreover, we exchange one Steiner node in $V_{j^*}$ with $V_j$ for some $j$. We first count the number of $I'$ hosted by $I$. We define $T'_j$ the same as the first case. Suppose $v \in V_{\ell^*}$, for any $1 \le \ell \le k^{3/5}$ and $\ell \neq \ell^*$, we define $v^*_{\ell}$ as the only one vertex in $V_j^* \cap V_{\ell}$, and $T^*_{\ell}$ as the set of terminals that does not have weight one to $v^*_{\ell}$. We also define $v^*{\ell^*} = v$ and $T^*_{\ell^*}$ as the set of terminals that does not have weight one to $v^*_{\ell^*}$. Now for any $1 \le \ell \le k^{3/5}$, for any group of terminals $t_1, \ldots, t_{1/\eps}$ such that for any $j' \neq j^*$, $t_{j'} \in T'_{j'} \cap T^*_{\ell}$ and $t_{j^*}=s_{j^*}$, we can make $v^*_{\ell}$ the \IM Steiner node of this group. Moreover, to do so, we can exchange $v^*_{\ell}$ from $V_{s_{j^*}}$ with any $V_{s_{j'}}$ since we will not change the weight between any pair of vertices by doing so. Since the algorithm only perform at most $\eps^2 k^{6/5}/10^9$ queries and settled or $(\eps/100)$-well discovered at most $O(k^{4/5})$ terminals, there are $(1-o(1))k^{3/5}$ number of index $\ell$ such that $\card{T^*_{\ell}} > (1-\eps/100)k$. For such $\ell$, the number of groups $t_1, \dots, t_{1/\eps}$ is at least $0.99k^{1/\eps -1}$. So the total number of instances $\I'$ hosted by $I$ is at least $(1-o(1)k^{3/5} \cdot k^{1/\eps -1} \cdot (1/\eps)$. On the other hand, for any instance $I'$, it is hosted by at most $k^{1/\eps-2} \cdot (1/\eps)$ instance $I$ since we cannot exchange $s_{\ell^*}$ and $u$ must be a \SP terminal in $S_i$. This implies that $v$ and $u$ has weight $1$ and $v$ is a \IM Steiner node is at most $(k^{1/\eps-2} \cdot (1/\eps))/((1-o(1))0.99 k^{1/\eps-1} k^{3/5} \cdot {1/\eps}) < 2/k^{8/5}$.
\end{proof}

Remember that a \SP query is called a \emph{good} query iff it discovers a crucial edge between a not-well-discovered terminal and a Steiner node. If a \SP terminal is ever settled, then it is either due to a good \SP query incident to it, or because it becomes well-discovered at some step. We analyze the probability that a \SP query $(u,v)$ is a good query, and $u$ is in the same group $S_i$ with some terminal $u'$ that is already settled or $(\eps/200)$-well discovered. By \Cref{prop:DY2}, the probability that the query is a good query and $v$ is an \IM Steiner node is at most $2/k^{8/5}$. On the other hand, if the query is a good query but $v$ is not an \IM terminal, it means $v \in V(u)$. By \Cref{prop:DY1}, the probability that it is a good query and $u$ and a fix settled or $(\eps/200)$-well discovered $u^*$ in the same group $S_i$ is at most $2/k^{7/5}$. Since there are at most $\eps k^{1/5}/10^4$ such $u^*$. So the probability that a \SP is a good query and $u$ is in the same group $S_i$ with some terminal $u'$ that is already settled or $(\eps/200)$-well discovered is at most $\eps/10^4 k^{6/5}$. With probablity at least $1-\eps/1000$, there is no such query through out the algorithm by Markov's Bound.

Now we consider the well-discovered terminals. When a terminal is well-discovered but not settled, it is still not $(\eps/100)$-well discovered. By \Cref{prop:DY2}, the probability that it is in the same group $S_i$ with some terminal $u'$ that is already settled or $(\eps/200)$-well discovered is at most $(2/k^{8/5}) \cdot (\eps k^{1/5}/10^4) \cdot k^{3/5} < \eps/100 k^{4/5}$. Since there are at most well-discovered $\eps k^{4/5}/10^4$, no well-discovered terminal belongs to the same group with some terminal $u'$ that is already settled or well-discovered.

Finally, we need to prove the assumption that there is no \IM Steiner node has at least $\eps k/10^4$ edges discovered. Since we only performed at most $\eps^2 k^{6/5}/10^9$ queries, there are at most $\eps k^{1/5}/10^4$ terminals that are queried at least $\eps k/10^5$ times. By \Cref{prop:DY2}, the probability that such terminal is \IM is at most $2/k^{8/5} \cdot k = 2/k^{3/5}$. Thus with proability at least $1-k^{-2/5}$, all these terminals are not \IM. This finishes the proof of \Cref{lem:settle-sp}.

\newpage

\appendix

\section{An $O(nk)$-Query $(5/3)$-Approximation Algorithm}
\label{apd: 5/3 upper}

In this section, we explain how the previous work \cite{zelikovsky199311} and \cite{du1995component} lead to an $O(nk)$-query algorithm for computing a $(5/3)$-approximate Steiner Tree. In fact, their results imply the following theorem.

\begin{theorem}
\label{thm: 5/3-main}
For any instance $(V,T,w)$ of \ST problem, there exists a Steiner tree $\tau^*$ with weight at most $(5/3)\cdot \stcost(V,T,w)$, such that every edge in $\tau^*$ is incident on some vertex in $T$.
\end{theorem}

We refer to such Steiner trees as \emph{good trees}.
From the above theorem, it is not hard to observe that querying all terminal related distances is sufficient to find a $(5/3)$-approximate Steiner Tree, and the query complexity is $O(nk)$.

We now explain how the results in previous work \cite{zelikovsky199311} and \cite{du1995component} imply \Cref{thm: 5/3-main}.

We start by introducing some definitions.
Let $\tau$ be a tree, let $v$ be a vertex of $\tau$, and let $v_1,\ldots,v_d$ be the neighbors of $v$. For each $1\le  i\le d$, we delete edges $(v,v_1),\ldots,(v,v_{i-1}),(v,v_{i+1}),\ldots,(v,v_{d})$, and define $\tau_i$ to be the connected component of the remaining graph that contains $v$, so $\tau_i$ is a subtree of $\tau$ that contains $v$. We say that subtrees $\tau_1,\ldots,\tau_d$ are obtained by \emph{splitting $\tau$ at $v$}.

Consider an instance $(V,T,w)$ and let $\tau$ be a Steiner tree. Let $c>1$ be an integer. We say that $\tau$ is a \emph{$c$-Steiner tree}, iff when we split $\tau$ at all terminals, then each resulting subtree contains at most $c$ terminals. It is easy to verify that any $2$-Steiner Tree is a terminal spanning tree. We now show that every $3$-Steiner tree can be converted into a good tree with at most the same cost. Let $\tau$ be a $3$-Steiner tree. Assume without loss of generality that every Steiner vertex has degree at least $3$ (as otherwise we can suppress such a vertex and get another Steiner tree with at most the same cost). We now claim that $\tau$ does not contain any Steiner-Steiner edge.
Assume not, then such a pair of Steiner vertices must both belong to some subtree obtained by splitting $\tau$ at all terminals, and such a subtree contains at least $(3+3-1-1)=4$ terminals, a contradiction.
It was shown in \cite{zelikovsky199311} and \cite{du1995component} that, for any instance $(V,T,w)$, there exists a $3$-Steiner tree with cost at most $(5/3)\cdot \stcost(V,T,w)$, and the ratio $5/3$ here cannot be improved. \Cref{thm: 5/3-main} now follows.

\section{Proof of \Cref{clm: YN metrics}}
\label{apd: Proof YN metrics}

	Let $v_1,v_2,v_3$ be three vertices in $V$. Assume $v_1\in V_{x_1}$, $v_2\in V_{x_2}$, and $v_3\in V_{x_3}$, where $x_1,x_2,x_3$ are nodes in tree $\rho$. We denote by $\ell_1,\ell_2,\ell_3$ the levels of $x_1,x_2,x_3$, respectively, and assume w.l.o.g. that $\ell_1\ge \ell_2$. Let $x'_1$ be a leaf of $\rho$ that lies in the subtree of $\rho$ rooted at $x_1$, and we define leaves $x'_2,x'_3$ similarly. 
	
	We first show that $\wn$ is a metric on $V$ by showing that $\wn(v_1,v_2)\le \wn(v_1,v_3)+\wn(v_2,v_3)$.
	By definition, 
	$\wn(v_1,v_2)=\dist_{\rho}(x'_1,x_1)+ \dist_{\rho}(x'_1,x_2)$.
	Let $\hat x$ be the lowest common ancestor of nodes $x_1$ and $x_2$ in $\rho$. Assume that $y$ is the unique vertex on the $x'_1-x'_2$ path that is closest (under $\dist_{\rho}$) to $x_3$. We distinguish between the following three cases, depending on the location of $y$.
	
	\textbf{Case 1:} $y$ lies between (excluding) $\hat x$ and (including) $x'_1$. On the one hand, $\wn(v_2,v_3)\ge \dist_{\rho}(x_2,x_3)$; on the other hand, $\wn(v_1,v_3)\ge \dist_{\rho}(x_1,x_3)+2\cdot \min\set{\dist_{\rho}(x'_1,x_1),\dist_{\rho}(x'_3,x_3)}$.
	
	If $y$ lies between (excluding) $\hat x$ and (including) $x_1$, then
	\[
	\begin{split}
	& \wn(v_1,v_3)  +\wn(v_2,v_3)  \ge \dist_{\rho}(x_2,x_3)+\dist_{\rho}(x_1,x_3)+2\cdot\min\set{\dist_{\rho}(x'_1,x_1),\dist_{\rho}(x'_3,x_3)}\\
	& = \dist_{\rho}(x_1,x_2)+2\cdot \dist_{\rho}(y,x_3)+2\cdot\min\set{\dist_{\rho}(x'_1,x_1),\dist_{\rho}(x'_3,x_3)}\\
	& = \dist_{\rho}(x_1,x_2)+2\cdot \dist_{\rho}(x'_1,x_1)+2\cdot \dist_{\rho}(y,x_3)+2\cdot\min\set{0,\dist_{\rho}(x'_3,x_3)-\dist_{\rho}(x'_1,x_1)}\\
	& = \dist_{\rho}(x_1,x_2)+2\cdot \dist_{\rho}(x'_1,x_1)+2\cdot\min\set{\dist_{\rho}(y,x_3),\dist_{\rho}(y,x_3)+\dist_{\rho}(x'_3,x_3)-\dist_{\rho}(x'_1,x_1)}\\
	& \ge  \dist_{\rho}(x_1,x_2)+2\cdot \dist_{\rho}(x'_1,x_1)+2\cdot\min\set{\dist_{\rho}(y,x_3),0}\\
	& \ge  \dist_{\rho}(x_1,x_2)+2\cdot \dist_{\rho}(x'_1,x_1)\\
	& =\dist_{\rho}(x'_1,x_1)+ \dist_{\rho}(x'_1,x_2) = \wn(v_1,v_2).
	\end{split}\]
	If $y$ lies between (excluding) $x_1$ and (including) $x'_1$, then
	\[
	\begin{split}
	& \wn(v_1,v_3)  +\wn(v_2,v_3)  \ge \dist_{\rho}(x_2,x_3)+\dist_{\rho}(x_1,x_3)+2\cdot\min\set{\dist_{\rho}(x'_1,x_1),\dist_{\rho}(x'_3,x_3)}\\
	& = \dist_{\rho}(x_1,x_2)+2\cdot \dist_{\rho}(y,x_1)+2\cdot \dist_{\rho}(y,x_3)+2\cdot\min\set{\dist_{\rho}(x'_1,x_1),\dist_{\rho}(x'_3,x_3)}\\
	& = \dist_{\rho}(x_1,x_2)+2\cdot \dist_{\rho}(x'_1,x_1)+2\cdot \dist_{\rho}(y,x_3)+2\cdot\min\set{\dist_{\rho}(y,x_1),\dist_{\rho}(x'_3,x_3)-\dist_{\rho}(y,x'_1)}\\
	& \ge \dist_{\rho}(x_1,x_2)+2\cdot \dist_{\rho}(x'_1,x_1)+2\cdot\min\set{\dist_{\rho}(y,x_3),\dist_{\rho}(y,x_3)+\dist_{\rho}(x'_3,x_3)-\dist_{\rho}(y,x'_1)}\\
	& =  \dist_{\rho}(x_1,x_2)+2\cdot \dist_{\rho}(x'_1,x_1)+2\cdot\min\set{\dist_{\rho}(y,x_3),0}\\
	& =  \dist_{\rho}(x_1,x_2)+2\cdot \dist_{\rho}(x'_1,x_1)\\
	& =\dist_{\rho}(x'_1,x_1)+ \dist_{\rho}(x'_1,x_2) = \wn(v_1,v_2).
	\end{split}
	\]
	
	\textbf{Case 2:} $y$ lies between (excluding) $\hat x$ and (including) $x'_2$. The analysis in this case is symmetric to that of Case $1$, with an additional observation that $\dist_{\rho}(x_1,x'_1)\le \dist_{\rho}(x_2,x'_2)$ (as $\ell_1\ge \ell_2$).
	
	\textbf{Case 3:} $y=\hat x$. In this case, from the definition of $\wn$, $\wn(v_2,v_3)\ge \dist_{\rho}(x_2,x_3)$, and $\wn(v_1,v_3)\ge \dist_{\rho}(x_1,x_3)+2\cdot \min\set{\dist_{\rho}(x'_1,x_1),\dist_{\rho}(x'_3,x_3)}$. Therefore,
	\[
	\begin{split}
	& \wn(v_1,v_3)  +\wn(v_2,v_3) \ge \dist_{\rho}(x_1,x_2)+2\cdot \dist_{\rho}(\hat x,x_3)+2\cdot\min\set{\dist_{\rho}(x'_1,x_1),\dist_{\rho}(x'_3,x_3)}\\
	& \ge \dist_{\rho}(x_1,x_2)+2\cdot\dist_{\rho}(x'_1,x_1)+2\cdot\min\set{\dist_{\rho}(\hat x,x_3),\dist_{\rho}(x'_3,x_3)+\dist_{\rho}(\hat x,x_3)-\dist_{\rho}(x'_1,x_1)}\\
	& \ge \dist_{\rho}(x_1,x_2)+2\cdot\dist_{\rho}(x'_1,x_1)+2\cdot\min\set{\dist_{\rho}(\hat x,x_3),0}\\
	& =  \dist_{\rho}(x_1,x_2)+2\cdot \dist_{\rho}(x'_1,x_1)\\
	& =\dist_{\rho}(x'_1,x_1)+ \dist_{\rho}(x'_1,x_2) = \wn(v_1,v_2).
	\end{split}
	\]
	This completes the proof that $\wn$ is a metric on $V$.
	
	We now proceed to show that $\wy$ is a metric on $V$ by showing that $\wy(v_1,v_2)\le \wy(v_1,v_3)+\wy(v_2,v_3)$. We distinguish between the following cases.
	
	\textbf{Case 1:} $v_1, v_2,v_3\in S$. In this case,
	\[\wy(v_1,v_2)=\dist_{\rho}(v_1,v_2)\le \dist_{\rho}(v_1,v_3)+\dist_{\rho}(v_2,v_3)= \wy(v_1,v_3)+\wy(v_2,v_3).\]
	
	\textbf{Case 2:} At most one of $v_1, v_2,v_3$ lies in $S$. In this case,
	\[\wy(v_1,v_2)=\wn(v_1,v_2)\le \wn(v_1,v_3)+\wn(v_2,v_3)= \wy(v_1,v_3)+\wy(v_2,v_3).\]
	
	\textbf{Case 3:} Exactly two of $v_1, v_2,v_3$ lie in $S$. We further consider the following subcases.
	
	\textbf{Case 3.1:} $v_1, v_2 \in S$, and $v_3 \notin S$. In this case,
	\[\wy(v_1,v_2)\le\wn(v_1,v_2)\le \wn(v_1,v_3)+\wn(v_2,v_3)= \wy(v_1,v_3)+\wy(v_2,v_3).\]
	
	\textbf{Case 3.2:} $v_2, v_3 \in S$, and $v_1 \notin S$. The analysis in this case uses almost identical arguments as Case 1 for showing that $\wn$ is a metric (since there we only uses the fact that $\wn(x_2,x_3)\ge \dist_{\rho}(x_2,x_3)$).
	
	\textbf{Case 3.3:} $v_1, v_3 \in S$, and $v_2 \notin S$. The analysis in this case uses almost identical arguments as Case 2 for showing that $\wn$ is a metric (since there we only uses the fact that $\wn(x_1,x_3)\ge \dist_{\rho}(x_1,x_3)$).
	
	This completes the proof that $\wy$ is a metric on $V$.

\section{Proof of \Cref{obs: edges love terminal}}
\label{apd: Proof of edges love terminal}

	Let $\tau'$ be an optimal \St of instance $(S,T,\wn)$, and assume for contradiction that $\tau'$ contains an edge $(u_x,u_{x'})$ where $u_x,u_{x'}\notin T$ (or equivalently $x,x'\notin L(\rho)$). Assume without loss of generality that the level of $x$ in $\rho$ is at least the level of $x'$. Let $\tilde x$ be any leaf in $\rho$ that is a descendant of $x$. Consider now the unique path in $\tau'$ connecting $u_{\tilde x}$ to $u_x$, that we denote by $P$.
	
	Assume first that the vertex $u_{x'}$ does not belong to $P$. Let $\tau$ be the tree obtained from $\tau'$ by replacing the edge $(u_x,u_{x'})$ with edge $(u_{\tilde x},u_{x'})$. It is easy to verify that $\tau$ is a \St. Moreover, from the definition of $\wn$, 
	$$\wn(u_x,u_{x'})=\dist_{\rho}(x,x')+2\cdot\dist_{\rho}(x,\tilde x)>\dist_{\rho}(x,x')+\dist_{\rho}(x,\tilde x)= \wn(u_{\tilde x},u_{x'}),$$ 
	which implies that $w(\tau)<w(\tau')$, a contradiction to the assumption that $\tau'$ is an optimal \St of the instance $(S,T,\wn)$.
	
	Assume now that vertex $u_{x'}$ belongs to $P$. Similarly, let $\tau$ be the tree obtained from $\tau'$ by replacing the edge $(u_x,u_{x'})$ with the edge $(u_{\tilde x},u_{x})$. It is easy to verify that $\tau$ is a \St. Moreover, from the definition of $\wn$, 
	$$\wn(u_x,u_{x'})=\dist_{\rho}(x,x')+2\cdot\dist_{\rho}(x,\tilde x)>\dist_{\rho}(x,\tilde x)= \wn(u_{\tilde x},u_{x}),$$ 
	which implies that $w(\tau)<w(\tau')$, again a contradiction to the assumption that $\tau'$ is an optimal \St of the instance $(S,T,\wn)$. This completes the proof of the observation.

\section{Proof of \Cref{obs: OPT_props}}
\label{apd: Proof of OPT_props}

First of all, it is easy to see that there exists an optimal \St such that every Steiner vertex has degree at least $3$, since we can compress degree-$2$ Steiner vertices without increasing the cost.

Consider now a tree $\tau'$ such that:
\begin{enumerate}
\item $\tau'$ is an optimal \St such that every Steiner vertex has degree at least $3$;
\label{prop_1}
\item on top of \ref{prop_1}, $\tau'$ minimizes the number of Steiner vertices; 
\label{prop_2}
\item on top of \ref{prop_1} and \ref{prop_2}, $\tau'$ maximizes the sum of levels of all its Steiner vertices; and
\label{prop_3}
\item on top of \ref{prop_1}, \ref{prop_2}, and \ref{prop_3}, $\tau'$ minimizes the number of edges incident to Steiner vertices.
\label{prop_4}
\end{enumerate}

For each Steiner vertex $u_x$ in $\tau'$, we denote by $d_0(x)$, $d_1(x)$, $d_2(x)$ the number of neighbors of $u_x$ in sets $T_0(x)$, $T_1(x)$, $T_2(x)$, respectively.
Consider now a Steiner vertex $u_x$ in $\tau'$. Let $u_{x'}$ be the parent of $u_x$, and let $u_{x_1}, u_{x_2}$ be the children of $u_x$, where $u_{x_1}\in S_1(x)$ and $u_{x_2}\in S_2(x)$.
We distinguish between the following cases.

\textbf{Case 1:} One of $d_0(x),d_1(x),d_2(x)$ is $0$.
Assume first that $d_0(x)=0$. Then if $d_1(x)\ge d_2(x)$, we can replace Steiner vertex $u_x$ with $u_{x_1}$ (that is, delete from $\tau'$ the vertex $u_x$ and all its incident edges, and replace them with vertex $u_{x_1}$ and edges in $\set{(u,u_{x_1})\mid (u,u_x)\in E(\tau')}$), without increasing the total cost. In this way, we either reduce the number of Steiner vertices by $1$, or increase the sum of levels of all Steiner vertices, a contradiction to either \ref{prop_2} or \ref{prop_3}. The case where $d_1(x)\le d_2(x)$ is symmetric. Assume now that $d_1(x)=0$ (the case where $d_2(x)=0$ is symmetric). Then if $d_2(x)\ge d_0(x)$, we can replace $u_x$ with $u_{x_2}$ either reducing the number of Steiner vertices or increasing the sum of levels of all Steiner vertices, leading to a contradiction to \ref{prop_2} or \ref{prop_3}; if $d_2(x)< d_0(x)$, we can replace $u_x$ with $u_{x'}$, reducing the total cost, leading to a contradiction to \ref{prop_1}.

\textbf{Case 2:} One of $d_1(x),d_2(x)$ is at least $2$. Assume w.l.o.g. that $d_1(x)\ge 2$. Let $u,u'$ be two vertices of $T_1(x)$ that are adjacent to $u_x$ in $\tau'$. We can replace the edge $(u,u_x)$ with edge $(u,u')$, and it is easy to verify from the definition of $\wn$ that this does increase the total cost, leading to a contradiction to \ref{prop_4}.

\textbf{Case 3:} $d_1(x)=d_2(x)=1$ and $d_0(x)> 2$. 
In this case, we can replace $u_x$ with $u_{x'}$, reducing the total cost, leading to a contradiction to \ref{prop_1}.

Altogether, we get that $d_1(x)=d_2(x)=1$ and $d_0(x)$ is either $1$ or $2$, completing the proof of property (i) in the observation. We now focus on proving property (ii).
Let $\tau'$ be the tree defined above. Let $u_1$ ($u_2$, resp.) be the neighbor of $u_x$ in $T_1(x)$ ($T_2(x)$, resp.).
Assume for contradiction that in there exists some vertex $u\in T_1(x)$ such that $u\notin W_1$. From similar arguments in the proof of \Cref{obs: edges love terminal}, we can replace the edge $(u_1,u_{x})$ with edge $(u_1,u)$, obtaining another optimal \St with less edges incident to Steiner vertices, a contradiction to \ref{prop_4}. Assume for contradiction that in there exists some vertex $u\notin T_1(x)$ such that $u\in W_1$. We root tree $W_1$ at $u_1$, and it is easy to see that there exists some pair $u',u''$ of vertices in $W_1$, such that $u'\in T_1(x)$, $u''\notin T_1(x)$ and $u'$ is the parent of $u''$. Similarly, we can replace the edge $(u',u'')$ with edge $(u_x,u'')$, obtaining another \St with strictly lower cost, a contradiction to \ref{prop_1}. Therefore, $T_1(x)\subseteq V(W_1)\subseteq S_1(x)$. The proof of
$T_2(x)\subseteq V(W_2)\subseteq S_2(x)$ is symmetric.

Last, we prove property (iii) in the observation. Let $\tau'$ be the tree defined above.
Consider a Steiner vertex $u_x$ in $\tau'$. We denote by $u_{x'}$ the parent of $u_x$, and denote by $u_{\hat x}$ the other child of $u_{x'}$.
Assume for contradiction that both $u_x$ and $u_{x'}$ belong to $\tau'$. From property (i), $u_{x'}$ has a neighbor in $T(x)$, that we denote by $u$. From property (ii), since $u\in T(x)$, vertex $u$ belongs to either $W_1$ or $W_2$. However, $u_{x'}\notin R(x)$, so $u$ does not belong to either $W_1$ or $W_2$, a contradiction to the fact that $u_{x'}$ and $u$ are connected by an edge.
Assume for contradiction that both $u_x$ and $u_{\hat x}$ belong to tree $\tau'$. Let $u$ be the neighbor of $u_x$ that belongs to the same connected component of $\tau'\setminus\set{u_x}$ as $u_{\hat x}$. From property (i) and (ii) on $u_{\hat x}$, $u$ does not belong to the corresponding subgraphs $\hat W_1,\hat W_2$ for $u_{\hat x}$. Let $\hat u$ be any leaf of $T(\hat x)$. Via similar arguments, we can replace the edge $(u_x,u)$ with the edge $(u_x,\hat u)$, obtaining another \St with strictly lower cost, a contradiction to \ref{prop_1}.
Assume for contradiction that none of $u_x,u_{x'}, u_{\hat x}$ belongs to $\tau'$. Via similar arguments, it is easy to verify that we can add either $u_x$ or $u_{\hat x}$ to $\tau'$ and deleting some edges, obtaining another \St with strictly lower cost, a contradiction to \ref{prop_1}.

\section{Analysis of the Algorithm in \Cref{subsec: >2_upper}}

\label{apd: Analysis of >2_upper}

In this section, we show that the spanning tree $\tau$ output by the algorithm in \Cref{subsec: >2_upper} is with high probability an $\alpha$-approximate Steiner Tree.
We denote by $\mst$ the minimum spanning tree cost on $T$.
Therefore, it suffices to show that, with high probability, 
$w(\tau)\le (\alpha/2)\cdot \mst$, as $\mst$ is at most twice the minimum Steiner Tree cost. Recall that $\beta=\alpha/(100\log n)$.

Let $\tau^*$ be a minimum spanning tree on $T$. Let $\pi=(u_1,u_2,\ldots,u_{2k-2})$ be an Euler-tour of $\tau^*$, and for each $1\le t\le 2k-2$, we let $R_{\pi,t}=\set{u_i\mid t\le i\le t+20\beta\log n}$.
We define a bad event $\xi$ as follows.

\paragraph{Bad event $\xi$.}
Let $\xi$ be the event that there exists some $t$, such that $R_{\pi,t}\cap T'=\emptyset$.
We now show that $\Pr[\xi]=O(n^{-9})$.
Since each edge of $\tau^*$ appears at most twice in the set $\set{(u_i,u_{i+1})\mid u_i\in R_{\pi,t}}$, $R_{\pi,t}$ contains at least $10\beta\log n$ distinct vertices. Therefore, the probability that a random subset of $\ceil{n/\beta}$ vertices in $V$ does not intersect with $R_{\pi,t}$ is at most $(1-(10\beta\log n/k))^{k/\beta}\le n^{-10}$. Taking the union bound over all $1\le t\le 2k-2$, we get that $\Pr[\xi_1]=O(n^{-9})$. 
Note that, if the event $\xi$ does not happen, then every consecutive window of $\pi$ of length $20\beta\log n$ contains at least one element of $T'$.

Let $u_{i_1},u_{i_2},\ldots,u_{i_{t'}}$ be the vertices of $\pi$ that belongs to $T'$ ($t'$ may be larger than $|T'|$ since we keep all copies of the same vertex), then for each $1\le j\le t'$, $|i_j-i_{j+1}|\le 20\beta\log n$.

Let $u$ be any terminal in $T\setminus T'$. Assume that the first appearance of $u$ in $\pi$ is between $u_{i_j}$ and $u_{i_{j+1}}$, so $w(u,T')\le \min\set{w(u,u_{i_j}),w(u,u_{i_{j+1}})}\le w(u_{i_j},u_{i_{j+1}})$, which is at most the total weight of all edges in $\pi$ between $u_{i_j}$ and $u_{i_{j+1}}$. So from triangle inequality, $\dist_{\pi}(u,T')\le \sum_{i_j\le t\le i_{j+1}-1}w(v_t,v_{t+1})$.
As $|i_j-i_{j+1}|\le 20\beta\log n$ holds for all $1\le j\le t'$, $\sum_{u\in T\setminus T'}w(u,f(u))\le 20\beta\log n\cdot w(\tau^*)\le 20\beta\log n\cdot \mst$.
Finally, since $w(\tau')\le \mst$, we conclude that $w(\tau)\le 21\beta\log n\cdot\mst\le (\alpha/2)\cdot \mst$.
Altogether, we conclude that, with probability $1-O(n^{-9})$,
the spanning tree $\tau$ output by the algorithm in \Cref{subsec: >2_upper} is an $\alpha$-approximate Steiner Tree.

\section*{Acknowledgements}

We thank anonymous reviwers for helpful comments and for pointing to us the previous work \cite{du1995component} and \cite{zelikovsky199311}.
We thank Mohammad Roghani, Sepideh Mahabadi, Ali Vakilian, and Jakub Tarnawski for pointing to us an inaccuracy of the previous version of this paper.

\bibliography{REF}

\end{document}